%% file: main.tex
\documentclass[11pt]{article}

\usepackage[letterpaper,margin=1.0in]{geometry}
\usepackage{amsmath, amssymb, amsthm, amsfonts}
\usepackage{thmtools,thm-restate}
\usepackage{bbm}

\usepackage{ifthen}

\usepackage{cite}
\usepackage{appendix}
\usepackage{graphicx}
\usepackage{xcolor}
\usepackage{algorithm}
\usepackage[noend]{algpseudocode}

\usepackage{dsfont}
\newcommand{\1}{\mathds{1}}

\usepackage[textsize=tiny]{todonotes}

\usepackage{framed}
\usepackage[framemethod=tikz]{mdframed}
\usepackage[bottom]{footmisc}
\usepackage{enumitem}
\setitemize{noitemsep,topsep=3pt,parsep=3pt,partopsep=3pt}
\usepackage[font=small]{caption}
\usepackage{xspace}

\definecolor{darkgreen}{rgb}{0,0.5,0}
\definecolor{darkblue}{rgb}{0,0,0.6}
\usepackage{hyperref}
\hypersetup{
    unicode=false,          
    colorlinks=true,        
    linkcolor=darkblue,          
    citecolor=darkgreen,        
    filecolor=magenta,      
    urlcolor=cyan           
}
\usepackage[capitalize, nameinlink]{cleveref}
\Crefname{theorem}{Theorem}{Theorems}
\Crefname{lemma}{Lemma}{Lemmas}
\Crefname{claim}{Claim}{Claims}
\Crefname{remark}{Remark}{Remarks}
\Crefname{observation}{Observation}{Observations}

\newtheorem{theorem}{Theorem}[section]
\newtheorem{lemma}[theorem]{Lemma}
\newtheorem{meta-theorem}[theorem]{Meta-Theorem}
\newtheorem{claim}[theorem]{Claim}
\newtheorem{remark}[theorem]{Remark}

\newtheorem{definition}{Definition}[section]

\algnewcommand\algorithmicswitch{\textbf{switch}}
\algnewcommand\algorithmiccase{\textbf{case}}

\algdef{SE}[SWITCH]{Switch}{EndSwitch}[1]{\algorithmicswitch\ #1\ \algorithmicdo}{\algorithmicend\ \algorithmicswitch}%
\algdef{SE}[CASE]{Case}{EndCase}[1]{\algorithmiccase\ #1}{\algorithmicend\ \algorithmiccase}%
\algtext*{EndSwitch}%
\algtext*{EndCase}%

\newcommand{\eps}{\varepsilon}

\newcommand{\CONGEST}{$\mathsf{CONGEST}$\xspace}
\newcommand{\LOCAL}{$\mathsf{LOCAL}$\xspace}
\newcommand{\OPT}{\mathsf{OPT}}
\newcommand{\local}{\LOCAL}
\newcommand{\congest}{\CONGEST}

\newcommand{\poly}{\operatorname{\text{{\rm poly}}}}
\newcommand{\quasipoly}{\operatorname{\text{{\rm quasipoly}}}}

\newcommand{\set}[1]{\left\{#1\right\}}

\renewcommand{\tilde}{\widetilde}

\DeclareMathOperator{\E}{\mathbb{E}}
\newcommand{\Labels}{\Sigma}
\newcommand{\R}{\mathbb{R}}
\newcommand{\Rp}{\mathbb{R}_{\geq 0}}
\newcommand{\N}{\mathbb{N}}
\newcommand{\calL}{\mathcal{L}}

\newcommand{\utility}{\mathbf{u}}
\newcommand{\cost}{\mathbf{c}}

\newcommand{\fE}{\mathcal{E}}

\newcommand{\tx}{\tilde{x}}

\newcommand{\FullOrShort}{full}

\ifthenelse{\equal{\FullOrShort}{full}}{
	
  \newcommand{\fullOnly}[1]{#1}
  \newcommand{\shortOnly}[1]{}

  }{

    \newcommand{\fullOnly}[1]{}
    \newcommand{\IncludePictures}[1]{}
   
  }

\begin{document}
\date{}
\title{Local Distributed Rounding: \\ Generalized to MIS, Matching, Set Cover, and Beyond}
 \author{
   Salwa Faour \\
   \small{University of Freiburg} \\
   \small{salwa.faour@cs.uni-freiburg.de}
   \and
   Mohsen Ghaffari \\
   \small{MIT}\\
   \small{ghaffari@mit.edu}
   \and
   Christoph Grunau \\
   \small{ETH Zurich}\\
   \small{cgrunau@inf.ethz.ch}
   \and
   Fabian Kuhn \\
   \small{University of Freiburg} \\
   \small{kuhn@cs.uni-freiburg.de}
   \and
   Václav Rozhoň \\
 \small{ETH Zurich} \\
 \small{rozhonv@inf.ethz.ch}
 }
\maketitle

\begin{abstract} 
We develop a general deterministic distributed method for locally rounding fractional solutions of graph problems for which the analysis can be broken down into analyzing pairs of vertices. Roughly speaking, the method can transform fractional/probabilistic label assignments of the vertices into integral/deterministic label assignments for the vertices, while approximately preserving a potential function that is a linear combination of functions, each of which depends on at most two vertices (subject to some conditions usually satisfied in pairwise analyses). The method unifies and significantly generalizes prior work on deterministic local rounding techniques [Ghaffari, Kuhn FOCS'21; Harris FOCS'19; Fischer, Ghaffari, Kuhn FOCS'17; Fischer DISC'17] to obtain polylogarithmic-time deterministic distributed solutions for combinatorial graph problems. Our general rounding result enables us to locally and efficiently derandomize a range of distributed algorithms for local graph problems, including maximal independent set (MIS), maximum-weight independent set approximation, and minimum-cost set cover approximation. As highlights, we in particular obtain the following results.
\begin{itemize}
    \item We obtain a deterministic $O(\log^2\Delta\cdot\log n)$-round algorithm for computing an MIS in the \LOCAL model and an almost as efficient $O(\log^2\Delta\cdot\log\log\Delta\cdot\log n)$-round deterministic MIS algorithm in the \CONGEST model. As a result, the best known deterministic distributed time complexity of the four most widely studied distributed symmetry breaking problems (MIS, maximal matching, $(\Delta+1)$-vertex coloring, and $(2\Delta-1)$-edge coloring) is now $O(\log^2\Delta\cdot\log n)$. Our new MIS algorithm is also the first direct polylogarithmic-time deterministic distributed MIS algorithm, which is not based on network decomposition.
    \item We obtain efficient deterministic distributed algorithms for rounding fractional solutions for maximum (weighted) independent set and minimum (weighted) set cover. We in particular give a deterministic $O(\log^2\Delta + \log^* n)$-round algorithms for computing an independent set of size $(1/2-\varepsilon)\cdot n/\deg_{\mathrm{avg}}$ and we give deterministic $O(\log^2(\Delta W) + \log^* n)$-round algorithms for computing a $(1-\varepsilon)/\Delta$-approximation of maximum weight independent set, and for computing a $(1-\varepsilon)/r$-approximation of maximum weight matching in hypergraphs of rank $r$. For minimum set cover instances with sets of size at most $s$ and where each element is contained in at most $t$ sets, we show that an $O(\log s)$-approximation can be computed in time $O(\log s \cdot\log^2 t + \log^* n)$.
\end{itemize}
\end{abstract}
\setcounter{page}{0}
\thispagestyle{empty}


{   \newpage
    \smallskip
    \hypersetup{linkcolor=blue}
    \tableofcontents
    \setcounter{page}{0}
    \thispagestyle{empty}
}
\newpage

\input{introduction}

\input{genericrounding}

\input{MIS}

\input{weightedIS}

\input{setCover}

\section*{Acknowledgments}
M.G., C.G. and V.R. were supported in part by the European Research Council (ERC) under the European Unions Horizon 2020 research and innovation program (grant agreement No.~853109) and the Swiss National Science Foundation (project grant 200021\_184735).

\bibliographystyle{alpha}
\bibliography{ref}

\end{document}

%% file: introduction.tex
\section{Introduction \& Related Work}
In this paper, we present a \emph{local rounding method}
that can be used to transform probabilistic/fractional assignments in a certain class of randomized distributed graph algorithms into deterministic/integral assignments. Roughly speaking, the method applies to randomized distributed algorithms that can be analyzed assuming pairwise independence. This allows us to transform such randomized algorithms into deterministic distributed algorithms. This unified approach leads to novel and/or improved algorithms for a number of the central problems of the area---e.g., maximal independent set, $(\Delta+1)$-coloring, maximal (hypergraph) matching, and set cover. Several of these problems had remained open for over three decades until a recent breakthrough on network decomposition~\cite{rozhonghaffari20}. Our approach is completely independent and yields faster algorithms (besides providing a new and more systematic construction for the network decomposition problem itself).

\subsection{Background and State of the Art}
\noindent\textbf{Model.} We work with the standard synchronous message-passing model of distributed computing~\cite{peleg00}. The network is abstracted as an undirected graph $G=(V, E)$, with $n:=|V|$. Initially, each node knows only its own unique $O(\log n)$-bit identifier, and nothing else about the graph $G$, except for potentially some upper bounds on global parameters such as $n$, the number of nodes, and $\Delta$, the maximum degree. Per round, each node can send one $B$-bit message to each of its neighbors. We usually assume $B=O(\log n)$ and refer to this model variant as \CONGEST. The variant that allows unbounded message sizes is referred to as \LOCAL. At the end of the computation, each node should know its own part of the output, e.g., whether it is in the computed independent set.

\paragraph{Linial's MIS question.}
Four of the best studied problems in distributed graph algorithms are maximal independent set (MIS), maximal matching, $(\Delta+1)$-vertex coloring, and $(2\Delta-1)$-edge coloring. MIS is the hardest of these four, as all others can be reduced to it~\cite{luby86,linial1987LOCAL}. Luby's classic algorithm from 1986 provides an $O(\log n)$-round randomized algorithm for MIS, and thus for all the others~\cite{luby86} (see also \cite{alon86}). It however remained open for over three decades whether there also is a deterministic $\poly\log n$-time algorithm for MIS, and this came to be known as \emph{Linial's open question}, first raised in~\cite{linial1987LOCAL, linial92}. See also the book of Barenboim and Elkin~\cite{barenboimelkin_book} for further discussions on the significance of this and other related problems.

\paragraph{The first solution, via network decomposition. } The first solution to the MIS question was provided recently by Rozho\v{n} and Ghaffari~\cite{rozhonghaffari20}. Their algorithm actually solves the network decomposition problem in $\poly\log n$ time, which roughly speaking partitions the vertices into $\poly\log n$ colors, such that the components in the subgraph induced by each color have $\poly\log n$ diameter. This directly leads to a $\poly\log n$-time deterministic algorithm for MIS in the \LOCAL model. In the \CONGEST model, a more complex $\poly\log n$-time MIS algorithm follows from combining network decomposition with a  deterministic \CONGEST algorithm of \cite{censor2017derandomizing} for low-diameter networks. The combination works by leveraging the pairwise analysis of the randomized MIS algorithms (Luby~\cite{luby1993removing} or Ghaffari~\cite{gmis}) and using the low-diameter components of the decomposition to fix the bits of the randomized algorithm one by one. Currently, the fastest known deterministic MIS algorithm both in \LOCAL and \CONGEST runs in $O(\log^5 n)$ rounds, and it is based on an improved variant of the Rozho\v{n}-Ghaffari decomposition presented in \cite{GGR20}.

As a side note, we remark that the decomposition result, in combination with \cite{ghaffari2017complexity, ghaffari2018derandomizing}, yields a \LOCAL-model derandomization with $\poly\log n$ round slow down, for $\poly\log n$-time checkable problems. However, there is no such general derandomization for the \CONGEST model.

\paragraph{A different line of attack, via rounding.}
A more direct line of attack toward developing $\poly\log n$-time deterministic algorithms for the four problems mentioned above is via \emph{local rounding}. Roughly speaking, this approach starts with a suitable fractional solution to a certain problem and gradually and deterministically rounds this fractional solution to obtain an integral solution with similar quality. One can trace this local rounding approach back to the work of Hanckowiak, Karonski, and Panconesi~\cite{hanckowiak01} who gave the first $\poly\log n$-time deterministic distributed algorithm for the maximal matching problem; though they never talked about rounding. The \emph{local rounding} nature of the approach was made explicit and used to improve the complexity for maximal matching in the work of Fischer\cite{fischer2020improved}. Fischer observed that for the matching problem, it is easy to obtain an $O(1)$-approximation of maximum fractional matching in a deterministic way, and that the real challenge is in rounding this fractional matching into an integral one without much loss in the size.  Fischer developed a specific $O(\log^2 \Delta)$-time rounding procedure for the matching problem (in bipartite graphs). Repeating this $O(\log n)$ times (in a suitable bipartite version of the graph) adds in each repetition a matching within a constant factor of the remaining maximum, hence reducing the maximum remaining matching by a constant factor, and therefore yields a maximal matching in $O(\log^2 \Delta \cdot \log n)$ time. That remains the state of the art for the maximal matching problem. It however should be noted that Fischer's rounding was very specific to matching in graphs, and it did not generalize to the other three problems. In fact, the method did not extend even to matching in hypergraphs of rank three, i.e., where each edge has three endpoints. 

A follow-up work by Fischer, Ghaffari, and Kuhn\cite{FischerGK17} developed a different rounding method for matching that extended to low-rank hypergraphs. Along with a reduction that they gave from $(2\Delta-1)$-edge coloring in graphs to maximal matching in hypergraphs of rank $3$, this led to a $\poly\log n$-time deterministic algorithm for $(2\Delta-1)$-edge coloring. This result put the second of the four problems in the $\poly\log n$-time deterministic regime. The exact complexity was notably higher. Harris\cite{harris2019distributed} improved the complexity of this hypergraph matching rounding. Among other speed-ups, that led the complexity of deterministic $(2\Delta-1)$-edge coloring to reach $\tilde{O}(\log^2 \Delta \cdot \log n)$, which nearly matches that of maximal matching.  

The above lines of work on local rounding appeared limited to rounding for matching in graphs or hypergraphs. Then (and shortly after the network decomposition result~\cite{rozhonghaffari20}), Ghaffari and Kuhn~\cite{GhaffariK21} developed an efficient rounding method for $\Delta+1$ coloring. Roughly speaking, their approach starts with a simple fractional assignment of the colors (where each node takes an equal portion of each of the colors in its palette) and gradually rounds this assignment until reaching (a potentially improper) integral color assignments. This is done in a manner that, the rounding approximately maintains a certain potential function, such that in the end, once given the final integral color assignments, we can efficiently turn it into a proper coloring of a constant fraction of the remaining node. Interestingly enough, this rounding also led to the $O(\log^2 \Delta \cdot \log n)$ time complexity for vertex coloring. This put the third of the four problems in the $\poly\log\Delta\cdot \log n$-time, and concretely $O(\log^2 \Delta \cdot \log n)$-time, deterministic regime. We note that, at that time, $\poly\log n$-round complexity coloring was already known from the decomposition result~\cite{rozhonghaffari20}, but the rounding-based method~\cite{GhaffariK21} gave a more direct and faster coloring algorithm.

Despite this exciting progress on local rounding, the known methods appeared ad hoc, specifically tailored to matching or coloring problems. For instance, they did not extend to the hardest of the four classic problems, the maximal independent set problem. 

\subsection{Our Contributions}
\label{sec:contributions}

In this paper, we vastly generalize the local rounding method, in such a way that, roughly speaking, we can now derandomize randomized local algorithms that can be analyzed with pairwise independence. This rounding works in a local and efficient manner, and without relying on network decompositions. We show that this generalized local rounding method yields new and/or improved deterministic distributed algorithms for a range of graph problems of interest. As a notable example, we obtain an $O(\log^2 \Delta \cdot \log n)$ round \LOCAL model algorithm for MIS, and a  $O(\log^2 \Delta \cdot \log\log \Delta \cdot \log n)$-round \CONGEST model variant. 

\begin{theorem} There is a deterministic distributed algorithm that computes an MIS in time $O(\log^2 \Delta \cdot \log n)$ in the \LOCAL model. A variant of this algorithm computes an MIS deterministically in $O(\log^2 \Delta\cdot  \log\log \Delta \cdot \log n)$ rounds of the \CONGEST model.
\end{theorem}
Hence, now all the four classic problems are in the $O(\log^2 \Delta \cdot \log n)$ deterministic round complexity regime, and in a unified way. This result provides a second solution to Linial's famous open question\cite{linial1987LOCAL}, which had remained open for over three decades. The new solution is completely independent of the first decomposition-based solution and is also more efficient.  In contrast, the fastest previously known MIS algorithm, in either \LOCAL or \CONGEST models, required $O(\log^5 n)$ rounds\cite{GGR20} and was based on network decomposition. 

We also note for maximal independent set and maximal matching, a celebrated recent result of Balliu et al.~\cite{balliu2019LB} gives a lower bound of $\Omega(\min\{\Delta, \log n/\log\log n\})$ for deterministic algorithms in the \LOCAL model. This almost justifies the need for one $\log n$ factor in the upper bound. In particular, in the regime where $\Delta=\poly(\log n)$, the lower bound becomes  $\Omega(\log n/\log\log n)$ and our \LOCAL-model upper bound is $O(\log n\cdot \log^2\log n)$, which means that the upper and lower bounds are matching up to exponentially lower order factors.

\paragraph{Other applications, set cover.} As another prominent application, we obtain an improved approximation algorithm for the minimum set cover problem, with an approximation factor roughly matching that of the best-known centralized algorithm (below which the approximation problem becomes $\mathsf{NP}$-hard):

\begin{theorem} There is a deterministic distributed algorithm that computes an $O(\log s)$ approximation of the minimum set cover problem in $\tilde{O}(\log s \log^2 t+ \log^* n)$ rounds of the \CONGEST model, where $s$ denotes the maximum set cardinality, and $t$ denotes the maximum number of sets that contain a given element.
\end{theorem}

Notably, the complexity of this algorithm depends logarithmically only on the local parameters of the problem (degrees in a bipartite formulation) and has only an additive $O(\log^* n)$ dependency on the network size, which is known to be necessary~\cite{lenzen08}. In contrast, the best previously known algorithm has a polylogarithmic dependency on the global network size $n$~\cite{Deurer2019, GGR20}. This is an important qualitative difference and has implications for other computational settings\footnote{In particular, in the \emph{local computations algorithms} (LCA) model of sublinear-time centralized computation \cite{rubinfeld2011LCA, alon2012LCA}, our result provides a deterministic set cover approximation with query complexity $d^{\tilde{O}(\log^3 d)} \poly(\log n)$, where $d=st$. This follows from a direct stimulation via gathering the relevant local neighborhood\cite{parnas2007approximating} (with an appropriate prior coloring~\cite{chang2019complexity}). This LCA can answer whether each asked set is in the selected cover or not, using only $d^{\tilde{O}(\log^3 d)} \poly(\log n)$ query accesses to the adjacency lists of the graph.  In contrast, a simulation of the decomposition-based approach would imply query complexity $d^{\poly(\log n)}$, which is even super-polynomial in $n$ and thus completely useless/uninteresting. The best known LCA for set cover has query complexity $d^{O(\log d \log\log d)} \poly(\log n)$~\cite{grunau2020improved}, but that heavily relies on randomness. No deterministic LCA with query complexity $\quasipoly(d) \cdot \poly(\log n)$ was known previously.}.

We also note that Deuer, Kuhn, and Maus~\cite{Deurer2019} previously used a deterministic rounding method to obtain an $O(\log \Delta)$ approximation for the minimum dominating set problem in $O(\Delta \poly(\log \Delta) + \poly(\log \Delta) \log^* n)$ rounds of the \CONGEST model, in any graph with maximum degree at most $\Delta$. The minimum dominating set problem is  equivalent to the set cover problem in the setting where $s=t=\Delta$, and one can provide a simple approximation preserving reduction in both directions. Using the straightforward direction of this connection (reducing dominating set to set cover by viewing each node as a set that includes all of its neighbors) and invoking our minimum set cover approximation algorithm, we obtain an $O(\log \Delta)$ approximation for minimum dominating set in $\tilde{O}(\log^3 \Delta +  \log^* n)$ rounds of the \CONGEST model. This improves on the $O(\Delta \poly(\log \Delta) + \poly(\log \Delta) \log^* n)$ complexity of ~\cite{Deurer2019} nearly exponentially. 

\paragraph{Other applications, network decomposition.} In fact, since the network decomposition problem has been known to be reducible to the set cover problem\cite{ghaffari2018derandomizing}, our method yields a novel $\poly(\log n)$-round algorithm for the network decomposition problem. One may view this as a more systematic solution for network decomposition, in contrast to the specific combinatorial approach of \cite{rozhonghaffari20}. This new construction reduces the decomposition problem to its pairwise analyzable core (captured in set cover) and then derandomizes the corresponding natural randomized algorithm via local rounding. This is more systematic in the sense that it is similar to other problems that are derandomized by going through pairwise analysis. 

\paragraph{Other applications, maximum independent set and hypergraph matching.} As another applications, we get efficient algorithms for computing large cardinality or weight independent sets. In the following, the neighborhood independence $\beta$ of a graph $G=(V,E)$ is the size of the largest independent set of any subgraph $G[N(v)]$ induced by the set of neighbors $N(v)$ of some node $v\in V$. Further, for a node weight function $w:V\to\N$ and a node set $S\subseteq V$, $w(S)=\sum_{v\in S} w(v)$.

\begin{theorem}\label{thm:ISthm}
  Let $G=(V,E)$ be a node-weighted $n$-graph of maximum degree $\Delta$, neighborhood independence $\beta$, node weights $w:V\to \N$, and maximum weight $W$. Further, let $\OPT$ be the weight of a maximum weight independent set and assume that $G$ is equipped with an proper $\xi$-coloring. Then, for every $\eps>0$, there are deterministic \CONGEST algorithms to compute independent sets
  \begin{align}
    \text{of weight } & \frac{1-\eps}{\beta}\cdot\OPT & \text{ in } & O\big(\log^2(\Delta W)\cdot\log(1/\eps) + \log^* \xi\big)\text{ rounds,}\label{eq:ISthm1}\\
    \text{of weight } & (1-\eps)\cdot \frac{w(V)}{\Delta+1} & \text{ in } & O\big(\log^2\Delta\cdot\log(1/\eps) + \log^* \xi\big)\text{ rounds, and}\label{eq:ISthm2}\\
    \text{of weight } & \left(\frac{1}{2}-\eps\right)\cdot \sum_{v\in V}\frac{(w(v))^2}{w(v) + w(N(v))} & \text{ in } & O\left(\frac{\log^2(\Delta/\eps)}{\eps} + \log^* \xi\right)\text{ rounds.}\label{eq:ISthm3}
  \end{align}
\end{theorem}

{\it \noindent\Cref{thm:ISthm} has several implications:}
\begin{itemize}
    \item Since $\beta\leq \Delta$, result \eqref{eq:ISthm1} implies a $(1-\eps)/\Delta$-approximation in the same round complexity. For small $\Delta$ (and $W$), this is a significant improvement over an $O(T_{\mathrm{MIS}}/\eps)$-round algorithm for this problem in \cite{KawarabayashiKS20} ($T_{\mathrm{MIS}}$ denotes the time for computing an MIS).
    \item Line graphs of hypergraphs of rank $r$ have neighborhood independence $r$. Thus, result \eqref{eq:ISthm1} implies the same result for computing a $(1-\eps)/r$-approximation for maximum weight matching in hypergraphs of rank $r$. For small $W$, this is an improvement over an algorithm  in \cite{harris2019distributed}.
    \item Result \eqref{eq:ISthm1} also implies an $O(\log^2\Delta + \log^*\xi)$-round algorithm for computing a constant maximum matching approximation in graphs. Repeating $O(\log n)$ times yields an $O(\log^2\Delta\cdot\log n)$-round maximal matching algorithm and in \Cref{sec:boundedindep} (\Cref{thm:maximalmatching}), we show that this also works in the \CONGEST model. This gives an alternative to the deterministic $O(\log^2\Delta\cdot \log n)$-round \CONGEST algorithm of Fischer~\cite{fischer2020improved} and it implies that our generic rounding framework implies the current best deterministic \LOCAL and \CONGEST algorithms for the four classic distributed symmetry breaking problems.
    \item In \cite{KawarabayashiKS20}, Kawarabayashi, Khoury, Schild, and Schwartzman showed that with randomization, an independent set of size $\Omega(n/\Delta)$ can be computed exponentially faster than an MIS and they raised the question whether the same is also true for deterministic algorithms. Result \eqref{eq:ISthm2} answers this question in the affirmative, because for deterministically computing an MIS, there is an $\Omega\big(\min\set{\Delta,\frac{\log n}{\log\log n}}\big)$-round lower bound even in the \LOCAL model~\cite{balliu2019LB}.
    \item By a simple application of the Cauchy-Schwarz inequality, the independent set weight in \eqref{eq:ISthm3} can be lower bounded by $\big(\frac{1}{2}-\eps\big)\cdot\frac{(w(V))^2}{w(V) + \sum_{v\in V} w(N(v))}$ (for details, see \cite{kako2005approximation}). Without the factor $\frac{1}{2}-\eps$, the two bounds are natural weighted generalizations of the well-known lower bound $\sum_{v\in V}\frac{1}{\deg(v)+1}\geq \frac{n}{\deg_{\mathrm{avg}}+1}$ on the size of a maximum cardinality independent set. The bound is sometimes also known as the Caro-Wei-Tur\'an bound~\cite{turan1941external,Wei1981generalizedturan,griggs1983lower}.
\end{itemize}

\subsection{Overview of Our Method}

In the following, we give an extended high-level description of our general rounding method and we also discuss the main novel ideas that are necessary to obtain the results given in \Cref{sec:contributions}.

\subsubsection[The Rounding of Ghaffari and Kuhn\cite{GhaffariK21} for $(\Delta+1)$-Coloring]{The Rounding of Ghaffari and Kuhn\cite{GhaffariK21} for \boldmath$(\Delta+1)$-Coloring} 

Our generic rounding algorithm is a generalization of a recent deterministic, distributed algorithm by Ghaffari and Kuhn~\cite{GhaffariK21} for solving the $(\Delta+1)$-coloring problem. We therefore first review (a slightly adapted version of) the algorithm of \cite{GhaffariK21}. At the core of the algorithm of \cite{GhaffariK21} is a method to color a constant fraction of the nodes of the graph in time $O(\log^2\Delta)$. Repeating this $O(\log n)$ times then colors the whole graph in time $O(\log^2\Delta\cdot\log n)$.\footnote{Formally, this requires the method for coloring a constant fraction of the nodes to work for the more general $(\mathit{degree}+1)$-list coloring problem. We however ignore this in this high-level overview.} The algorithm for coloring a constant fraction of the nodes is based on derandomizing the following trivial randomized algorithm for the same problem: Every node $v\in V$ chooses one of the $\Delta+1$ colors uniformly at random and $v$ keeps the color if no higher-ID neighbor chooses the same color. When doing this, in expectation, a constant fraction of the nodes can keep their color. One way to see this is by the following analysis. For every edge $\set{u,v}$ of the graph, the probability that both nodes $u$ and $v$ choose the same color is $1/(\Delta+1)$. Hence, the total expected number of such monochromatic edges is at most $|E|/(\Delta+1)<n/2$. For each such edge, we uncolor the lower ID node. The expected number of nodes that do not keep their color is thus less than $n/2$.

\paragraph{Rounding a Fractional Solution.} In \cite{GhaffariK21}, the above step is derandomized by considering the assignment of a uniformly random color to each node as a fractional assignment of colors to nodes, and picking a single color per node is then considered as rounding this fractional color assignment to an integral color assignment. To do this, the algorithm \cite{GhaffariK21} assigns a cost to each edge, which is equal to the probability of both endpoints picking the same color when independently choosing the color at random from the current fractional color assignment. We can then define the total cost of a given fractional color assignment as the sum over the individual edge costs. Note that by the linearity of expectation, the total cost is equal to the expected number of monochromatic edges if each node picks its color randomly according to its current fractional distribution over colors. As discussed above for the initial fractional assignment (i.e., if each node has a uniform distribution over all $\Delta+1$ colors), the total cost is at most $n/2$. The goal now is to find an integral assignment of colors for which the total cost (i.e., the total number of monochromatic edges) is not much larger (say, at most $3n/4$) such that still, a constant fraction of the nodes can keep their colors. First note that if we sequentially iterate over the nodes, there is a simple way to achieve this. When rounding the fractional assignment of a node $v$, $v$ just picks the color that minimizes the sum of the costs of its edges (i.e., the color that minimizes the total number of monochromatic edges of $v$). In this way, we can obtain an integral color assignment for which the total cost is at most the total cost of the initial fractional color assignment and thus at most $n/2$. Because nodes that are non-adjacent do not share edges, they can do this rounding in parallel. If we are given a proper vertex coloring with $C$ colors, the rounding can therefore be done in $O(C)$ rounds in the distributed setting.

\paragraph{Fast Rounding of a Fractional Solution.} 
While the described rounding algorithm can perfectly preserve the cost of a given fractional color assignment, having to iterate through the colors of a proper vertex coloring of the graph will be extremely slow. In this way, we can at best hope for an algorithm with a round complexity that is linear in $\Delta$, which is exponentially slower than what we aim to achieve. In order to obtain a faster rounding algorithm, \cite{GhaffariK21} applies two main ideas. First, instead of iterating over the colors of a proper vertex coloring, the algorithm first computes a defective $C$-coloring and then iterates over the $C$ colors of this coloring. In a defective coloring, each node is allowed to have some neighbors of the same color. By extending ideas of \cite{Kuhn2009WeakColoring,BEG18,KawarabayashiS18}, it is shown in \cite{GhaffariK21} that on a graph $G=(V,E)$ with edge weights $w(e)\geq0$, for a given parameter $\delta>0$, one can compute a vertex coloring with only $O(1/\delta)$ colors such that the total weight of the monochromatic edges is at most a $\delta$-fraction of the total weight of all edges. We call such a coloring a \emph{weighted $\delta$-relative average defective coloring}. If one starts with an initial proper $O(\Delta^2)$-coloring of $G$ (which can be computed in time $O(\log^* n)$ by using an algorithm of \cite{linial92}), a weighted $\delta$-relative average defective $O(1/\delta)$-coloring can be computed deterministically in time $O(1/\delta + \log^*\Delta)$. The idea now is to compute a weighted $\delta$-relative average defective $O(1/\delta)$-coloring where the edge weights are equal to the edge costs before rounding and to run the above simple distributed rounding algorithm on the graph induced by the bichromatic edges (w.r.t.\ the defective coloring). For this, we have to iterate over $O(1/\delta)$ colors, which even when choosing $\delta=1/\poly\log \Delta$ is much faster than iterating over $\Omega(\Delta)$ colors. The problem with the approach is that the cost of monochromatic edges (w.r.t.\ the defective coloring) can now grow arbitrarily. Those edges are not considered during rounding and their cost can increase from initially $1/(\Delta+1)$ to $1$ if both nodes of an edge choose the same color. In order to avoid such drastic cost increases of the monochromatic edges, we have to do the rounding in several small steps. If in each iteration, each node at most doubles the fractional value for each color, the probability of an edge becoming monochromatic can increase by at most a factor $4$, even if the nodes are rounded in a worst-case way. As a consequence, in such a rounding step, the total cost over all edges can grow by at most a factor of $1+O(\delta)$. In the algorithm of \cite{GhaffariK21}, the step-wise rounding is done in such a way that with every iteration, the minimal non-zero fractional color value is at least doubled. One can then get from fractional values $1/O(\Delta)$ to integral values in $O(\log\Delta)$ steps. By choosing $\delta=c/\log\Delta$ for a sufficiently small constant $c>0$, the total cost can grow by at most a constant factor that can be made arbitrarily close to $1$. Since in each rounding step, we have to compute a weighted $\delta$-relative average defective coloring with $O(1/\delta)$ colors and iterate over $O(1/\delta)=O(\log\Delta)$ colors of this defective coloring, each rounding step requires $O(\log\Delta)$ rounds. We can therefore obtain a coloring with at most $3n/4$ monochromatic edges in $O(\log^2\Delta)$ rounds and we can thus properly color at least $n/4$ nodes in $O(\log^2\Delta)$ rounds.

\subsubsection{Extending the Algorithm to Obtain Large Independent Sets}

In the present paper, we generalize the rounding method of \cite{GhaffariK21} in a significant way to make it applicable to a much wider family of problems. Our simplest application of the more general rounding algorithm is a deterministic algorithm for obtaining large (or heavy in the case of node-weighted graphs) independent sets. As this algorithm already allows to highlight some of the key challenges and ideas, we discuss it here first. 

For this high-level discussion, assume that we want to compute an independent set of size $\Omega(n/\Delta)$. There is a simple randomized algorithm to achieve this (in expectation). First, each node marks independently itself with probability $1/\Delta$ and afterward, for every edge $\set{u,v}$, if both $u$ and $v$ are marked, then we unmark the lower-ID node. The set of nodes that are still marked then clearly form an independent set of $G$. The expected number of marked nodes is $n/\Delta$ and for every edge, both nodes are marked with probability $1/\Delta^2$. The expected number of edges for which both nodes are marked is therefore at most $|E|/\Delta^2 \leq n/(2\Delta)$. Therefore, in expectation, we obtain an independent set of size at least $n/(2\Delta)$. In order to obtain a deterministic variant of this algorithm, we can try to adapt the idea of \cite{GhaffariK21}. A fractional independent set is an assignment of fractional values $x_v\in [0,1]$ to each node and we can define a potential function $\Phi(\vec{x})$ that measures the expected size of the resulting independent set when running the randomized algorithm described above, i.e., if each node $v$ gets marked independently with probability $x_v$ and if for every edge $\set{u,v}$, if both $u$ and $v$ are marked, the lower-ID node unmarks itself. When doing this, we obtain
\[
    \Phi(\vec{x}) = \sum_{v\in V} x_v - \sum_{\set{u,v}\in E} \!\!x_u\cdot x_v
    = \sum_{\set{u,v}\in E} \!\!\bigg(\underbrace{\frac{x_u}{\deg(u)}+\frac{x_v}{\deg(v)} - x_u\cdot x_v}_{=: \phi_{\set{u,v}}(\vec{x})}\bigg).
\]
That is, the expected size of the resulting independent set can be expressed by a potential function $\Phi(\vec{x})$ that can be written as a sum over individual edge potentials $\phi_e(\vec{x})$. By iterating over the colors of a proper coloring of the graph $G$, we can therefore again round the fractional value of each node such that the independent set in the resulting rounded solution is at least as large as the expected independent set size initially. However, unlike in the case of the coloring algorithm of \cite{GhaffariK21}, when speeding up defective coloring and rounding gradually, a direct generalization of the algorithm of \cite{GhaffariK21} does not work. We next discuss the reason and explain what we do instead.

\paragraph{Fast Rounding of Large Fractional Independent Sets.} Note that the potential $\Phi(\vec{x})$ can be split into a positive \emph{utility} $\utility(\vec{x})$ and a negative \emph{cost} $\cost(\vec{x})$. We have $\Phi(\vec{x})=\utility(\vec{x}) - \cost(\vec{x}) =\sum_{e\in E} \phi_e(\vec{x}) = \sum_{e\in E} \big(\utility_e(\vec{x}) - \cost_e(\vec{x})\big)$, where
\[
\utility(\vec{x}) = \sum_{v\in V} x_v,\ 
\cost(\vec{x}) = \sum_{e\in E} x_u x_v,\text{ and for }
e=\set{u,v}\in E: 
\utility_{e}(\vec{x}) = \frac{x_u}{\deg(u)}+\frac{x_v}{\deg(v)},\ 
\cost_{e}(\vec{x}) = x_u x_v. 
\]
Assume again that we round gradually such that in a single rounding step each fractional $x_v$ value is at most doubled and that we use a defective coloring to obtain a fast implementation of such a rounding step. We then again have to assume that the rounding on monochromatic edges happens in a worst-case way. However, while gradual rounding guarantees that   $\cost_{e}(\vec{x})$ can increase by at most a factor $4$, we cannot lower bound $\utility_{e}(\vec{x})$ and even if we could, the difference $\phi_e(\vec{x})=\utility_e(\vec{x})-\cost_e(\vec{x})$ can change by an arbitrary factor. Note that before the rounding step we could even have $\phi_e(\vec{x})=0$ so that the multiplicative change in the edge potential even becomes unbounded. 

It is however possible to efficiently implement a gradual rounding step if for the fractional assignment $\vec{x}$ before the rounding, $\utility(\vec{x})$ is by a constant factor larger than $\cost(\vec{x})$, i.e., if for example $\utility(\vec{x})\geq 2\cost(\vec{x})$. In this case, one can compute a weighted $\delta$-relative defective coloring w.r.t.\ edge weights $\utility_e(\vec{x})+\cost_e(\vec{x})$. One can then show that even doing a worst-case gradual rounding step for all the monochromatic edges (w.r.t.\ the defective coloring), the potential decreases by at most $O\big(\delta\cdot(\utility(\vec{x}) +\cost(\vec{x}))\big)=O\big(\delta\cdot\Phi(\vec{x})\big)$. For our initial fractional independent set, the condition  $\utility(\vec{x})\geq 2\cost(\vec{x})$ is satisfied and we can therefore efficiently carry out the first rounding step. However, in a single rounding step, we can only guarantee that the difference $\utility(\vec{x})-\cost(\vec{x})$ is approximately preserved and we cannot guarantee that the terms $\utility(\vec{x})$ and $\cost(\vec{x})$ are approximately preserved individually. We therefore cannot guarantee that the condition $\utility(\vec{x})\geq 2\cost(\vec{x})$ is also approximately preserved. To cope with this challenge, we introduce one additional idea, which we describe next.

\paragraph{Dynamically Adapting the Potential Function.}
Instead of using a single potential function $\Phi(\vec{x})$ for the whole rounding process, we use a sequence $\Phi_0(\vec{x}), \Phi_1(\vec{x}), \dots, \Phi_{T}(\vec{x})$ of potential functions for the $T=O(\log\Delta)$ gradual rounding steps. Assume that initially, we have $\utility(\vec{x})\geq 2\cost(\vec{x})$ and that we are satisfied if we can maintain the value of $\utility(\vec{x}) -\cost(\vec{x})$ within some constant factor. We can then define
\begin{equation}\label{eq:potentialoverview}
  \Phi_i(\vec{x}) := \utility(\vec{x}) - \eta_i\cdot \cost(\vec{x})\text{, where }\eta_i = \frac{1}{2}\left(3 -   \frac{i}{T}\right).
\end{equation}
That is, we have $\Phi_0(\vec{x}) = \utility(\vec{x}) - \frac{3}{2}\cdot \cost(\vec{x})$, $\Phi_T(\vec{x}) = \utility(\vec{x}) - \cost(\vec{x})$, and $\Phi_{i+1}(\vec{x}) - \Phi_i(\vec{x}) = \frac{3}{2T}\cdot\cost(\vec{x})$. In rounding step $i$, we use the potential function $\Phi_i$ for the rounding. By slightly increasing the gap between the positive and the negative term in the potential function for each rounding step, we can make sure that $\utility(\vec{x})$ is always sufficiently larger than $\cost(\vec{x})$ in order to efficiently perform the rounding step and to make sure that after the rounding step $i$, $\Phi_i$ is within a $(1-O(\delta))$-factor of $\Phi_{i-1}$. That is, in the end, we obtain an integral assignment $\vec{x}'$ for which $\utility(\vec{x}')- \cost(\vec{x}')$ is within a constant factor of $\utility(\vec{x}) - \frac{3}{2}\cdot\cost(\vec{x})$ for the original fractional assignment $\vec{x}$. Overall, we obtain an algorithm to obtain an independent set of size $\Omega(n/\Delta)$ in $O(\log^2\Delta + \log^* n)$ rounds, where the $O(\log^* n)$-term comes from computing a proper $O(\Delta^2)$-coloring of $G$ at the very beginning of the algorithm. The technical details of the rounding procedure appear in \Cref{sec:basicrounding}.

The described rounding method for independent sets works for all fractional independent set assignments $\vec{x}$ with $\utility(\vec{x})\geq 2\cost(\vec{x})$. In this way, we can directly obtain an independent of size $n/\deg_{\mathrm{avg}}$, where $\deg_{\mathrm{avg}}$ denotes the average degree, and analogously in weighted graphs, we can obtain an independent set of total weight asymptotically at least a weighted average degree fractional of the total weight. For the corresponding definitions and details, we refer to \Cref{sec:turan}. In the same way, in graphs of neighborhood independence at most $\beta$, we can obtain a $1/O(\beta)$-approximation to the maximum weight independent set problem in time $O(\log^2(\Delta W)+\log^* n)$ (where $W$ is the ratio between largest and smallest weight). By adapting an idea of \cite{KawarabayashiKS20} and repeating the $1/O(\beta)$-approximation algorithm in an appropriate way, at the cost of an additional $\log(1/\eps)$ factor in the round complexity, we can even obtain a $(1-\eps)/\beta$-approximation for maximum weight independent set.

\subsubsection{Our Generic Rounding Algorithm}
\label{sec:genericroundingsummary}
In \Cref{sec:genericrounding}, we describe a general rounding algorithm that is based on the ideas discussed for the maximum independent set problem above. We consider local graph problems where each node has to choose one label from a finite alphabet $\Labels$ of labels. In the $(\Delta+1)$-coloring problem, the set of labels is $\Labels=\set{1,\dots,\Delta+1}$ or more generally the colors from some given color space if we consider the list coloring variant of the problem. In the independent set problem discussed above, there are only two possible labels, either a node decides to be in the independent set or it decides not to be in the independent set. In fact, in all problems for which is provide a new algorithm in this paper, we only have two labels. However, to cover the larger range of possible problems, the generic algorithm in \Cref{sec:genericrounding} is written for an arbitrary set $\Labels$ of labels.

As in the maximum independent set example above, the quality of an assignment of labels to the nodes is measured by utility function $\utility(\cdot)$ and a cost function $\cost(\cdot)$. Both functions can be expressed as the sum of pairwise functions and the set of node pairs that contribute to the overall utility and/or cost naturally define a graph $H=(V_H,E_H)$ on the set of nodes. That is, utility and cost are defined as the sum of individual edge utilities and costs over the edges in $E_H$. Sometimes this graph $H$ is equal to the communication graph (such as in the discussed vertex coloring and  maximum independent set examples) and sometimes $H$ is a graph that can be simulated efficiently by a distributed algorithm on the network graph (such as in the examples of MIS and set cover discussed below). In the latter case, the graph $H$ might also contain multiple edges between the same pair of nodes and those parallel edges might be simulated in different ways on the underlying communication graph. In our formal set-up, we therefore think of $H$ as a multigraph. A fractional label assignment $\lambda$ assigns a fractional value $\lambda_\alpha(v)\in[0,1]$ for every label $\alpha\in \Labels$ to every node $v\in V$ such that for all $v\in V$, $\sum_{\alpha\in \Labels} \lambda_\alpha(v)=1$, i.e., the fractional assignment defines a probability distribution over the labels to each node. The utility and cost of a fractional assignment are given as the expected utility and cost values when each node $v$ picks its label independently according to the distribution given by its fractional values $\lambda_\alpha(v)$.

Given some fractional label assignment $\lambda$, the goal of the rounding algorithm is to obtain an integral label assignment $\lambda'$ for which $\utility(\lambda')-\cost(\lambda')\geq (1-\eps)\cdot\big(\utility(\lambda)-\cost(\lambda\big)$ for some given parameter $\eps>0$. If initially $\utility(\lambda)-\cost(\lambda)\geq \mu \utility(\lambda)$ for some $\mu\in(0,1]$, our rounding algorithm achieves this in time linear proportional $1/(\eps \mu)$. Our algorithm is based on gradually rounding the solution as described above for rounding fractional independent set solutions. More specifically, for some integer $k\geq 0$, we call a fractional label assignment $\lambda$ $2^{-k}$-integral if all fractional values $\lambda_\alpha(v)$ are integer multiples of $2^{-k}$. In a single rounding step, we turn a $2^{-k}$-integral solution into a $2^{-(k-1)}$-integral one. In the general case, the rounding is implemented in exactly the same way as sketched above for the case of independent sets. in particular, we have to use a dynamically changing potential function analogous to \eqref{eq:potentialoverview} in order to keep utility and cost sufficiently separated to make each rounding step efficient. In each step, we have to use a weighted average defective coloring with $O(k/(\eps \mu))$ colors so that an implementation of a single rounding step would take $O(k/(\eps \mu))$ rounds when running it on $H$ and with potentially large messages. Below, we discuss some of the ideas that are necessary to implement the rounding algorithm with small messages. The precise formal setup, the detailed generic rounding algorithm, its analysis, and implementation in the distributed setting appear in \Cref{sec:genericrounding}.

\subsubsection{Maximal Independent Set}

The standard randomized distributed algorithm to compute an MIS of a graph $G=(V,E)$ is an algorithm that is usually referred to as Luby's algorithm~\cite{alon86,luby86}. One version of this algorithm consists of iteratively applying the following basic step. Every node $v\in V$ marks itself with probability $1/(2\deg(v))$. Each marked node $u$ joins the independent set unless $u$ has a higher priority marked neighbor $v$. A neighbor $v$ has higher priority if either $\deg(v)>\deg(u)$ or if $\deg(v)=\deg(u)$ and if $v$ has a larger ID than $u$. After that, the new independent set nodes and their neighbors are removed from the graph and the step is again applied on the remaining graph. The standard analysis shows that, in expectation, a constant fraction of the edges of the graph are removed in each step. This is shown by defining a set of good nodes, showing that a constant fraction of all edges are incident to at least one good node, and showing that good nodes are removed with constant probability. In \Cref{sec:MIS}, we show that this algorithm and analysis can be adapted to the setting to which we can apply our generic rounding method. This allows to deterministically find an independent set that removes a constant fraction of the good nodes and thus a constant fraction of all the edges.

To illustrate some of the key ideas, in this overview, we consider a significantly simpler setting. We assume that $G$ is $\Delta$-regular and we discuss how to find an independent set that removes a constant fraction of all the nodes in this case. In \Cref{sec:MIS}, we do a careful per-node analysis to count the number of edges incident on good nodes that are removed. Here, where things are more symmetric, we can directly analyze the total number of nodes that are removed. 

For the initial fractional solution, we set $x_v=1/(2\Delta)$ for each node $v\in V$. After rounding to an integral solution, we proceed as in the independent set algorithm above. For every edge $\set{u,v}$ on which both nodes are marked, we unmark the lower-ID node. The set of marked nodes then forms an independent set. In order for this independent set to remove a constant fraction of the nodes, we have to build a utility and a cost function that measures the total number of nodes that get removed. Each node $v\in V$ can remove all its $\Delta$ neighbors if $v$ joins the independent set. We therefore define the overall utility as $\utility(\vec{x})=\Delta\cdot\sum_{v\in V} x_v$. In order to lower bound the number of nodes that get removed, we have to (a) deduct the number of neighbors of nodes that get unmarked because they have a higher-ID marked neighbor and (b) make sure that we do not double count the removal of nodes. For (a), we can proceed in the same way as when computing large independent sets. For each edge $\set{u,v}$ for which both nodes are marked, the lower-ID node gets unmarked and we have to subtract $\Delta$ from the number of removed nodes. Therefore, the cost function contains a term of the form $\Delta\cdot\sum_{\set{u,v}\in E}x_u x_v$. For (b), note that the removal of a node $v$ is counted $\kappa$ times if $\kappa$ neighbors of $v$ join the independent set. Hence, the removal of node $v$ is counted at most $\kappa$ times if $\kappa$ neighbors of $v$ are initially marked. Further, the number of marked neighbors of $v$ minus $1$ (i.e., the number of times $v$'s removal is overcounted) can clearly be upper bounded by the total number of node pairs among the neighbors of $v$ such that both nodes in the pair are marked. The cost function therefore also contains a second term that upper bounds the number of overcounted node removals as $\sum_{u\in V}\sum_{\set{v,w}\in \binom{N(u)}{2}}x_v x_w$, where $\binom{N(u)}{2}$ denotes the set of $2$-element subsets of $N(u)$. We can therefore lower bound the expected number of removed nodes as
\[
\utility(\vec{x}) - \cost(\vec{x}) = \Delta\cdot\sum_{u\in V} x_u\  -\ \Delta\cdot\!\!\sum_{\set{u,v}\in E} x_u\cdot x_v - \sum_{u\in V}\sum_{\set{v,w}\in \binom{N(u)}{2}} x_v\cdot x_w.
\]
Note that the utility and cost functions can be expressed as sums over individual edge utilities and costs, if we do this on a graph $H$ for which we add an edge between any two nodes at distance $2$ in $G$. By using our generic rounding algorithm, we can therefore round a fractional solution to an integral one of approximately the same quality (i.e., an independent set that removes approximately $\utility(\vec{x}) - \cost(\vec{x})$ nodes, where $\vec{x}$ is the initial fractional assignment). For the initial fractional assignment $x_v=1/(2\Delta)$ for all $v\in V$, the expected number of removed nodes can be lower bounded by
\[
\utility(\vec{x}) - \cost(\vec{x}) = \frac{n}{2} - \Delta\cdot\frac{n\Delta}{2}\cdot\frac{1}{4\Delta^2} -
n\cdot \binom{\Delta}{2} \cdot \frac{1}{4\Delta^2}\geq \frac{n}{4}.
\]
Since we start with a $\Theta(1/\Delta)$-fractional assignment, the rounding will require $O(\log\Delta)$ steps and each step requires $O(\log\Delta)$ rounds (when communicating on $H$). In time $O(\log^2\Delta)$ we can therefore compute an independent set that removes a constant fraction of the nodes (after computing an initial $O(\Delta^2)$-coloring in time $O(\log^* n)$ at the very beginning of the algorithm). We show in \Cref{sec:MIS} that, in the same time complexity and with a natural generalization of the same cost function, we can compute an independent set that removes a constant fraction of the edges in time $O(\log^2\Delta)$ in the \LOCAL model. Therefore, we  obtain an $O(\log^2\Delta\cdot\log n)$-round deterministic MIS algorithm for the \LOCAL model. 

\subsubsection{Set Cover Approximation}
We do not provide a complete overview of our set cover approximation algorithm here, however, we remark one key difference: In the set cover problem, we have $O(\log s)$ iterations of rounding the fractional solution (where $s$ is the maximum set size). To avoid the need for gathering the parameters of the instance remaining after each iteration (which would require global communication), we define a potential function (in the format of utility minus cost) for the entirety of the iterations. This is unlike the maximal independent set solution, where each iteration is handled separately, with its own potential function. The particular potential function we use and our gradual rounding ensure that, while per iteration we lose a constant factor in comparison to what the randomized algorithm would achieve, over the entire span of the $O(\log s)$ iterations, we still manage to cover all but a $1/poly(s)$ fraction of elements. After that, a simple greedy step of taking a set for each remaining element finishes the problem at a negligible cost. We refer to \Cref{sec:setcover} for the details and the algorithm.

\subsubsection{Implementation With Small Messages}
Note that both for computing an MIS and for computing a small set cover, the rounding graph includes edges that are between nodes at distance of $2$ in the communication graph $G$. While such edges are trivial to simulate in the \LOCAL model with only a factor $2$ overhead, the same is not true for the \CONGEST model. That is, when  running a \CONGEST algorithm on the rounding graph $H$ on the underlying communication graph $G$, the necessary message size on $G$ can potentially grow by a factor linear in $\Delta$, which would not be feasible within $\poly(\log n)$-time. Therefore, to obtain efficient \CONGEST algorithms on $G$ that perform rounding algorithms on $H$, we cannot work with a black-box interface; we have to instead examine the actual communication needed for the rounding. 

For each virtual edge of $H$ which is between two nodes $v$ and $u$ that are at distance $2$ and have a common neighbor $w$, we make $w$ the manager of this virtual edge and let $w$ simulate the computations and communications relevant for this edge. Notice that a node $v$ might not even know all of its virtual edges. However, the managers of all those virtual edges are direct neighbors of $v$. Since $v$ can share its state (e.g., its current color, or its current fractional label assignment) efficiently with its neighbors, the managers are able to simulate the computations and communications of the virtual edges. The only possible difficulty is that this information has to be delivered back to node $v$, who is the endpoint of those virtual edges and will use this information to take some action, e.g., set its new color or increase/decrease its fractional value. Fortunately, in all cases, simply aggregating the information that needs to be sent back to $v$ suffices, so each manager $w$ can send to its direct neighbor $v$ some aggregate information about all the virtual edges that $w$ is managing for $v$. Furthermore, we note that doing this naively would result in an additional $\log \Delta$ factor in the round complexity, due to the message sizes in a part of the defective coloring procedure. However, fortunately, in that case, we can make do with a $2$ approximate aggregation of the values, instead of the exact values, and this reduces the round complexity loss exponentially, to only a $\log\log \Delta$ factor (which is essentially the only difference between the round complexity in the \LOCAL model and that of the \CONGEST model. Please see \Cref{sec:d2communication} for details.

\subsection{Mathematical Preliminaries and Notation}
We use $\Rp:=\set{x\in \mathbb{R} : x\geq 0}$ to denote the non-negative reals. Further, for a positive integer $N$, we use $[N]$ as a shortcut for $[N]:=\set{1,\dots,N}$. For an undirected multigraph $G=(V,E)$ and a node $v\in V$, we use the following notation. For every edge $e$, we use $V(e)$ to denote the set of nodes of $e\in E$ (note that if $G$ has no self-loops, $|V(e)|=2$ for all $e\in E$). Further, for every node $v\in V$, we use $E(v)\subseteq E$ to denote the set of edges that contain node $v$. For convenience, for any subsets $F\subseteq E$ of the edges and any node $v\in V$, we also define the shortcut $F(v):=E(v)\cap F$. If $G$ is equipped with edge (or node) weights $w:E\to \R$ (or $w:V\to \R$), for any set $F\subseteq E$ (or $U\subseteq V$), we use $w(F)$ (or $w(U)$) to denote the sum of the weights of all edges in $F$ (or nodes in $U$). When dealing with simple graphs, we also use the more standard definition and consider the edges $E$ as a subset of $\binom{V}{2}$, where $\binom{V}{2}$ denotes the set of $2$-element subsets of $V$. For a node $v\in V$, we use $N(v)$ to denote the set of neighbors of $v$ and we use $N^+(v):=\set{v}\cup N(v)$ to denote the inclusive neighborhood of $v$. For a graph $G=(V,E)$, we will typically use $n$ and $\Delta$ to refer to the number of nodes and the maximum degree of $G$, respectively. Finally, we say that a graph $G=(V,E)$ has \emph{neighborhood independence} $\beta\geq 1$ if for every node $v\in V$, the the largest independent set of the subgraph induced by $N(v)$ (or equivalently by $N^+(v)$) is of size at most $\beta$. Note that in particular, line graphs of hypergraphs of rank $r$ have neighborhood independence at most $r$. 

As a technical tool to obtain fast deterministic rounding algorithms, we need to compute \emph{defective colorings} of the graph. More concretely, we need an edge-weighted version of standard defective coloring as it was introduced in \cite{KawarabayashiS18}. In order to obtain efficient defective coloring algorithms, we use a relaxed variant of defective coloring in which the (relative) defect per node only needs to be small on average.

\begin{definition}[Weighted Defective and Weighted Average Defective Coloring]\label{def:avgdefcoloring}
  Given a weighted undirected multigraph $G=(V,E,w)$ with non-negative edge weights $w(e)\geq 0$ for all $e\in E$, a parameter $\eps>0$, and an integer $C\geq 1$, a \emph{weighted $\eps$-relative defective $C$-coloring} of $G$ is an assignment $\varphi:V\to [C]$ of colors in $[C]$ to the nodes of $G$ such that
  \[
  \forall v\in V\,:\,\sum_{e \in E(v), \, V(e)=\{v, u\}}\1_{\set{\varphi(u)=\varphi(v)}}\cdot
  w(e) \ \leq\  \eps\cdot\sum_{e \in E(v)} w(e).
  \]
  The coloring is called a \emph{weighted average $\eps$-relative defective $C$-coloring} if the above condition only holds on average over all nodes:
  \[
  \sum_{v\in V}\,\,\sum_{e \in E(v), \, V(e)=\{v, u\}}\1_{\set{\varphi(u)=\varphi(v)}}\cdot
  w(e) \ \leq\  \eps\cdot\sum_{v\in V}\sum_{e \in E(v)} w(e).
  \]
\end{definition}

As shown in \cite{KawarabayashiS18}, by slightly adapting an algorithm of \cite{Kuhn2009WeakColoring}, an $\eps$-relative defective $O(1/\eps^2)$-coloring can be computed in time $O(\log^* \xi)$ in the \CONGEST model, if an initial proper vertex coloring with $\xi$ colors is given. Further, by using algorithm of \cite{GhaffariK21} (or by adapting an  coloring algorithm of \cite{BEG18}), an $\eps$-relative average defective $O(1/\eps)$-coloring can be computed in time $O(1/\eps + \log^* \xi)$. We also note that the notion of average defect is a relaxation of the notion of arbdefect that was introduced in \cite{barenboim10}.
 

%% file: genericrounding.tex
\section{Generic Distributed Rounding Algorithm}
\label{sec:genericrounding}

\subsection{Setup: Multigraph with Edge Utilities and Costs}
\label{sec:roundingsetup}

Let $H=(V,E)$ be an undirected multigraph with no self-loops and let $\Labels$ be a finite set of possible node labels. A label assignment for $H$ is a function $\ell: V\to \Labels$ that assigns a label $\alpha=\ell(v)\in \Labels$ to each node $v\in V$. A fractional label assignment $\lambda:V\to [0,1]^{|\Labels|}$ for $H$ is an assignment of a distribution (i.e., convex combination) of labels to each node. That is, for each node $v\in V$, we have $\sum_{\alpha\in \Labels}\lambda_\alpha(v)=1$, where we use $\lambda_\alpha(v)\in [0,1]$ to refer to the component corresponding to label $\alpha$ in the vector $\lambda(v)$. For an integer $K\geq 1$, a fractional label assignment $\lambda$ is called $1/K$-integral if for every $v\in V$ and every $\alpha\in \Labels$, $\lambda_\alpha(v)=a/K$ for some integer $a\geq 0$. For a set of nodes $S\subseteq V$ and a label assignment $\ell$ (or a fractional label assignment $\lambda$), we use $\ell(S)$ (or $\lambda(S)$) to denote the (fractional) label assignment to the set of nodes in $S$. Similarly for an edge $e\in E$, we use $\ell(e)$ and $\lambda(e)$ to denote $\ell(V(e))$ and $\lambda(V(e))$. Further, for each node $v$, we use $\calL(v)$ (and $\Lambda(v)$) to denote the set of possible (fractional) label assignments to $v$ and we similarly use $\calL(S)$, $\Lambda(S)$, $\calL(e)$, $\Lambda(e)$ for node sets $S$ and edges $e$.

For a multi-graph $H=(V,E)$, we define a non-negative utility function $\utility:E\times \calL(V) \to \Rp$ and a non-negative cost function $\cost:E\times \calL(V) \to \Rp$. 
That is, for a given label assignment $\ell\in \calL(V)$, $\utility$ and $\cost$ assign non-negative utility and cost values $\utility(e,\ell)$ and $\cost(e,\ell)$ to every edge $e\in E$. Although utility and cost of an edge $e$ are defined as a function of the label assignment to all nodes in $V$, they must only depend on the label assignment to the two nodes $V(e)$ of $e$. That is, for any two label assignments $\ell,\ell'\in \calL(V)$ such that $\ell(e)=\ell'(e)$ (i.e., $\ell$ and $\ell'$ assign the same labels to the nodes in $V(e)$), it must hold that $\utility(e,\ell)=\utility(e,\ell')$ and $\cost(e,\ell)=\cost(e,\ell')$. 

We slightly overload notation and also define $\utility$ and $\cost$ for a fractional label assignment $\lambda\in\Lambda(V)$. In this case, utility and cost of $e$ are defined as the expected values of $\utility$ and $\cost$ if the labels of the two nodes $V(e)$ of $e$ are chosen independently at random from the distributions given by the fractional label assignment. That is, if $V(e)=\set{u,v}$ and if for each $(\alpha,\beta)\in \Sigma^2$, $\ell_{\alpha,\beta}(e)$ is a label assignment that assigns label $\alpha$ to $u$ and label $\beta$ to $v$, then for $\lambda\in \Lambda(V)$, we have
\begin{equation}\label{eq:fractionalutilitycost}
    \utility(e, \lambda) := \!\!\sum_{(\alpha,\beta)\in \Sigma^2}\!\!
    \lambda_\alpha(u)\cdot\lambda_\beta(v)\cdot
    \utility(e,\ell_{\alpha,\beta}(e))
    \quad\text{and}\quad
    \cost(e, \lambda) := \!\!\sum_{(\alpha,\beta)\in \Sigma^2}\!\!
    \lambda_\alpha(u)\cdot\lambda_\beta(v)\cdot
    \cost(e, \ell_{\alpha,\beta}(e)).
\end{equation}
Finally, for a set of edges $F\subseteq E$ and a label assignment $\ell$, we use $\utility(F,\ell)=\sum_{e\in F}\utility(e,\ell)$, and $\cost(F,\ell)=\sum_{e\in F}\cost(e,\ell)$. The definitions extend analogously if the label assignment $\ell$ is replaced with a fractional label assignment $\lambda$. Finally, to denote the total cost and utility of a label assignment $\ell$, we define $\utility(\ell):=\utility(E,\ell)$ and $\cost(\ell):=\cost(E,\ell)$. Again, the definition extends analogously if $\ell$ is replaced with a fractional label assignment $\lambda$.

\paragraph{Polynomially Bounded Instances.} In some cases, the message sizes (or the time complexities in case message sizes are bounded) of our algorithms depend on the range of values that utilities and costs can have. We say that a given instance (i.e., multigraph with corresponding utility and cost functions and with a fractional label assignment $\lambda$) is \emph{polynomially bounded in $q$} for some parameter $q>1$ if there is a value $Q\leq q^d$ for some constant $d>0$ such that every non-zero edge utility or cost is lower bounded by a value $1/Q$ and upper bounded by $Q$ and such that any non-zero fractional value $\lambda_\alpha(v)\geq 1/Q$. Most of our instances will be polynomially bounded in the maximum degree $\Delta$ of the graph.

\subsection{Abstract Basic Rounding Algorithm}
\label{sec:basicrounding}

In the following, we first describe our local rounding algorithm in an abstract form, which is somewhat independent of the specific communication model. In the concrete distributed implementation of the algorithm, we will assume that in the multigraph $G=(V,E)$ on which we run the rounding, the nodes $V$ are active entities that locally perform their part of the rounding algorithm. In some applications, some edges of $G$ might however not be physical communication links, so the communication between nodes in $V$ depends on the relation between $G$ and the underlying physical communication network. In \Cref{alg:basicroundingstep}, we first describe the algorithm for a single rounding step in detail. The full rounding procedure is then given by \Cref{alg:basicrounding}.

\begin{algorithm}[ht]
  \caption{Basic Rounding Step, Parameters $\delta\in[0,1]$, $\eta\geq 1$}\label{alg:basicroundingstep}
  \textbf{Given:} multigraph $G=(V,E)$ with utility and cost functions $\utility$ and $\cost$ and with a $1/(2K)$-integral fractional label assignment $\lambda$ (for a given integer $K\geq 1$).
\begin{algorithmic}[1]
  \State For each $e\in E$, define weight $w_e := \utility(e,\lambda)+ \eta\cost(e,\lambda)$
  \State Compute a weighted $\delta/6$-relative average defective $p$-coloring of $G$ w.r.t.\ the edge weights $w_e$
  \State Define $E_b\subseteq E$ to be the set of bichromatic edges (w.r.t.\ the defective coloring)
  \ForAll {colors $\gamma\in \set{1,\dots,p}$}
  \ForAll {nodes $v\in V$ of color $\gamma$ (in the defective coloring) \textbf{in parallel}}
  \State ${\Labels}_v:=\set{\alpha\in \Labels : \lambda_\alpha(v)=\frac{2i+1}{2K}\text{ for some integer }i\geq 0}$
  \ForAll {labels $\alpha\in{\Labels}_v$}
  \State Define $\lambda^{(v,\alpha)}$ s.t.\ $\lambda^{(v,\alpha)}_\alpha(v)=1$, $\lambda^{(v,\alpha)}_\beta(v)=0$ for $\beta\neq \alpha$, and $\lambda^{(v,\alpha)}=\lambda$ otherwise
  \State Define $\phi_{v,\alpha}:=\utility\big(E_b(v),\lambda^{(v,\alpha)}\big)-\eta\cdot\cost\big(E_b(v),\lambda^{(v,\alpha)}\big)$
  \State Define $\theta_{v,\alpha}:=\utility\big(E_b(v),\lambda^{(v,\alpha)}\big)+\eta\cdot\cost\big(E_b(v),\lambda^{(v,\alpha)}\big)$.
  \State $v$ obtains estimate $\hat{\phi}_{v,\alpha}$ satisfying $\phi_{v,\alpha}\geq \hat{\phi}_{v,\alpha}\geq 
  \phi_{v,\alpha} - \frac{\delta}{6}\cdot\theta_{v,\alpha}$
  \EndFor
  \State Partition ${\Labels}_v$ into ${\Labels}_v^+$ and ${\Labels}_v^-$ s.t.\ $|{\Labels}_v^+|=|{\Labels}_v^-|$ and s.t.\ $\forall(\alpha,\beta)\in {\Labels}_v^+\times {\Labels}_v^-$, $\hat{\phi}_{v,\alpha}\geq \hat{\phi}_{v,\beta}$
  \State For $\alpha\in {\Labels}_v^+$, set $\lambda_\alpha(v):=\lambda_{\alpha}(v) + \frac{1}{2K}$ and for $\alpha\in {\Labels}_v^-$, set $\lambda_\alpha(v):=\lambda_{\alpha}(v) - \frac{1}{2K}$
  \EndFor
  \EndFor
\end{algorithmic}
\textit{For line 12 of the algorithm, note that the set $\Labels_v$ must have an even size. Therefore, it can be split into $\Labels_v^+$ and $\Labels_v^-$ such that $|{\Labels}_v^+|=|{\Labels}_v^-|$.}
\end{algorithm}

\begin{lemma}\label{lemma:basicroundingstep}
  Let $G=(V,E)$ be a multigraph, which is equipped with utility and cost functions $\utility$ and $\cost$ and with a $1/(2K)$-integral fractional label assignment $\lambda$ for a given integer $K\geq 1$. Let $\delta\in[0,1]$ and $\eta\geq 1$ be two parameters. Assume that \Cref{alg:basicroundingstep} is run on $G$ with the given parameters and let $\lambda'$ be the fractional label assignment after running the algorithm. Then, $\lambda'$ is a $1/K$-integral fractional label assignment and
  \[
    \utility(\lambda') - \eta\cost(\lambda') \geq \utility(\lambda) - \eta\cost(\lambda) -
    \delta\cdot\big(\utility(\lambda)+\eta\cost(\lambda)\big).
  \]
\end{lemma}
\begin{proof}
  First, the algorithm only changes fractional values $\lambda_\alpha(v)$ of the form $(2i+1)/(2K)$ and all those values are either decremented or incremented by exactly $1/(2K)$. The fractional label assignment at the end of the algorithm is thus $1/K$-integral.
  
  For each color $\gamma\in\set{1,\dots,p}$ of the average defective coloring, we define $\lambda_\gamma$ to be the fractional label assignment after processing the nodes of color $\gamma$ in \Cref{alg:basicroundingstep}. We further define $\lambda_0$ to be the fractional label assignment at the beginning of the algorithm. Note that $\lambda_0=\lambda$ and $\lambda_p=\lambda'$.

  We first upper bound how the utility and cost of an edge change during the algorithm. Consider an edge $e$ between two nodes $u$ and $v$ (i.e., $V(e)=\set{u,v}$). Assume that in the average defective coloring $u$ has color $\gamma_u$ and $v$ has color $\gamma_v$ and w.l.o.g., assume that $\gamma_u\leq \gamma_v$. Note that when rounding the fractional label assignment, the fractional value of a node for a given label can at most double. When rounding a single node, the utility and cost of its incident edges can therefore also at most double and when rounding both nodes of an edge, the utility and cost can grow by a factor at most $4$. For any $\gamma\in\set{0,\dots,p}$ and edge $e$ between $u$ and $v$, we therefore have
\begin{equation}\label{eq:edgeroundingupper}
  \utility(e,\lambda_\gamma) \leq 
  \begin{cases}
    \utility(e, \lambda) & \text{if } \gamma < \gamma_u\\
    2\utility(e, \lambda) & \text{if } \gamma_u\leq \gamma < \gamma_v\\
    4\utility(e, \lambda) & \text{if } \gamma \geq \gamma_v
  \end{cases}
  \qquad\text{and}\qquad
  \cost(e,\lambda_\gamma) \leq 
  \begin{cases}
    \cost(e, \lambda) & \text{if } \gamma < \gamma_u\\
    2\cost(e, \lambda) & \text{if } \gamma_u\leq \gamma < \gamma_v\\
    4\cost(e, \lambda) & \text{if } \gamma \geq \gamma_v.
  \end{cases}
\end{equation}

  As defined in \Cref{alg:basicroundingstep}, let $E_b$ the the set of bichromatic edges w.r.t.\ the computed average defective $p$-coloring and let $E_m:=E\setminus E_b$ be the set of monochromatic edges w.r.t.\ the same coloring. Note that we have $\utility(\lambda)=\utility(E,\lambda)=\utility(E_b,\lambda)+\utility(E_m,\lambda)$ and analogously also for the cost function $\cost(\cdot)$ and the fractional assignment $\lambda'$. We first consider the contribution of the monochromatic edges to $\utility(\lambda')-\eta\cost(\lambda')$. We have
  \begin{eqnarray}
    \utility(E_m, \lambda')-\eta\cost(E_m, \lambda') & \stackrel{\utility(\cdot)\geq 0, \eqref{eq:edgeroundingupper}}{\geq} & -4\eta\cost(E_m, \lambda)\nonumber\\
    & = & \utility(E_m, \lambda) - \eta\cost(E_m, \lambda) - \utility(E_m, \lambda) - 3\eta\cost(E_m, \lambda)\nonumber\\
    & \geq & \utility(E_m, \lambda) - \eta\cost(E_m, \lambda) - 3\big(\utility(E_m, \lambda) + \eta\cost(E_m, \lambda)\big)\nonumber\\
    & \geq & \utility(E_m, \lambda) - \eta\cost(E_m, \lambda) -\frac{\delta}{2}\cdot \big(\utility(\lambda) + \eta\cost(\lambda)\big).\label{eq:boundmonochromatic}
  \end{eqnarray}
  The last inequality follows because the given vertex coloring is a weighted $\delta/6$-average defective coloring w.r.t. edge weights $\utility(e, \lambda) + \eta\cost(\lambda)$ and thus $\utility(E_m, \lambda) + \eta\cost(E_m, \lambda)\leq \delta/6\cdot(\utility(\lambda) + \eta\cost(\lambda))$.

  We next look at the contribution of the bichromatic edges to $\utility(\lambda')-\eta\cost(\lambda')$. For this, we first consider the effect of the rounding of a single node $v$. Assume that node $v$ has color $\gamma_v$ in the average defective coloring. Note that before iteration $\gamma_v$ of the for loop of \Cref{alg:basicroundingstep} in which the fractional assignment of $v$ is changed, the fractional label assignment is $\lambda_{\gamma_v-1}$ and directly after changing the fractional assignment of $v$, the fractional label assignment is $\lambda_{\gamma_v}$. Also note that by the definition of $\phi_{v,\alpha}$ in line 9 of \Cref{alg:basicroundingstep}, we have
  \begin{equation}\label{eq:vpotentialbefore}
    \utility(E_b(v), \lambda_{\gamma_v-1}) - \eta\cost(E_b(v), \lambda_{\gamma_v-1}) = 
    \sum_{\alpha\in\Labels} \lambda_{\gamma_v-1, \alpha}(v)\cdot \phi_{v,\alpha},
  \end{equation}
  where $\lambda_{\gamma_v-1,\alpha}(v)$ denotes the fractional value that $v$ has for label $\alpha$ in the fractional label assignment $\lambda_{\gamma_v-1}$. When rounding the fractional values of $v$, the fractional values are incremented by $1/(2K)$ for all labels $\alpha\in \Labels_v^+$ and the fractional values are decremented by $1/(2K)$ for all labels $\alpha\in\Labels_v^-$. Because the neighbors of a color different from $\gamma_v$ do not change their fractional values at the same time, after rounding the values of $v$ (i.e., at the end of iteration $\gamma_v$ of the for loop), we have
  \begin{eqnarray}
    \utility(E_b(v), \lambda_{\gamma_v}) - \eta\cost(E_b(v), \lambda_{\gamma_v}) 
    & = &
          \sum_{\alpha\in\Labels} \lambda_{\gamma_v, \alpha}\cdot \phi_{v,\alpha}\nonumber\\
    & \stackrel{\eqref{eq:vpotentialbefore}}{=} &
    \utility(E_b(v), \lambda_{\gamma_v-1}) - \eta\cost(E_b(v), \lambda_{\gamma_v-1})
    + \varphi_v.\label{eq:bichromaticchange}
  \end{eqnarray}
  where 
  \begin{equation}
    \label{eq:vpotentialchange}
    \varphi_v = \frac{1}{2K}\cdot \left[\sum_{\alpha\in\Labels_v^+} \phi_{v,\alpha} - \sum_{\alpha\in\Labels_v^-}\phi_{v,\alpha}\right].
  \end{equation}
  As in line 10 of \Cref{alg:basicroundingstep}, define $\theta_{v,\alpha}:=\utility\big(E_b(v), \lambda_{\gamma_v-1}^{(v,\alpha)}\big)+\eta\cost\big(E_b(v), \lambda_{\gamma_v-1}^{(v,\alpha)}\big)$, where $\lambda_{\gamma_v-1}^{(v,\alpha)}$ is a fractional label assignment that is equal to $\lambda_{\gamma_v-1}$, except that at node $v$, we have $\lambda_{\gamma_v-1,\alpha}(v)=1$ and $\lambda_{\gamma_v,\beta}(v)=0$ for $\beta\neq \alpha$. Note that similarly to \eqref{eq:vpotentialbefore} we have
  \begin{equation}\label{eq:utilitycostsum}
    \utility\big(E_b(v), \lambda_{\gamma_v-1}\big)+\eta\cost\big(E_b(v), \lambda_{\gamma_v-1}\big) = 
  \sum_{\alpha\in \Labels} \lambda_{\gamma_v-1,\alpha}(v)\cdot \theta_{v,\alpha}.
  \end{equation}
  By line 11 of \Cref{alg:basicroundingstep}, for every $\alpha\in \Labels_v^+\cup\Labels_v^-$, $v$ obtains an estimate $\hat{\phi}_{v,\alpha}$ such that $\phi_{v,\alpha}\geq \hat{\phi}_{v,\alpha}\geq \phi_{v,\alpha}-\frac{\delta}{6}\cdot\theta_{v,\alpha}$. We can therefore rewrite $\varphi_v$ (cf.\ \Cref{eq:vpotentialchange}) as
  \begin{eqnarray}
    \varphi_v 
    & \geq &
    \frac{1}{2K}\cdot \underbrace{\left[\sum_{\alpha\in\Labels_v^+} \hat{\phi}_{v,\alpha} - \sum_{\alpha\in\Labels_v^-}\hat{\phi}_{v,\alpha}\right]}_{\geq 0\text{ (by def.\ of $\Labels_v^+$ and $\Labels_v^-$)}}
    - \frac{1}{2K}\cdot\sum_{\alpha\in\Labels_v^-}\frac{\delta}{6}\cdot\theta_{v,\alpha}\nonumber\\
    & \stackrel{(\lambda_{\gamma_v-1,\alpha}(v)\geq \frac{1}{2K})}{\geq} &
    - \frac{\delta}{6}\cdot \sum_{\alpha\in\Labels_v^-}
    \lambda_{\gamma_v-1,\alpha}(v)\cdot \theta_{v,\alpha}\nonumber\\
    & \stackrel{\eqref{eq:utilitycostsum}}{\geq} & 
    -\frac{\delta}{6}\cdot\left(
    \utility\big(E_b(v), \lambda_{\gamma_v-1}\big)+\eta\cost\big(E_b(v), \lambda_{\gamma_v-1}\big)
    \right).\label{eq:varphibound}
  \end{eqnarray}
  By \eqref{eq:bichromaticchange}, we have 
  \begin{equation}\label{eq:totalpotentialchange}
    \utility(E_b, \lambda') - \eta\cost(E_b, \lambda') =
    \utility(E_b, \lambda) - \eta\cost(E_b, \lambda) + \sum_{v\in V} \varphi_v
  \end{equation}
  By using \eqref{eq:varphibound}, the term $\sum_{v\in V}\varphi_v$ can be lower bounded as
  \begin{eqnarray}
    \sum_{v\in V}\varphi_v
    & \geq &
             -\frac{\delta}{6}\cdot\sum_{v\in V}
             \left(\utility\big(E_b(v), \lambda_{\gamma_v-1}\big)+\eta\cost\big(E_b(v), \lambda_{\gamma_v-1}\big)\right)\nonumber\\
    & = &
          -\frac{\delta}{6}\cdot\sum_{e\in E_b}\sum_{v\in V(e)}
          \left(\utility\big(e, \lambda_{\gamma_v-1}\big)+\eta\cost\big(e, \lambda_{\gamma_v-1}\big)\right)\nonumber\\
    & \geq &
             -\frac{\delta}{6}\cdot\sum_{e\in E_b}
             3\big(\utility(e, \lambda)+\eta\cost(e,\lambda)\big)\label{eq:boundonincrease}\\
    & = & -\frac{\delta}{2}\cdot\big(\utility(E_b,\lambda)+\eta\cost(E_b,\lambda)\big)
    \ \geq\ -\frac{\delta}{2}\cdot\big(\utility(\lambda)+\eta\cost(\lambda)\big).\nonumber
  \end{eqnarray}
  Inequality \eqref{eq:boundonincrease} follows because by \eqref{eq:edgeroundingupper} for a bichromatic edge $e$ with $V(e)=\set{u,v}$ and $\gamma_u<\gamma_v$, we have $\utility(e, \lambda_{\gamma_u-1})\leq \utility(e,\lambda)$ and $\utility(e,\lambda_{\gamma_v-1})\leq 2 \utility(e,\lambda)$ and analogously $\cost(e, \lambda_{\gamma_u-1})\leq \cost(e,\lambda)$ and $\cost(e,\lambda_{\gamma_v-1})\leq 2 \cost(e,\lambda)$. The last inequality follows because utility and cost are non-negative functions. In combination with Equations \eqref{eq:boundmonochromatic} and \eqref{eq:totalpotentialchange}, we obtain
  \[
    \utility(\lambda') - \eta\cost(\lambda') \geq \utility(\lambda) -\eta\cost(\lambda) -
    \delta\cdot\big(\utility(\lambda)+\eta\cost(\lambda)\big)
  \]
  as required by the claim of the lemma.
\end{proof}

\Cref{alg:basicroundingstep} describes a single rounding step of our abstract generic rounding algorithm. One application of this basic step doubles the integrality of a given fractional label assignment. Given a $1/K$-integral fractional label assignment, we therefore have to invoke  \Cref{alg:basicroundingstep} $\log K$ times in order to obtain an integral label assignment. The details of this are given by \Cref{alg:basicrounding}.

\begin{algorithm}
  \caption{Basic Rounding Algorithm, Parameters $\eps\in[0,1]$, $\mu\in(0,1]$}\label{alg:basicrounding}
  \textbf{Given:} multigraph $G=(V,E)$ with utility and cost functions $\utility$ and $\cost$ and with a $1/2^k$-integral fractional label assignment $\lambda$ (for a given integer $k\geq 1$) such that $\utility(\lambda)-\cost(\lambda)\geq \mu\utility(\lambda)$.
\begin{algorithmic}[1]
  \State Define $\delta:=\frac{\eps\cdot\mu}{6k}$
  \For{$i\in\set{1,\dots,k}$}
  \State $\eta_i := 1+\left(1-\frac{i}{k}\right)\frac{\eps\cdot\mu}{2}$
  \State Run \Cref{alg:basicroundingstep} with param.\ $\delta$ and $\eta_i$ on the current $\frac{1}{2^{k-i+1}}$-integral fract.\ assignment $\lambda$
  \EndFor
  \State Define integral label assignment $\ell$ s.t.\ $\forall v\in V, \alpha\in \Labels : \ell(v)=\alpha\ \Leftrightarrow\ \lambda_\alpha(v)=1$
\end{algorithmic}
\end{algorithm}

\begin{lemma}\label{lemma:basicrounding}
  Let $G=(V,E)$ be a multigraph, which is equipped with utility and cost functions $\utility$ and $\cost$ and with a $1/2^k$-integral fractional label assignment $\lambda$ for a given integer $k\geq 1$. Let $\eps\in [0,1]$ and $\mu\in(0,1]$ be two parameters. If $\utility(\lambda)-\cost(\lambda)\geq \mu\utility(\lambda)$, \Cref{alg:basicrounding} returns an integral label assignment $\ell$ with
  \[
  \utility(\ell)-\cost(\ell) \geq (1-\eps)\cdot \big(\utility(\lambda) - \cost(\lambda)\big).
  \]
\end{lemma}
\begin{proof} 
  For the proof, we define $\lambda_0$ to be the $1/2^k$-integral fractional label assignment with which \Cref{alg:basicrounding} is started and for $i\in\set{1,\dots,k}$, we define $\lambda_i$ to be the fractional label assignment at the end of the $i^{\mathrm{th}}$ iteration of the for loop in \Cref{alg:basicrounding}. For all $i\in\set{0,\dots,k}$, we further define a potential
  \[
  \Phi_i := \utility(\lambda_i) - \eta_i\cdot \cost(\lambda_i),
  \]
  where $\eta_i =1 + \big(1-\frac{i}{k}\big)\frac{\eps\cdot\mu}{2}$ for all $i\in\set{0,\dots,k}$.
  By induction on $i$, we first show that a) $\lambda_i$ is a $2^i/2^k$-integral fractional label assignment and b) 
  \begin{equation}\label{eq:potentialinduction}
    \Phi_i \geq (1-\delta)^i\cdot \Phi_0.
  \end{equation}
  \noindent\textit{Base Case $i=0$:} We know that $\lambda_0$ is a $1/2^k$-integral fractional label assignment and \Cref{eq:potentialinduction} certainly also holds for $i=0$. Thus, both a) and b) hold.
  
  \medskip
  
  \noindent\textit{Induction Step $i>0$:} The induction hypothesis implies that a) $\lambda_i$ is a $2^i/2^k$-integral fractional label assignment and b) $\Phi_i \geq (1-\delta)^i\cdot \Phi_0$. We need to show that a) $\lambda_{i+1}$ is a $2^{i+1}/2^k$-integral fractional label assignment and  b) $\Phi_{i+1} \geq (1-\delta)^{i+1}\cdot \Phi_0$.
  
  For a), notice that at the beginning of the ${i+1}^{\mathrm{st}}$ iteration of the for loop  of \Cref{alg:basicrounding}, the current fractional label assignment is $\lambda_i$, which by the induction hypothesis is a $2^i/2^k$-integral fractional label assignment. By \Cref{lemma:basicroundingstep}, \Cref{alg:basicroundingstep} doubles the integrality of the given fractional label assignment. Hence, when applying \Cref{alg:basicroundingstep} in iteration $i+1$ of \Cref{alg:basicrounding}, we obtain a $2^{i+1}/2^k$-integral fractional label assignment and thus a) follows.
  
  For b), by \Cref{lemma:basicroundingstep}, at the end of the $i+1^{\mathrm{st}}$ iteration of the for loop in \Cref{alg:basicrounding}, we have 
  \begin{eqnarray*}
	\Phi_{i+1} & = & \utility(\lambda_{i+1}) - \eta_{i+1} \cost(\lambda_{i+1})\\
	& \stackrel{(L.\ \ref{lemma:basicroundingstep})}{\geq}  &  \utility(\lambda_{i}) - \eta_{i+1} \cost(\lambda_{i})- \delta \cdot ( \utility(\lambda_{i}) + \eta_{i+1} \cost(\lambda_{i}))\\
	& = &  \utility(\lambda_{i}) - \eta_{i}\cdot \cost(\lambda_{i})+ (\eta_{i}-\eta_{i+1}) \cdot \cost(\lambda_{i}) -\delta\cdot ( \utility(\lambda_{i}) + \eta_{i+1}\cdot \cost(\lambda_{i}))\\
	& =  &  (1-\delta)\cdot\underbrace{\big(\utility(\lambda_{i}) - \eta_{i}\cdot \cost(\lambda_{i})\big)}_{=\Phi_i} + \underbrace{(\eta_{i}-\eta_{i+1})}_{=\frac{\eps\cdot\mu}{2k}} \cdot \cost(\lambda_{i})  -\delta \cdot \underbrace{(\eta_i  + \eta_{i+1})}_{\leq 3}\cdot\cost(\lambda_i)\\
   	& \geq & (1-\delta) \cdot \Phi_i
    + \cost(\lambda_{i}) \cdot  \left( \frac{\eps \cdot \mu}{2 k} - 3\cdot \delta \right)\\
    & \stackrel{\big(\delta=\frac{\eps\cdot\mu}{6k}\big)}{=} &  (1-\delta) \cdot \Phi_i \\
    & \stackrel{(\text{I.H.})}{\geq} &  (1-\delta)^{i+1} \cdot \Phi_0.
  \end{eqnarray*}
  In the fourth line, $\eta_i+\eta_{i+1}\leq 3$ holds because by definition, $\eta_j\leq 1 + \frac{\eps\cdot\mu}{2}\leq \frac{3}{2}$ for all $j\in \set{0,\dots,k}$. Hence, \Cref{eq:potentialinduction} holds and thus the induction step is complete.
  
  We next show that for the initial potential $\Phi_0$, it holds that
  \begin{equation}\label{eq:initialpotential}
    \Phi_0 \geq \left(1-\frac{\eps}{2}\right)\cdot\big(\utility(\lambda_0)-\cost(\lambda_0)\big).
  \end{equation}

  We need to show that $\Phi_0 = \utility(\lambda_0) - \eta_0\cdot \cost(\lambda_0) \geq \left(1-\frac{\eps}{2}\right)\cdot\big(\utility(\lambda_0)-\cost(\lambda_0)\big)$, which is equivalent to showing that $ \eta_0\cdot \cost(\lambda_0) \leq \utility(\lambda_0) -\left(1-\frac{\eps}{2}\right)\cdot\big(\utility(\lambda_0)-\cost(\lambda_0)\big)
  =   \frac{\eps}{2} \cdot (\utility(\lambda_0)- \cost(\lambda_0))+\cost(\lambda_0)$.
  This is equivalent to showing that 
  $(\eta_0-1) \cdot \cost(\lambda_0) \leq  \frac{\eps}{2} \cdot (\utility(\lambda_0)- \cost(\lambda_0))$, which follows because
  \begin{eqnarray*}
    (\eta_0-1) \cdot \cost(\lambda_0) 
    & \leq  &  (\eta_0-1) \cdot (1-\mu) \cdot \utility(\lambda_0)\\
    & = & \left(\left(1+\frac{\eps \cdot \mu}{2}\right)-1\right) \cdot  (1-\mu) \cdot \utility(\lambda_0)\\
    & \leq &  \frac{\eps \cdot \mu}{2}\cdot  \utility(\lambda_0)\\
    & \leq &  \frac{\eps}{2}\cdot  \big(\utility(\lambda_0)- \cost(\lambda_0)\big).
  \end{eqnarray*}
  The first and last inequality follow since we have that $\utility(\lambda_0)-\cost(\lambda_0)\geq \mu\utility(\lambda_0)$.
  The claim of the lemma now follows because
  \begin{eqnarray*}
    \utility(\lambda_k) - \cost(\lambda_k) & = & \Phi_k\ \stackrel{\eqref{eq:potentialinduction}}{\geq}\ 
    (1-\delta)^k\cdot \Phi_0\\
    & \stackrel{\eqref{eq:initialpotential}}{\geq} &
    (1-\delta)^k\cdot\left(1-\frac{\eps}{2}\right)\cdot \big(\utility(\lambda_0)-\cost(\lambda_0)\big)\\ & \geq & (1-\eps)\cdot\big(\utility(\lambda_0)-\cost(\lambda_0)\big).
  \end{eqnarray*}
  The last inequality follows because $\delta=\frac{\eps\mu}{8k}\leq \frac{\eps}{2k}$ and thus $(1-\delta)^k\geq 1-\frac{\eps}{2}$.
\end{proof}

\subsubsection{Preprocessing of Fractional Label Assignments}
\label{sec:preprocessingassignment}

\Cref{alg:basicrounding} requires the given fractional solution to be $1/2^k$-integral for some positive integer $k$. It is often more natural to write down a fractional solution by using general real fractional values. The following simple technical lemma shows how a given fractional label assignment can be rounded to a $1/2^k$-fractional label assignment.

\begin{lemma}\label{lemma:fractionalrounding}
  Let $G=(V,E)$ be a multigraph, which is equipped with utility and cost functions $\utility$ and $\cost$ and with a fractional label assignment $\lambda$ such that for all $v\in V$ and all $\alpha\in \Labels$, $\lambda_\alpha(v)=0$ or $\lambda_\alpha(v)\geq \lambda_{\min}$ for some given value $\lambda_{\min}$. Let $\eps\in [0,1]$ and $\mu\in(0,1]$ be two parameters. If $\utility(\lambda)-\cost(\lambda)\geq \mu\utility(\lambda)$, the nodes can internally (i.e., without communication) compute a new $1/2^k$-integral fractional label assignment $\lambda'$ for some integer $k$ such that $2^k=O\big(\frac{1}{\eps\cdot\mu\cdot \lambda_{\min}}\big)$,
  \[
  \utility(\lambda')-\cost(\lambda') \geq (1-\eps)\cdot \big(\utility(\lambda) - \cost(\lambda)\big),
  \quad\text{and}\quad
  \utility(\lambda')-\cost(\lambda') \geq \frac{\mu}{2}\cdot\utility(\lambda').
  \]
\end{lemma}
\begin{proof}
  Clearly, for a given integer $k\geq 0$, each node $v\in V$ can compute new values $\lambda_\alpha'(v)$ such that for every $\alpha\in\Labels$, $\lambda_\alpha'(v)=0$ if $\lambda_\alpha(v)=0$ and otherwise, $\lambda_\alpha'(v)$ is an integer multiple of $2^{-k}$ such that $|\lambda_\alpha(v) - \lambda_\alpha'(v)|\leq 2^{-k}$. Node $v$ can just round up or down each $\lambda_\alpha(v)$ to the nearest larger or smaller integer multiple of $2^{-k}$ such that the total sum of all fractional values of $v$ remains $1$. For all $v\in V$ and $\alpha \in \Labels$, we then have
  \begin{equation}\label{eq:roundedfractvalues}
    \left(1-\frac{1}{\lambda_{\min}\cdot 2^k}\right)\cdot \lambda_\alpha(v)
    \ \leq\ \lambda_\alpha'(v)\ \leq\ 
    \left(1+\frac{1}{\lambda_{\min}\cdot 2^k}\right)\cdot \lambda_\alpha(v).
  \end{equation}
  Recall that utility and cost are computed as a sum over the individual edge utilities and costs. By the definition of the utility and cost of an edge for a fractional assignment (\Cref{eq:fractionalutilitycost}), for an edge $e\in E$ with $V(e)=\set{u,v}$, we have
  \begin{eqnarray}
    \utility(e, \lambda')
    & \stackrel{\eqref{eq:fractionalutilitycost}}{\geq} & 
             \min_{(\alpha,\beta)\in\Labels^2}\frac{\lambda_{\alpha}'(u)\cdot\lambda_{\beta}'(v)}{\lambda_{\alpha}(u)\cdot\lambda_{\beta}(v)}\cdot
             \utility(e, \lambda)
             \ \stackrel{\eqref{eq:roundedfractvalues}}{\geq}
             \left(1-\frac{1}{\lambda_{\min}\cdot 2^k}\right)^2\cdot \utility(e, \lambda),\label{eq:edgeutilitylower}\\
    \utility(e, \lambda')
    & \stackrel{\eqref{eq:fractionalutilitycost}}{\leq} & 
             \max_{(\alpha,\beta)\in\Labels^2}\frac{\lambda_{\alpha}'(u)\cdot\lambda_{\beta}'(v)}{\lambda_{\alpha}(u)\cdot\lambda_{\beta}(v)}\cdot
             \utility(e, \lambda)
             \ \stackrel{\eqref{eq:roundedfractvalues}}{\leq}
             \left(1+\frac{1}{\lambda_{\min}\cdot 2^k}\right)^2\cdot \utility(e, \lambda). \label{eq:edgeutilityupper}
  \end{eqnarray}
  We choose $k$ as the smallest integer such that $2^k\geq \frac{9}{\eps\cdot\mu\cdot\lambda_{\min}}$. We then get
  \begin{eqnarray}
    \utility(\lambda')
    & = & 
          \sum_{e\in E} \utility(e, \lambda')
          \ \geq\ 
          \left(1-\frac{\eps\cdot\mu}{9}\right)^2\cdot\sum_{e\in E} \utility(e, \lambda)
          \ \geq\  \left(1-\frac{\eps\cdot\mu}{4}\right)\utility(\lambda), \label{eq:utilitylower}\\
    \utility(\lambda')
    & = & 
          \sum_{e\in E} \utility(e, \lambda')
          \ \leq\ 
          \left(1+\frac{\eps\cdot\mu}{9}\right)^2\cdot\sum_{e\in E} \utility(e, \lambda)
          \stackrel{(\eps\mu\leq 1)}{\leq} \left(1+\frac{\eps\cdot\mu}{4}\right)\utility(\lambda).\label{eq:utilityupper}
  \end{eqnarray}
  Clearly, Equations \eqref{eq:edgeutilitylower}, \eqref{eq:edgeutilityupper}, \eqref{eq:utilitylower}, and \eqref{eq:utilityupper} also hold for the cost function $\cost(\cdot)$ in exactly the same way. We therefore obtain
  \begin{eqnarray*}
    \utility(\lambda') - \cost(\lambda') 
    & \stackrel{(\eqref{eq:utilityupper}, \eqref{eq:utilitylower}\text{ for }\cost(\cdot))}{\geq} &
             \utility(\lambda) - \cost(\lambda) - \frac{\eps \mu}{4}\cdot\big(\utility(\lambda) + \cost(\lambda)\big)\\
    & \stackrel{(\cost(\lambda)\leq\utility(\lambda))}{\geq} &
             \utility(\lambda) - \cost(\lambda) - \frac{\eps \mu}{2}\cdot\utility(\lambda)\\
    & \stackrel{(\utility(\lambda)\leq \frac{1}{\mu}\cdot(\utility(\lambda)-\cost(\lambda))}{\geq} & 
             (1-\eps)\cdot\big(\utility(\lambda)-\cost(\lambda)\big).
  \end{eqnarray*}
  Hence, the first lower bound on $\utility(\lambda')-\cost(\lambda')$ claimed by the lemma holds. We further get
  \begin{eqnarray*}
    \utility(\lambda') - \cost(\lambda')
    & \stackrel{(\eqref{eq:utilityupper}\text{ for }\cost(\cdot))}{\geq} &
             \utility(\lambda') - \left(1+\frac{\eps\mu}{4}\right)\cdot\cost(\lambda)\\
    & \stackrel{(\cost(\lambda)\leq (1-\mu)\utility(\lambda))}{\geq} &
             \utility(\lambda') - \left(1+\frac{\eps\mu}{4}\right)\cdot(1-\mu)\cdot\utility(\lambda)\\
    & \stackrel{\eqref{eq:utilityupper}}{\geq} &
             \utility(\lambda') - \left(1+\frac{\eps\mu}{4}\right)^2\cdot(1-\mu)\cdot\utility(\lambda')\\
    & = & 
          \left[\left(1-\frac{\eps}{2}\right)\cdot\mu + \left(\frac{\eps}{2}-\frac{\eps^2}{16}\right)\cdot\mu^2 + \frac{\eps^2\mu^3}{16}\right]
          \cdot \utility(\lambda')
    \ \stackrel{(\eps\leq 1)}{\geq} \
          \frac{\mu}{2}\cdot \utility(\lambda').
  \end{eqnarray*}
  Hence, also the second claim of the lemma holds.
\end{proof}

\subsubsection{Node Utilities and Costs}
\label{sec:nodeutilitiesandcosts}

In some applications, a part of the cost and/or utility only depends on the labeling of individual nodes and not on the combined labeling of neighboring nodes. In such cases, it is natural to extend the cost and utility functions and also define the utility/cost of single nodes. The following lemma shows that node utilities and costs can be easily incorporated into our rounding framework.

\begin{lemma}\label{lemma:nodeutilitycost}
  Let $G=(V,E)$ be a multigraph, which is equipped with utility and cost functions $\utility$ and $\cost$. Assume that, the functions $\utility$ and $\cost$ are extended to also assign a non-negative utility and cost to each node $v\in V$, where the utility and cost of a node is a function of the node's label. \Cref{alg:basicroundingstep,alg:basicrounding} can be adapted to incorporate node utilities and costs so that \Cref{lemma:basicroundingstep,lemma:basicrounding} hold accordingly.
\end{lemma}
\begin{proof}
  Essentially, a node $v\in V$ can distribute its utility and cost arbitrarily among its edges and then use \Cref{alg:basicroundingstep,alg:basicrounding} to do the rounding. To make sure that handling node utilities and costs is independent of the concrete communication model, we do this as follows. For every node $v\in V$, we define a virtual dummy node $\tilde{v}$, which is simulated by node $v$. Nodes $v$ and $\tilde{v}$ are connected by a virtual edge $\tilde{e}_v$ and $\tilde{e}_v$ is the only edge of node $\tilde{v}$. The utility and cost of edge $\tilde{e}_v$ is always equal to the utility and cost of node $v$. The edge utility and cost of $\tilde{e}_v$ therefore only depend on the (fractional) label assignment to $v$ and are independent of the label of $\tilde{v}$. We can therefore initially assign an arbitrary label $\alpha$ to node $\tilde{v}$ (i.e., if we start with a fractional label assignment $\lambda$, then $\lambda_\alpha(\tilde{v})=1$ and $\lambda_\beta(\tilde{v})=0$ for all $\beta\neq \alpha$. In this way, the fractional label assignment of $\tilde{v}$ is never changed during the rounding process. When computing the defective coloring, we always assign $\tilde{v}$ a color that is different from $v$'s color so that the edge $\tilde{e}_v$ is bichromatic. Note that we can clearly assume that the defective coloring always has at least $2$ colors as otherwise all edges would be monochromatic and the rounding could not satisfy any non-trivial guarantees.
\end{proof}

\subsection{Distributed Implementation of Basic Rounding Algorithm}
\label{sec:distributedrounding}

As long as every edge $e\in E$ in the multigraph $G=(V,E)$ is between nodes that are close in the underlying communication graph, implementing \Cref{alg:basicroundingstep,alg:basicrounding} is straightforward in the \LOCAL model. We however need to be more careful if the communication links in the underlying communication graph have limited capacity (e.g., when implementing the algorithms in the \CONGEST model).

\subsubsection[Direct Communication on $G$]{Direct Communication on \boldmath$G$}
\label{sec:directcommunication}

We first assume that communication is done directly between neighbors in $G$ and we explicitly analyze the necessary maximum message size for implementing our rounding algorithm.

\begin{lemma}\label{lemma:distributedonG}
  Let $G=(V,E)$ be a multigraph, which is equipped with utility and cost functions $\utility$ and $\cost$ and with a fractional label assignment $\lambda$ such that for every label $\alpha\in \Labels$ and every $v\in V$, if $\lambda_\alpha(v)>0$, then $\lambda_\alpha(v)\geq \lambda_{\min}$ for some given value $\lambda_{\min}\in(0,1]$. Further assume that $G$ is equipped with a proper $\zeta$-vertex coloring. Let $\eps\in [0,1]$ and $\mu\in(0,1]$ be two parameters. If $\utility(\lambda)-\cost(\lambda)\geq \mu\utility(\lambda)$ and if each node knows the utility and cost functions of its incident edges, there is a deterministic $O\big(\frac{1}{\eps \mu}\cdot \log^2\big(\frac{1}{\eps\mu\lambda_{\min}}\big)+ \log\big(\frac{1}{\eps\mu\lambda_{\min}}\big)\cdot \log^* \zeta\big)$-round distributed message passing algorithm on $G$ that computes an integral label assignment $\ell$ for which
  \[
  \utility(\ell)-\cost(\ell) \geq (1-\eps)\cdot \big(\utility(\lambda) - \cost(\lambda)\big).
  \]
  The algorithm uses messages of at size most $O\big(\min\set{|\Labels|, \log|\Labels|/(\eps\mu\lambda_{\min})}+\log\zeta\big)$ bits. If the given fractional label assignment $\lambda$ is $1/2^k$-integral for some integer $k\geq 1$, the round complexity of the algorithm is $O\big(\frac{k^2}{\eps\mu} + k \log^*\zeta\big)$ and the maximum message size is $O\big(\min\set{|\Labels|, 2^k\log|\Labels|}+\log\zeta\big)$ bits.
\end{lemma}
\begin{proof} The algorithm is \Cref{alg:basicrounding} which itself is simply a number of invocations of \Cref{alg:basicroundingstep}.

In the beginning, each node sends its initial fractional label assignments to all of its neighbors. This can be done using messages of size $O(|\Labels| k)$, where we use $O(k)$ bits to encode the fractional value of each label. Since the runtime of the overall algorithm exceeds $\Theta(k)$, without any loss in round complexity, we can perform this communication by using $k$ rounds and sending one message of size $O(|\Labels|)$ per round. When $2^k<|\Labels|$, a more efficient method would be to send $O(\log|\Labels|)$ bits for each of the at most $2^k$ non-zero fractional values, for a total of $O(2^{k} log |\Labels|)$ bits. These share the initial fractional label assignments in at most $k$ rounds using messages of size $O\big(\min\set{|\Labels|, 2^k\log|\Labels|}\big)$ bits. 

Then, we have $k$ rounding steps, where we invoke \Cref{alg:basicroundingstep}. Per step, in line 2 of \Cref{alg:basicroundingstep}, we first compute an average defective coloring with defect $\delta/6=\frac{\eps \mu}{36k}$, in time $O\big(\frac{k}{\eps\mu} + \log^*\zeta\big)$ using messages of $O(\log\zeta)$ bits, via the average weighted defective coloring algorithm of Ghaffari and Kuhn~\cite[Lemma 2.3]{GhaffariK21}. The produced coloring has $O\big(\frac{k}{\eps\mu}\big)$ colors and therefore the inner loop in line 4 of \Cref{alg:basicroundingstep} has $O\big(\frac{k}{\eps\mu}\big)$ iterations.

In each iteration, it is needed that each node $v$ knows the fractional label assignment of all its neighbors. In the first rounding step, this is given as nodes shared their initial fractional label assignments at the beginning, as we discussed above. After that, in each rounding step, the value of each label either decreases by a $ 2$ factor, increases by a $2$ factor, or remains as it was. Hence, each node can inform its neighbors about the update to its fractional label assignment using messages with $O(1)$ bits per label, i.e., $O(|\Labels|)$. Alternatively, as above, when $2^k<|\Labels|$, a more efficient method would be to send $O(\log|\Labels|)$ bits for each of the at most $2^k$ non-zero fractional values, for a total of $O(2^{k} \log |\Labels|)$ bits. Hence, the fractional value updates can be performed using messages of size at most $O\big(\min\set{|\Labels|, 2^k\log|\Labels|}\big)$ bits. Then, given that each node $v$ knows the fractional label assignment of its neighbors, it can compute $\phi_{v,\alpha}$ and $\theta_{v,\alpha}$, in lines 9 and 10 of \Cref{alg:basicroundingstep}. It thus can perform the rest of the computations of this rounding step locally, after which it can inform its neighbors about its fractional label assignment update using a message of size $O\big(\min\set{|\Labels|, 2^k\log|\Labels|}\big)$ bits, as discussed above. 
Since there are $k$ rounding steps and each takes $O\big(\frac{k}{\eps\mu} + \log^*\zeta\big)$ rounds, the claimed complexity follows. 
\end{proof}

\subsubsection[Communication on $G^2$]{Communication on \boldmath$G^2$}
\label{sec:d2communication}

In some cases, the graph on which one runs the rounding algorithm is not equal to the communication graph $G$. For example, for computing an MIS in \Cref{sec:MIS} or a set cover in \Cref{sec:setcover}, our cost functions has contributions by pairs of nodes that are not direct neighbors in $G$, but which have a common neighbor in $G$. We generically analyze such a communication setting here.

\paragraph{Communication Model.} Formally, we assume that $G=(V,E)$ is the communication graph and that the rounding is done on a virtual multigraph $H=(V_H,E_H)$ with the following properties. 

\begin{definition}[d2-Multigraph]\label{def:d2multigraph}
  A \emph{d2-multigraph} is a multigraph $H=(V_H,E_H)$ that is simulated on top of an underlying communication graph $G=(V,E)$ by a distributed message-passing algorithm on $G$. The nodes of $H$ are a subset of the nodes of $G$, i.e., $V_H\subseteq V$. The edge set $E_H$ consists of two kinds of edges, \emph{physical edges} and \emph{virtual edges}. Physical edges in $E_H$ are edges between direct neighbors in $G$. For each physical edge in $e\in E_H$ with $V(e)=\set{u,v}$, both nodes $u$ and $v$ know about $e$. Virtual edges in $E_H$ are edges between two nodes $u,v\in V_H$ for which there is a common neighbor in $G$.  Each virtual edge between two nodes $u$ and $v$ is known and managed by a common neighbor $u$ and $v$. We define a function $\xi:E_H\to V$ to refer to the node managing the virtual edge.
\end{definition}

We will always assume that the nodes know all relevant information about the edges of $H$ they manage.
For example, if we run an algorithm on a graph $H$ with edge weights $w:E_H\to\Rp$, then for every physical edge $e$ between $u$ and $v$, we assume that $u$ and $v$ know $w(e)$ and if $e$ is a virtual edge between $u$ and $v$, the common neighbor $\xi(e)$ of $u$ and $v$ knows $w(e)$. Similarly, when running an instance of the rounding algorithm of \Cref{sec:basicrounding} on $H$, the node(s) simulating an edge $e\in E_H$ know the utility and cost functions for edge $e$. Note also that nodes $u$ and $v$ might not be aware of virtual edges between them and that $H$ might have several edges between two nodes $u$ and $v$. If $u$ and $v$ are neighbors in $G$, there can be one physical edge between $u$ and $v$. If $u$ and $v$ have common neighbors in $G$, there can potentially be a virtual edge between $u$ and $v$ for each common neighbor of $u$ and $v$ in $G$. We could also allow multiple physical edges or multiple virtual edges for the same common neighbor in $G$. However such edges can typically also easily be aggregated into a single physical edge and one virtual edge per common neighbor in $G$. In the following, we assume standard synchronous, distributed message-passing algorithms on $G$, where in each round, a message of a certain size can be exchanged over each edge in $G$ and we show how to utilize this communication to implement \Cref{alg:basicroundingstep,alg:basicrounding} on the virtual multigraph $H$.

\paragraph{Computation of a Defective Coloring.}
 As an important part of \Cref{alg:basicroundingstep}, the nodes have to compute a weighted average defective coloring of $H$. We first analyze how efficiently such a defective coloring can be computed. In the following, we say that an edge weight function $w$ is polynomially bounded in $q$ if there is a value $Q\leq \poly(m)$ such that for every edge $e$, $w(e)=0$ or $w(e)\in[1/Q, Q]$. Note that the following lemma statement contains a somewhat artificial technical condition $\delta > 2^{-O(\sqrt{\log n})}$ on the relative defect parameter $\delta$. This condition is added to make the bounds look slightly nicer and cleaner. At the cost of a minor additional term in the round complexity, the condition could also be dropped. Also, the condition could also be strengthened to $\delta>2^{-O(\log^\vartheta n)}$ for any constant $\vartheta<1$ and in fact even for some $\vartheta=1-o(1)$. Note than in our applications the value of $\delta$ is typically $1/\poly\log \Delta$ or $1/\poly\log n$.

\begin{lemma}\label{lemma:d2defective}
  Let $H=(V_H,E_H)$  be an $n$-node d2-multigraph with edge weights $w:E_H\to \Rp$, assume that  $G=(V,E)$ is the underlying communication graph, and let $\Delta$ denote the maximum degree of $G$. Assume that the edge weights $w(e)$ are polynomially bounded in $q\geq \Delta$ and let $\delta > 2^{-O(\sqrt{\log n})}$ be a parameter. A weighted average $\delta$-relative defective $O(1/\delta)$-coloring of $H$ can be computed deterministically $O\big(\frac{\log\log q}{\delta}+\log^* n\big)$ \CONGEST rounds on $G$.
\end{lemma}
\begin{proof}
  As a first step, we describe how to compute a weighted $\delta$-relative defective $O(1/\delta^2)$-coloring. For this, we adapt an algorithm that has been described in \cite{Kuhn2009WeakColoring} and that has been adapted to the weighted case in \cite{KawarabayashiS18}. The algorithm is based on a classic $O(\log^* n)$-round $O(\Delta^2)$-coloring algorithm of Linial~\cite{linial1987LOCAL}. The algorithm starts with an initial (defective) vertex coloring (in our case given by the unique $O(\log n)$-bit IDs) and consists of $O(\log^* n)$ consecutive steps. In each step, the coloring is improved in the following way. As a function of its current color (and without communicating), each node computes a set of candidate colors for the next coloring. Each node then picks a candidate color that (approximately) minimizes the total weight to neighbors (of different initial colors) that also have this candidate color. If in each step of the algorithm, a node chooses a candidate color for which the total weight to neighbors with the same candidate color is within a factor $2$ of the minimum, the relative defect in the end also only grows by a factor $2$. In order to determine its new color, in each step, each node, therefore for each of its candidate colors needs to approximately learn the total weight of its edges to neighboring nodes that also have this candidate color.

  Before making the algorithm more concrete, let us focus on the communication required from a single step. First note that since node colors can always be represented by $O(\log n)$ bits (we start with $O(\log n)$-bit IDs), in a single round on the communication graph $G$, each node can learn the colors of all its $G$-neighbors. This implies that each node knows the candidate colors of all its $G$-neighbors. Assume now that in a given step, each node $v$ has at most $p$ candidate colors. Therefore, for each of its at most $p$ candidate colors $z$, each node $v\in V_H$ needs to learn an estimate of the total weight of its edges in $E_H$ to nodes that also have $z$ as a candidate color. Because $v$ knows the candidate colors of its $G$-neighbors, it exactly knows the total weight of the physical edges to neighbors with $z$ as a candidate color. Let us consider a virtual edge $e$ between node $v$ and some other node $u$. The common $G$-neighbor $\xi(e)$ of $u$ and $v$ knows the weight of $e$ and $\xi(e)$ also knows that candidate sets of both $v$ and $u$ and potentially also of other virtual edges of $v$ for which it is responsible. For each candidate color $z$ of $v$, node $\xi(e)$ can aggregate the total weight of $v$'s virtual edges for which $\xi(e)$ is responsible and which go to nodes that also have candidate color $z$. Because the weights are polynomially bounded in $q\geq \Delta$, the total weight of those virtual edges is either $0$ or a value between $1/\poly(q)$ and $\poly(q)$. Since $v$ only needs to learn the weight to other nodes with candidate color $z$ up to a factor $2$, it suffices to communicate one of the $O(\log q)$ different possible values over each edge of $G$. Therefore, for each candidate color, the algorithm has to send $O(\log\log q)$ bits over each $G$-edge. A defective coloring step with at most $p$ candidate colors can therefore be carried out in $O\big(1+\frac{p\log\log q}{\log n}\big)$ rounds in the \CONGEST model.
  
  Let us now look at the details of the individual steps of the defective $O(1/\delta^2)$-coloring algorithm. In \cite{linial1987LOCAL} (also cf.\ \cite{erdos82,MausTonoyan20}), it is shown that for positive integer parameters $s$ and $N$ and a sufficiently large constant $c>0$, there are $N$ sets $S_1,\dots,S_N$ such that
  \[
  \forall i\in [N]: S_i\in [s^2\cdot \tau],\
  |S_i| \leq s\cdot \tau,\ \text{and}\ 
  \forall i\neq j: |S_i\cap S_j| \leq \tau,\ \text{where}\ 
  \tau = c\cdot \min\set{\log N, \log_s^2 N}.
  \]
  Note that if we have a proper $N$-coloring of the nodes, and when choosing $s=1/\delta$, we can compute a weighted $\delta$-defective $\tau/\delta^2$-coloring as follows. Each node $v$ of initial color $i\in [N]$, picks set $S_i$ as its set of candidate colors. Because the candidate sets of nodes of different colors (and thus of neighbors) intersect in at most $\tau$ colors, each neighbor of $v$ only shares at most a $\delta$-fraction of $v$'s candidate colors. Therefore, there must be a color among $v$'s candidate colors such that the total weight of $v$'s edges to neighbors that also have this candidate color is at most a $\delta$-fraction of the total weight of all of $v$'s edges. If $v$ only knows a constant factor approximation to the total weight of its edges for each candidate color, taking the color with the best estimate still guarantees a relative defect of $\delta$ if we choose $s=c'/\delta$ for a sufficiently large constant $c'$. In the following, we will implicitly assume that all $s$-values are multiplied by a sufficiently large constant factor to make sure that a constant-factor estimate for the total edge weight per candidate color is sufficient. If the initial $N$-coloring is already a defective coloring with some relative defect $\delta_0$, then the argument holds for the bichromatic edges (w.r.t.\ to the initial $N$-coloring) and the new defective coloring becomes $(\delta_0+\delta)$-relative defective.
  
  In \cite{Kuhn2009WeakColoring,KawarabayashiS18}, it is shown that when starting with an $N$-coloring and running $t=O(\log^* n)$ steps $i=1,\dots,t$ of the above algorithm with $s_i=2^{t-i+1}/\delta$, the overall relative defect in the end is $1/s_1+\cdots+s_{t}<\delta$ and the final number of colors is $O(s_{t-1}^2)=O(1/\delta^2)$. In the very first step, when starting with the initial $O(\log n)$-bit IDs, we have $\tau=\Theta(\log n)$ and in subsequent steps, $\tau$ will be exponentially smaller. As $\tau$ goes linearly into the number of candidate colors and thus the number of bits that have to be transmitted, we would like $s$ to be as small as possible in the very first step. We therefore add a step $0$, where we set $s_0=4/\delta$ and we multiply all the $s_i$ for $i>0$ by $4$. Let $\tau_i$ be the value of $\tau$ in step $i$ and let $p_i$ be the maximum number of candidate colors per node in step $i$. The overall relative defect is then less than $\delta/2$. In step $0$, the number of candidate colors per node is then at most $p_0=s_0\cdot\tau_0=O\big(\frac{\log n}{\delta}\big)$. In subsequent steps $i=1,\dots,t$, the number of candidate colors is at most $p_i \leq s_i\cdot \tau_i \leq s_1\cdot \tau_1$. After the step $0$, the number of colors is $O\big(\frac{\log n}{\delta^2}\big)$ and we therefore have $\tau_1 = O(\log\log n + \log(1/\delta))$. We further have $s_1=2^{O(\log^* n)}/\delta$. The total number of \CONGEST rounds on $G$ to run all $O(\log^* n)$ steps is therefore
  \[
  O\left(\log^* n + \frac{(s_0\tau_0 + s_1\tau_1\cdot\log^* n)\cdot \log\log q}{\log n}\right)
  = O\left(\log^* n + \frac{\log\log q}{\delta}\right).
  \]
  In the above equation, we use that $s_1\tau_1\log^* n = O(1/\delta)\cdot 2^{O(\log^* n)}\cdot O(\log^* n)\cdot O(\log\log n + \log(1/\delta)) = O(\log(n)/\delta)$. Note that we have $\delta \geq 2^{-O(\sqrt{\log n})}$ and thus $\log(1/\delta)=O(\sqrt{\log n})$.
  
  We now have a weighted $\delta/2$-relative defective coloring with $O(1/\delta^2)$ colors and we still need to turn this into a $\delta$-relative average defective coloring with $O(1/\delta)$ colors. For this, we adapt a distributed coloring algorithm that was presented in \cite{BEG18}. At the core of the algorithm of \cite{BEG18} is the following idea. For any prime $p$, there are $p^2$ orderings of the numbers $1,\dots,p$ such that any two orderings coincide in at most $1$ place. We now choose $p\geq 8/\delta$ and such that $p^2$ is larger than the initial number of $O(1/\delta^2)$ colors. In this way, nodes of a different initial color pick different orderings of the numbers $1,\dots,p$. The algorithm now consists of $p$ steps. For a node $v$, let $z_{v,1},\dots,z_{v,p}$ be the ordering of colors $1,\dots,p$ of node $v$. In step $i$, $v$ tries to take color $z_{v,i}$ and it takes the color if less than a $\delta/4$-fraction of the weight of its edges go to nodes that have already committed to taking color $z_{v,i}$ in previous steps or that are currently also trying to take color $z_{v,i}$. Note that if every node takes a color, the weighted average relative defect is at most $\delta/2$. This can be seen by orienting every edge towards the node that first commits to taking a color and orienting arbitrarily in case both nodes take the color in the same step. Because every node $v$ can conflict with each neighbor at most twice (once when trying the same color and once when $u$ has already committed to a color and $v$ tries this color), the total weight of the monochromatic outgoing edges of a node $v$ is at most at $\delta/4$-factor of the total total weight of all edges of $v$. The total weight of all monochromatic edges is therefore a $\delta/4$-factor of $2$ times the total weight of all edges. For further details, we refer to \cite{BEG18}.
  
  It remains to show that each step of this algorithm can be implemented efficiently on the d2-multigraph $H$ in the \CONGEST model on $G$. To implement a step, a node $v$ needs to know the total weight of all its edges to neighbors that have already committed to the color that $v$ currently tries or that are currently trying the same color. Again, the algorithm can easily be adapted such that it suffices for each node $v$ to learn a constant approximation of this total edge weight. This now is the same communication problem that we had in the first part of the algorithm for a single candidate color. We therefore need to communicate $O(\log\log q)$ bits over each edge of $G$. This can clearly be done in  $O\big(1+\frac{\log\log q}{\log n}\big)=O(\log\log q)$ rounds in the \CONGEST model. The total time for the second step is therefore $O(p\log\log q)=O\big(\frac{\log\log q}{\delta}\big)$ rounds. This concludes the proof.
\end{proof}

\paragraph{Implementing the Rounding Algorithm.} We next show how efficiently \Cref{alg:basicroundingstep,alg:basicrounding} can be implemented on a d2-multigraph of an underlying communication graph, when using \CONGEST algorithms on the communication graph. For the following lemma, recall that a rounding instance is called polynomially bounded in some value $q$ if all utility and cost values are either $0$ or bounded between $1/\poly(q)$ and $\poly(q)$ and if all fractional node values are either $0$ or lower bounded by $1/\poly(q)$.

\begin{lemma}\label{lemma:d2rounding}
  Let $H=(V_H,E_H)$ be a d2-multigraph of an underlying communication graph $G=(V,E)$ of maximum degree $\Delta$. Assume that $H$ is equipped with utility and cost functions $\utility(\cdot)$ and $\cost(\cdot)$ (with label set $\Labels$) and with a fractional label assignment $\lambda$. Further assume that the given rounding instance is polynomially bounded in a parameter $q \leq n$. Then for every constant $c>0$ and every $\eps,\mu>\max\set{1/q^c, 2^{-c\sqrt{\log n}}}$, if $\utility(\lambda)-\cost(\lambda)>\mu\utility(\lambda)$, there is a deterministic \CONGEST algorithm on $G$ to compute an integral label assignment $\ell$ for which $\utility(\ell)-\cost(\ell)\geq (1-\eps)\cdot\big(\utility(\lambda)-\cost(\lambda)\big)$ and such that the round complexity of the algorithm is
  \[
    O\left(\frac{\log^2 q}{\eps\cdot\mu}\cdot\left(\frac{|\Labels| \log(q\Delta)}{\log n} + \log\log q \right)+\log q\cdot\log^* n \right).
  \]
\end{lemma}
\begin{proof} 
  First note that by applying \Cref{lemma:fractionalrounding}, we can turn the initial fractional label assignment into a $2^{-k}$-integral fractional label assignment with $k = O\big(\frac{\log q}{\eps\cdot\mu}\big)$ (recall that the initial fractional assignment is polynomially bounded in $q$). Since we assume that $\eps,\mu\geq 1/\poly(q)$, we have $O\big(\frac{\log q}{\eps\cdot\mu}\big)=O(\log q)$.

  We assume that throughout the algorithm, each node of the communication graph $G$ keeps track of the fractional assignment of all neighbors in $G$. To achieve this, initially, each node of $G$ needs to learn $O(\log q)$ bits per label. This requires $O\big(1+ \frac{|\Labels|\log q}{\log n}\big)$ rounds in the \CONGEST model on $G$. To maintain the fractional values of all $G$-neighbors, each node has to learn $O(1)$ bits per label from each neighbor after each change to the fractional assignment (i.e., after each iteration of the outer for loop in \Cref{alg:basicroundingstep}). We will account for the time for doing this, when analyzing the cost of such a step in the following.
  
  \Cref{alg:basicrounding} consists of $O\big(\frac{\log q}{\eps\cdot\mu}\big)=O(\log q)$ runs of \Cref{alg:basicroundingstep}. We therefore need to understand the time required to run a single instance of \Cref{alg:basicroundingstep}. At the beginning of \Cref{alg:basicroundingstep}, we compute a weighted $\delta$-relative defective coloring. Note that since all utilities, costs, and fractional values are polynomially bounded in $q$, also the edge weights of this defective coloring instance is polynomially bounded in $q$. Further, because nodes in $G$ know the fractional assignments of their $G$-neighbors, each node knows the weights of all edges in $H$ for which it is responsible (i.e., nodes know the weights of their physical edges and of the virtual edges for which they are the responsible middle node). By \Cref{lemma:d2defective}, we can therefore compute such a weighted $\delta$-relative defective coloring in $O\big(\frac{\log\log q}{\delta} + \log^* n\big)$ rounds and with $O(1/\delta)$ colors. By the description of \Cref{alg:basicrounding}, we have $\delta=\Theta\big(\frac{\eps \mu}{\log q}\big)$ and thus, the time for computing the defective coloring is $O\big(\frac{\log q \cdot \log\log q}{\eps\cdot \mu} + \log^* n\big)$.
  
  Now, the algorithm iterates over the $O(1/\delta)$ colors. In each of these iterations, each node $v$ for each label $\alpha\in \Labels$ needs to learn an estimate of $\phi_{v,\alpha}$ as defined in lines 9--11 of \Cref{alg:basicroundingstep}. By knowing the fractional assignment of all $G$-neighbors, each node $v$ exactly know the contribution of its physical edges to $\phi_{v,\alpha}$. However, this information is not sufficient to know the contribution of the virtual edges to $\phi_{v,\alpha}$. Node $v$ has to learn this information from its neigbors that are responsible for those edges. Consider some neighbor $w$ of $v$ and let $E_{w,v}$ be the set of virtual edges of $v$ for which $w$ is the responsible middle node. Further, let $E'_{w,v}$ be the subset of those edge that are bichromatic w.r.t.\ to the current defective coloring. Node $w$ knows the fractional assignment of both nodes for all edges in $E'_{w,v}$ and it can therefore compute the exact contribution of the edges in $E'_{w,v}$ to $\phi_{v,\alpha}$. Note that the number of edges in $E'_{w,v}$ is at most $O(\Delta)$. The contribution of edges in $E'_{w,v}$ to $\phi_{v,\alpha}$ is therefore upper bounded by $O(\Delta)\cdot\poly(q)=\poly(q\Delta)$. Node $v$ however only needs to learn an estimate of $\phi_{v,\alpha}$ that is accurate up $\Theta(\delta\cdot\theta_{v,\alpha})$. The value of $\theta_{v,\alpha}$ is either $0$ (in this case $\phi_{v,\alpha}$ and also the contribution of $E_{w,v}'$ to $\phi_{v,\alpha}$ are also $0$) or it is lower bounded by $1/\poly(q)$. If the contribution of $E_{w,v}'$ to $\phi_{v,\alpha}$ is non-zero, $v$ therefore has to learn an estimate of it that is accurate up to an additive term that is lower bounded by $1/\poly(q)$. We can therefore discretize so that the approximate contribution of edges in $E_{w,v}'$ to $\phi_{v,\alpha}$ can only take $\poly(q\Delta)$ different values. Hence, $w$ only has to send $O(\log (q\Delta))$ bits for each possible label $\alpha\in \Labels$ to $v$. To accomplish line 11 of \Cref{alg:basicroundingstep}, we therefore have to send at most $O(|\Labels|\cdot\log (q\Delta))$ bits over each edge of $G$. Updating the knowledge of the fractional assignment of $G$-neighbors can then be done by exchanging another $O(|\Labels|)$ bits over each edge of $G$. The total time for running one loop (iterating over one color of the defective coloring) in \Cref{alg:basicroundingstep} is therefore $O\big(1 + \frac{|\Labels|\cdot\log (q\Delta)}{\log n}\big)$ in the \CONGEST model on $G$. Combining everything, we therefore obtain a time complexity of
  \[
  O\left(\frac{\log q\cdot\log\log q}{\eps\cdot \mu} + \log^* n + \frac{\log q}{\eps\cdot \mu}\cdot\left(1+\frac{|\Labels|\cdot\log (q\Delta)}{\log n}\right)\right)
  \]
  for each instance \Cref{alg:basicroundingstep}. By multiplying with $O(\log q)$, i.e., with the number of times we run  \Cref{alg:basicroundingstep}, we obtain the bound that is claimed by the lemma.
\end{proof}

\paragraph{Remark on a slight generalization} For some applications, we use a slight generalization of the d2-multigraph definition, provided above in \Cref{def:d2multigraph}, where we allow the managers of virtual edges to be some node further away, so long as we can provide some communication primitives, as we describe next.
\begin{definition}[long-range d2-Multigraph]\label{def:long-range-d2multigraph}
  A long-range \emph{d2-multigraph} is a multigraph $H=(V_H,E_H)$ that is simulated on top of an underlying communication graph $G=(V,E)$ by a distributed message-passing algorithm on $G$. The nodes of $H$ are a subset of the nodes of $G$, i.e., $V_H\subseteq V$. The edge set $E_H$ consists of two kinds of edges, \emph{physical edges} and \emph{virtual edges}. Physical edges in $E_H$ are edges between direct neighbors in $G$. For each physical edge in $e\in E_H$ with $V(e)=\set{u,v}$, both nodes $u$ and $v$ know about $e$. Virtual edges in $E_H$ are edges between two nodes $u,v\in V_H$, and for each such virtual edge, there is a manager node $w$ which knows about this edge. We define a function $\xi:E_H\to V$ to refer to the node managing the virtual edge. 
  
  We next describe the assumed communication primitives. Let $M(v)$ be the set of nodes $w$ who manage virtual edges that include $v$. We assume $T$-round primitives that provide the following:   (1) each node $v$ can send one $O(\log n)$-bit message that is delivered to all nodes in $M(v)$ in $T$ rounds; (2) given $O(\log n)$-bit messages prepared at nodes $M(v)$ specific for node $v$, node $v$ can receive an aggregation of these messages, e.g., the summation of the values, in $T$ rounds.
\end{definition}
We get an analogue of \Cref{lemma:d2rounding} for rounding in long-range d2-multigraphs:

\begin{lemma}\label{lemma:long-range-d2rounding}
  Let $H=(V_H,E_H)$ be a long-range-d2-multigraph of an underlying communication graph $G=(V,E)$ of maximum degree $\Delta$, where the communication primitives have round complexity $T$. Assume that $H$ is equipped with utility and cost functions $\utility(\cdot)$ and $\cost(\cdot)$ (with label set $\Labels$) and with a fractional label assignment $\lambda$. Further assume that the given rounding instance is polynomially bounded in a parameter $q \leq n$. Then for every constant $c>0$ and every $\eps,\mu>\max\set{1/q^c, 2^{-c\sqrt{\log n}}}$, if $\utility(\lambda)-\cost(\lambda)>\mu\utility(\lambda)$, there is a deterministic \CONGEST algorithm on $G$ to compute an integral label assignment $\ell$ for which $\utility(\ell)-\cost(\ell)\geq (1-\eps)\cdot\big(\utility(\lambda)-\cost(\lambda)\big)$ and such that the round complexity of the algorithm is
  \[
    T \cdot O\left(\frac{\log^2 q}{\eps\cdot\mu}\cdot\left(\frac{|\Labels| \log(q\Delta)}{\log n} + \log\log q \right)+\log q\cdot\log^* n \right).
  \]
\end{lemma}
\begin{proof}[Proof Sketch]
One change is that we need a version of \Cref{lemma:d2defective} for long-range d2-multigraph, but that follows by the same proof, with only a $T$ factor slow down, because each virtual edge manager needs to learn the current color of its endpoints, and then informs each endpoint about the total weight of edges to endpoints that have chosen each particular color (in the collection this endpoint is interested in). The former clearly fits the first communication primitive of long-range d2 multi-graphs, sending from $v$ to $M(v)$, and the latter fits the second primitive, $v$ receiving an aggregation of messages prepared by $M(v)$ for $v$.

Then, the rest of the proof is identical to that of \Cref{lemma:d2rounding}, again with only a $T$ factor slow-down, where $T$ denotes the round complexity of the corresponding two communication primitives. Namely, in our rounding, each virtual edge manager needs to receive the current fractional label of each of the endpoints. Then it prepares a message for each endpoint, and each endpoint needs to receive the summation of these messages.  Again, the former fits the first communication primitive of long-range d2 multi-graphs, sending from $v$ to $M(v)$, and the latter fits the second primitive, $v$ receiving an aggregation of messages prepared by $M(v)$ for $v$.
\end{proof}


%% file: MIS.tex
\section{Maximal Independent Set}
\label{sec:MIS}

In this section, we describe a deterministic distributed algorithm that computes an MIS in $O(\log^2 \Delta \,\cdot \,\log n)$ rounds of the \LOCAL model. Furthermore, we explain how a variant of this MIS algorithm can be implemented in $\tilde{O}(\log^2 \Delta \,\cdot \, \log n)$ rounds of the \CONGEST model.

\begin{theorem}
\label{thm:mis}
There is a deterministic distributed algorithm that computes a maximal independent set in $O(\log^2 \Delta  \,\cdot \,\log n)$ rounds of the \LOCAL model, and in $O(\log^2 \Delta \,\cdot \, \log\log \Delta  \,\cdot \,\log n)$ rounds of the \CONGEST model, in any $n$-node graph with maximum degree at most $\Delta$.
\end{theorem}

The rest of the section is devoted to the proof of \cref{thm:mis}. 
We first recall Luby's classic algorithm\cite{luby86}, which in each iteration chooses an independent set of nodes such that, when we add them to the output and remove them from the graph along with their neighbors, in expectation half of the edges of the graph are removed. 
We discuss the randomized analysis of this algorithm and explain how we can formulate it in the framework of the rounding procedure of \cref{sec:genericrounding}, such that we can derandomize each iteration in $O(\log^2 \Delta)$ rounds of the \LOCAL model, and $O(\log^2 \Delta  \,\cdot \,\log\log \Delta \,+\, \log\Delta \,\cdot \, \log^* n)$ rounds of the \CONGEST model. In the latter bound, the second term is upper bounded by the first term in all cases of interest where $\Delta>\log n$ and thus $\log \Delta \gg \log^* n$. This is because for $\Delta<\log n$, the theorem statement already follows from the classic $O(\Delta+\log^* n)$-round algorithm of Barenboim, Elkin, and Kuhn~\cite{BEK15}.

\paragraph{Luby's Randomized MIS Algorithm.} The starting point is to recall Luby's classic randomized algorithm from \cite{luby1993removing}. 
Each iteration of it works as follows. 
We mark each node $v$ with probability $1/(20 \deg(v))$. Then, for each edge $\{u,v\}$, let us orient this edges as $u\rightarrow v$ if and only if $\deg(u)<\deg(v)$ or $\deg(u)=\deg(v)$ and $ID(u)<ID(v)$. For each marked node $v$, we add $v$ to the independent set if and only if $v$ does not have a marked out-neighbor. Finally, as a clean-up step at the end of this iteration, we remove all nodes that have been added to the independent set along with their neighbors. We then proceed to the next iteration.

\paragraph{Derandomizing Luby via Local Rounding.} It is well-known that in each iteration of Luby's algorithm a constant fraction of the edges of the remaining graph gets removed, in expectation. Hence, the process terminates in $O(\log n)$ iterations with probability $1 - 1/\poly(n)$. 
We explain how to derandomize each iteration of the algorithm in $O(\log^2 \Delta)$ rounds, such that we still remove a constant fraction of the edges per iteration. 
For the rest of this proof, we focus on an arbitrary iteration, and we assume that $G=(V, E)$ is the graph induced by the remaining vertices at the beginning of this iteration. Let $\vec{x} \in \{0, 1\}^{|V|}$ be the indicator vector of whether different nodes are marked, that is, we have $x_v=1$ if $v$ is marked and $x_v=0$ otherwise. Let $R_v(\vec{x})$ be the indicator variable of the event that $v$ gets removed, for the marking vector $\vec{x}$. Let $Z(\vec{x})$ be the corresponding number of removed edges. Luby's algorithm determines the markings $\vec{x}$ randomly. Our task is to derandomize this and select the marked nodes in a deterministic way such that when we remove nodes added to the independent set (those marked nodes that do not have a marked out-neighbor) and their neighbors, along with all the edges incident on these nodes, at least a constant fraction of edges $E$ get removed.

\paragraph{Good and bad nodes and prevalence of edges incident on good nodes.} We call any node $v$ \emph{good} if and only if it has at least $\deg(v)/3$ incoming edges. A node $v$ that is not good is called bad. It can be proven \cite{luby1993removing} that 
\begin{align}
    \label{eq:luby0}
\sum_{\textit{good vertex\;} v} \deg(v) \geq |E|/2.
\end{align}

Even though the reader may skip this paragraph, for completeness we include the reason as it is a simple and intuitive charging argument. 
Recall that by definition any bad node has less than $1/3$ of its edges incoming. 
Thus any edge incoming to a bad node $v$ can be charged to two unique edges going out of $v$, in such a manner that each edge of the graph is charged at most once. 
Hence, the number of edges incoming to bad nodes is at most $|E|/2$. 
Thus, the number of edges that have at least one good endpoint is at least $|E|/2$, which implies the desired bound $\sum_{\textit{good vertex\;} v} \deg(v) \geq |E|/2$.

\paragraph{Lower bounding removed edges.}
We can lower bound the number of removed edges as $$Z(\vec{x}) \geq \sum_{\textit{good vertex\;} v} \deg(v) \cdot R_{v}(\vec{x})/2.$$ 
Here, the $2$ factor in the denominator is because for an edge, both endpoints might be good nodes and get removed. 
Given that $\sum_{\textit{good vertex\;} v} \deg(v) \geq |E|/2$, to prove that $\mathbb{E}[Z(\vec{x})] = \Omega(|E|)$, it suffices to show that each good vertex $v$ has $Pr[R_v(\vec{x})] =\Omega(1)$. 
This fact can be proven via elementary probability calculations. 
Next, we discuss how to prove it using only pairwise independence in the analysis and then how this nicely fit our deterministic local rounding framework.

\paragraph{Pessimistic estimator of removed edges via pairwise-independent analysis.} Let us use $IN(u)$ and $OUT(u)$ to denote in-neighbor and out-neighbors of a vertex $u$. Consider a good node $v$ and consider all its incoming neighbors $u$, i.e., neighbors $u$ such that $(\deg(u), ID(u))<(\deg(v), ID(v))$. We know that for each such neighbor $u$, we have $Pr[x_u] \geq \frac{1}{20 \deg(v)}$. Furthermore, since $v$ is good, it has at least $\deg(v)/3$ such neighbors. 
Hence, we have $\sum_{\textit{incoming neighbor}\; u} Pr[x_u] \geq 1/60$. 
Choose a subset $IN^*(v)\subseteq IN(v)$ of incoming neighbors such that 
\begin{align}
\label{eq:luby1}
\sum_{u \in IN^{*}(v)} Pr[x_u] \in [1/60, 4/60].
\end{align}
 
Notice that such a subset $IN^*(v)$ exists since the summation over all incoming neighbors is at least $1/60$ and each neighbor's probability is at most $1/20$. 
On the other hand, notice that for any node $u$, we have 
\begin{align}
\label{eq:luby2}
\sum_{w \in OUT(u)} Pr[x_{w}] \leq 1/20.    
\end{align}
This is because $|OUT(u)| \leq \deg(u)$ and for each ${w}\in OUT(u)$, we have $(\deg({w}), ID({w}))>(\deg(u), ID(u))$ and thus $Pr[x_{w}] \leq 1/(20\deg(u))$.

A sufficient event $\mathcal{E}(v,u)$ that causes $v$ to be removed is if some $u\in IN^*(v)$ is marked and no other node in $IN^{*}(v)\cup OUT(u)$ is marked. 
By union bound, this event's indicator is lower bounded by $$x_u - \sum_{u' \in IN^*(v), u\neq u'} x_{u}\cdot x_{u'} - \sum_{w \in OUT(u)} x_{u}\cdot x_{w}.$$ 
Furthermore, the events $\fE(v,u_1), \fE(v,u_2), \dots, \fE(v,u_{|IN^*(v)|})$ are mutually disjoint for different $u_1, u_2, \dots,$ $u_{|IN^*(v)|} \in IN^*(v)$. 
Hence, we can sum over these events for different $u\in IN^*(v)$ and conclude that
\begin{align*}
R_v(\vec{x}) & \geq \sum_{u\in IN^*(v)} \bigg(x_u - \sum_{u' \in IN^*(v), u\neq u'} x_{u}\cdot x_{u'} - \sum_{w \in OUT(u)} x_{u}\cdot x_{w}\bigg) \\
 & =
\sum_{u\in IN^*(v)} x_u - \sum_{u, u' \in IN^*(v)} x_{u}\cdot x_{u'} - \sum_{u\in IN^*(v)}\sum_{w \in OUT(u)} x_{u}\cdot x_{w}
\end{align*}
Therefore, our overall pessimistic estimator for the number of removed edges gives that 

\begin{align*}
Z(\vec{x}) \geq & \sum_{\textit{good vertex \,} v}  (\deg(v)/2) \cdot R_v(\vec{x})  \\
\geq  & \sum_{\textit{good vertex \,} v} (\deg(v)/2) \cdot \bigg(
\sum_{u\in IN^*(v)} x_u - \sum_{u, u' \in IN^*(v)} x_{u}\cdot x_{u'} - \sum_{u\in IN^*(v)}\sum_{w \in OUT(u)} x_{u}\cdot x_{w}\bigg)
\end{align*}

\paragraph{Formulating the pessimistic estimator in our rounding framework.} We next describe how we can formulate the above pessimistic estimator in the framework of our rounding procedure described in \Cref{sec:basicrounding}. The labeling space is whether each node is marked or not, i.e., each node takes simply one of two possible labels $\{0,1\}$ where $1$ indicates that the node is marked. For a given label assignment $\vec{x} \in \{0,1\}^{|V|}$, we define the utility function as $$\utility(\vec{x}) = \sum_{\textit{good vertex \,} v} (\deg(v)/2) \cdot \big(\sum_{u\in IN^*(v)} x_u\big),$$ and the cost function as 
$$\cost(\vec{x}) =  \sum_{\textit{good vertex \,} v} (\deg(v)/2) \cdot \bigg(
\sum_{u, u' \in IN^*(v)} x_{u}\cdot x_{u'} + \sum_{u\in IN^*(v)}\sum_{w \in OUT(u)} x_{u}\cdot x_{w}\bigg).$$

If the label assignment is relaxed to be a fractional assignment $\vec{x}\in [0,1]^{|V|}$, where intuitively now $x_v$ is the probability of $v$ being marked, the same definitions apply for the utility and cost of this fractional assignment. 

The utility function is simply a summation of functions, each of which depends on the label of only one vertex. Hence, it directly fits out rounding framework as discussed in \Cref{sec:nodeutilitiesandcosts}. 

To capture the cost function as a summation of costs over edges, we next define an auxiliary multi-graph $H$ as follows: For each good node $v$, for every two vertices $u, u'\in IN^{*}(v)$, we add an auxiliary edge between $u$ and $u'$, with a cost function which is equal $\deg(v)/2$ when both $u$ and $u'$ are marked, and zero otherwise. Furthermore, for each $u \in IN^{*}(v)$ and each $w\in OUT(u)$, we add to the edge $(u, w)$ a cost function which is equal to $\deg(v)/2$ when both $u$ and $w$ are marked and zero otherwise. Notice that $H$ is a $d2$-multigraph of $G$ according to the definition in \Cref{sec:d2communication}.

For the fractional label assignment $\vec{x}\in [0,1]^{|V|}$ in Luby's algorithm --- i.e., where $x_v=1/(20\deg(v))$--- these utility and cost functions are clearly polynomially bounded in $\Delta$, simply because each term is at least $1/{20\Delta}$ and there are no more than $\poly(\Delta)$ terms per node. We next argue that these utility and cost functions also satisfy the key requirement of \Cref{lemma:d2rounding} with $\mu=1/2$:

\begin{claim}\label[claim]{clm:MIS1} For the fractional label assignment $\vec{x}\in [0,1]^{|V|}$ in Luby's algorithm --- i.e., where $x_v=1/(20\deg(v))$--- we have $\utility(\vec{x})-\cost(\vec{x}) \geq \utility(\vec{x})/2$.
\end{claim}

\begin{proof}
We have \begin{align*}
\utility(\vec{x})-\cost(\vec{x})  = &  \sum_{\textit{good vertex \,} v} (\deg(v)/2) \cdot \bigg(\sum_{u\in IN^*(v)} x_u - \sum_{u, u' \in IN^*(v)} x_{u}\cdot x_{u'} - \sum_{u\in IN^*(v)}\sum_{w \in OUT(u)} x_{u}\cdot x_{w}\bigg)  \\
= & \sum_{\textit{good vertex \,} v} (\deg(v)/2) \cdot \bigg(\sum_{u\in IN^*(v)} x_u \cdot \big(1- \sum_{u' \in IN^*(v)} x_{u'} - \sum_{w \in OUT(u)} x_{w}\big)\bigg)  \\
\geq & \sum_{\textit{good vertex \,} v} (\deg(v)/2) \cdot \bigg(\sum_{u\in IN^*(v)} x_u \big(1- 4/60 - 1/20 \big)\bigg) \\
\geq & \sum_{\textit{good vertex \,} v} (\deg(v)/2) \cdot \big(\sum_{u\in IN^*(v)} x_u/2\big)  = \utility(\vec{x})/2
\end{align*}
where we have used \cref{eq:luby1,eq:luby2}. 
\end{proof}

Then, in the \CONGEST model, we apply the rounding procedure of \Cref{lemma:d2rounding} with $\eps=1/2$ on the fractional label assignment $\vec{x}\in [0,1]^{|V|}$ in Luby's algorithm. The algorithm runs in $O(\log^2 \Delta \cdot \log\log \Delta + \log \Delta \log^* n)$ rounds of the \CONGEST model. If we were in the more relaxed \LOCAL model which allows unbounded message sizes, then auxiliary graph $H$ can be directly simulated in graph $G$ with no asymptotic round complexity loss, and we can already invoke \Cref{lemma:distributedonG} which performs a rounding with the same guarantees in $O(\log^2 \Delta + \log \Delta \log^* n)$ rounds. The $\log^* n$ term can be replaced by a $O(\log^* \Delta)$, which thus makes the entire second additive term negligible, by computing a $\poly(\Delta)$-color distance-$2$ coloring of $G$ at the beginning of the first iteration. The one-time $O(\log^* n)$ additive round complexity of computing this distance-$2$ coloring of $G$, which happens at the very beginning of the algorithm, is subsumed by our $O(\log n)$ number of iterations.

Suppose that we get an integral marking assignment $\vec{y}\in \{0,1\}^{|V|}$. From \Cref{lemma:d2rounding}, we know that $Z(\vec{y})=\utility(\vec{y})-\cost(\vec{y}) \geq (1/2) \cdot (\utility(\vec{x})-\cost(\vec{x})).$ Next, we argue that this implies $Z(\vec{y}) \geq |E|/500$.

\begin{claim}\label[claim]{clm:MIS2} For the fractional label assignment $\vec{x}\in [0,1]^{|V|}$ in Luby's algorithm --- i.e., where $x_v=1/(20\deg(v))$--- we have $Z(\vec{x}) = \utility(\vec{x})-\cost(\vec{x}) \geq |E| /240$. Hence, for the integral marking assignment we obtain from rounding $\vec{x}$ by invoking \Cref{lemma:d2rounding} in the \CONGEST model or \Cref{lemma:distributedonG} in the \LOCAL model, we have $Z(\vec{y})=\utility(\vec{y})-\cost(\vec{y}) \geq (1/2) \cdot (\utility(\vec{x})-\cost(\vec{x})) \geq |E|/500$.
\end{claim}
\begin{proof}From \Cref{clm:MIS1}, we have $Z(\vec{x}) \geq \utility(\vec{x})-\cost(\vec{x}) \geq \utility(\vec{x})/2$. Hence, 

\begin{align*} Z(\vec{x}) \geq \utility(\vec{x})/2 &= 
\sum_{\textit{good vertex \,} v} (\deg(v)/2) \bigg(\sum_{u\in IN^*(v)} x_u/2\bigg)  \\
&\ge \sum_{\textit{good vertex \,} v} (\deg(v)/2) \cdot (1/120) \geq |E|/240,
\end{align*}
where we first used \cref{eq:luby1} that says that $\sum_{u \in IN^*(v)} x_u \ge 1/60$ and then we used \cref{eq:luby0} that bounds $\sum_{\textit{good vertex \,} v} \deg(v) \geq |E|/2$. 

Since $Z(\vec{y})=\utility(\vec{y})-\cost(\vec{y}) \geq (1/2) \cdot (\utility(\vec{x})-\cost(\vec{x}))$, the claim follows. 
\end{proof}

Hence, from the rounding procedure described above, which runs in $O(\log^2 \Delta \cdot \log\log \Delta + \log \Delta \cdot \log^* n)$ rounds of the \CONGEST model or $O(\log^2 \Delta + \log \Delta \log^* \Delta) = O(\log^2 \Delta)$ rounds of \LOCAL model,  we get an integral marking assignment $\vec{y}$ with the following guarantee: if we add marked nodes $u$ that have no marked out-neighbor to the independent set and remove them along with their neighbors, we remove at least a $1/500$ fraction of the remaining edges. Hence, $O(\log n)$ such iterations suffice to complete the computation and have a maximal independent set. This completes the proof of \cref{thm:mis}. 

%% file: weightedIS.tex
\section{Maximum Weight Independent Set}
\label{sec:weightedIS}

In this section, we apply our generic rounding scheme of \Cref{sec:basicrounding} to obtain fast deterministic algorithms for computing a large weight independent set of a graph $G=(V,E)$ with node weights $w:V\to\N$.\footnote{We assume that the node weights are positive integers to keep the algorithms and their analyses simpler. The algorithms can however all be generalized to handle arbitrary positive weights.} Note that we here assume that $G$ is a simple graph and we can therefore define $E$ as a subset of $\binom{V}{2}$. Further, when discussing distributed algorithms on $G$, for simplicity, we will generally assume that initially all nodes $u\in V$ know the weights of all their neighbors. Equivalently we can assume that we have enough bandwidth to exchange this information in a single round (e.g., that node weights can be encoded using $O(\log n)$ bits when using the \CONGEST model). Further, for a subset $U\subseteq V$ of nodes, we use $w(U):=\sum_{u\in U}w(u)$ as a shortcut for the total weight of nodes in $U$. We will also assume that all nodes know an upper bound $W$ on the maximum node weight.

While our rounding method allows us to study problems with a potentially large alphabet of node labels, as in \Cref{sec:MIS}, we will only need $2$ labels: either a node is in the independent set or the node is not in the independent set. We can thus again characterize a fractional solution by a single fractional value $x_v\in[0,1]$ per node $v\in V$, where $x_v$ is the fractional value for label $1$, i.e., the fractional value for being in the set. In an integral solution, we have $x_v\in \set{0,1}$ and the computed subset of nodes is defined by the nodes $v$ with $x_v=1$.

\paragraph{Definition of Utility and Cost.} We next define the utility and cost function for computing heavy independent sets. For simplicity, we define the utility as a sum over nodes instead of edges. It is however straightforward to get an equivalent definition that fits the framework of \Cref{sec:roundingsetup}. Let $\vec{x}\in[0,1]^{|V|}$ be the vector of fractional values $x_v$. Utility $\utility(\vec{x})$ and cost $\cost(\vec{x})$ are then defined as
\begin{equation}
  \label{eq:ISutilityandcost}
  \utility(\vec{x}) := \sum_{v\in V} w(v)\cdot x_v
  \qquad\text{and}\qquad
  \cost(\vec{x}) := \sum_{\set{u,v}\in E} \min\set{w(u), w(v)}\cdot x_u\cdot x_v.
\end{equation}

\paragraph{From an Integral Solution to an Independent Set.} Given integral node value $x_v\in\set{0,1}$, it is straightforward to obtain an independent set of total value $\utility(\vec{x}) - \cost(\vec{x})$: we start with the set $I_0$ of nodes $v$ for which $x_v = 1$ and for every two neighbors $u,v$ in $I_0$ we simply remove the one with smaller weight. This is formalized as follows. 

\begin{lemma}\label{lemma:computingIS}
  Let $G=(V,E)$ be a graph with node weights $w:V\to \N$ and let $\vec{x}$ be a vector of (integral) node values $x_v\in\set{0,1}$. Then, there is a deterministic $1$-round \CONGEST algorithm to compute an independent set $I$ of $G$ of total weight at least $\utility(\vec{x}) -\cost(\vec{x})$.
\end{lemma}
\begin{proof}
  Let $I_0\subseteq V$ be the set of nodes $v\in V$ with $x_v=1$. We define a set $I'\subseteq V$ of nodes to be removed as follows. A node $u\in I_0$ is in $I'$ iff there is a neighbor $v\in N(u)\cap I_0$ such that either $w(v)>w(u)$ or $w(v)=w(u)$ and the ID of $v$ is larger than the ID of $u$. The independent set $I$ is then defined as $I:=I_0 \setminus I'$. Clearly, this algorithm can be implemented in a single \CONGEST algorithm. To determine if a node $u\in I_0$ is in the set $I$, it only needs to know which of its neighbors are in $I_0$ and it needs to know the weights and IDs of those neighbors (recall that we assume that node weights can be communicated in a single round or that they the node weights of neighbors are known initially).  It is further not hard to see that $I$ is an independent set. For every edge $\set{u,v}\in E$, if $u\in I_0$ and $v\in I_0$, at least one of the two nodes is added to $I'$. If $w(u)\neq w(v)$, then the node of smaller weight is added to $I'$, otherwise, the node of smaller ID is added to $I'$. It remains to show that the weight $w(I)$ of the independent set $I$ is as claimed. We clearly have
  \[
  w(I) = w(I_0) - w(I')\qquad\text{and}\qquad w(I_0) = \utility(\vec{x}).
  \]
  We therefore need to show that $w(I')\leq \cost(\vec{x})$. To see this, note that by the construction of $I'$, the inclusion of a node $v$ in $I'$ can be blamed on an incident edge $\set{u,v}$ and every edge with $x_u=x_v=1$.  Further, every such edge can be blamed by at most one node $v\in I'$ and we also know that if $v\in I'$ is blamed on edge $\set{u,v}$, then $w(u)\geq w(v)$. The contribution of edge $\set{u,v}$ to $\cost(\vec{x})$ is equal to $\min\set{w(u), w(v)}\cdot x_u\cdot x_v = w(v)$ and we therefore have $w(I')\leq \cost(\vec{x})$ as required.
\end{proof}

\paragraph{Basic Independent Set Rounding Algorithm.}
We next show that given a fractional solution $\vec{x}$ for which the utility is sufficiently larger than the cost, there is an efficient \CONGEST algorithm to compute a large independent set. Specifically, we prove the following lemma, which we will use in the remaining algorithms of this section.

\begin{lemma} \label{lemma:basicISalgorithm}
  Let $G=(V,E)$ be a graph with maximum degree $\Delta$ and node weight function $w:V\to \N$. Further assume that we are given a fractional independent set solution $\vec{x}$ with $\utility(\vec{x})\geq 2\cost(\vec{x})$. If $G$ is equipped with a $\zeta$-vertex coloring for some $\zeta\leq\poly(n)$, for every $\eps\in(0,1]$, an independent set of total weight at least $(1/2 - \eps)\utility(\vec{x})$ can be computed in $O(\log^2(\Delta/\eps) + \log^* \zeta)$ deterministic rounds in the \CONGEST model.
\end{lemma}
\begin{proof}
  We first show that the fractional solution $\vec{x}$ can be transformed into a fractional solution $\vec{x}'$ with no small positive fractional values. Specifically, for all $v\in V$, we set $x_v' := x_v + \frac{\eps}{2\Delta}$. We then have
  \begin{eqnarray}
	\utility(\vec{x}') 
	& = &   
	\sum_{v \in V} w(v)\cdot \left (x _v+ \frac{\eps}{2\Delta}\right)
	\ =\ \utility(\vec{x}) + w(V)\cdot\frac{\eps}{2\Delta},\label{eq:utilityprime}\\
	\cost(\vec{x}')
	& =  &
	\!\!\!\!\sum_{\set{u,v}\in E}\!\!\! \min\{w(u),w(v)\} \cdot 
	\left(x_u + \frac{\eps}{2\Delta}\right)\cdot
	\left(x_v + \frac{\eps}{2\Delta}\right)\nonumber\\
	& \le  &
	\!\!\!\!\sum_{\set{u,v}\in E}\!\!\! \left(
 \min\{w(u),w(v)\} \cdot x_u x_v 
 + w(u) x_u \frac{\eps}{2\Delta} + w(v) x_v \frac{\eps}{2\Delta} + \frac{w(u) + w(v) }{2} \cdot \left( \frac{\eps}{2\Delta} \right)^2	\right)
 \nonumber\\
	& \leq &
	\!\!\!\!\sum_{\set{u,v}\in E}\!\!\!  \min\{w(u),w(v)\} \cdot  x_u x_v + 
	\sum_{v\in V} \frac{\deg_G(v)\cdot  w(v)\cdot x_v \cdot\eps}{2\Delta} +
	\frac{1}{2} \cdot \sum_{v \in V} \frac{\deg_{G}(v)\cdot w(v)\cdot\eps^2}{4\Delta^2} \nonumber\\
	& \leq &
	\cost(\vec{x}) + \frac{\eps}{2}\cdot\utility(\vec{x}) + w(V)\cdot\frac{\eps^2}{8\Delta}\nonumber\\
	& \leq &
	\left(\frac{1}{2}+\frac{\eps}{2}\right)\cdot\utility(\vec{x}) + \frac{\eps}{4}\cdot\big(\utility(\vec{x}')-\utility(\vec{x})\big).\label{eq:costprime}
  \end{eqnarray}
  In the last inequality, we use that $\utility(\vec{x})\geq 2\cost(\vec{x})$ and we use \eqref{eq:utilityprime}. We now have
  \begin{eqnarray*}
    \utility(\vec{x}') - \cost(\vec{x}')
    & = &
    \utility(\vec{x}) + \big(\utility(\vec{x}') - \utility(\vec{x})\big) - \cost(\vec{x}')\\
    & \stackrel{\eqref{eq:costprime}}{\geq} &
    \left(\frac{1}{2}-\frac{\eps}{2}\right)\cdot\utility(\vec{x}) + 
    \left(1-\frac{\eps}{4}\right)\cdot \big(\utility(\vec{x}') - \utility(\vec{x})\big)
    \geq
    \left(\frac{1}{2}-\frac{\eps}{2}\right)\utility(\vec{x}').
  \end{eqnarray*}
  In the last inequality, we use that $\utility(\vec{x}')\geq \utility(\vec{x})$ and that $1-\eps/4\geq 0$. By using \Cref{lemma:distributedonG} and the fact that for all $v\in V$, $x_v'\geq \frac{\eps}{2\Delta}$, we directly get that an independent set of size $(1-\eps)\cdot \big(\frac{1}{2}-\frac{\eps}{2}\big)\cdot\utility(\vec{x}')\geq \big(\frac{1}{2}-\eps\big)\cdot\utility(\vec{x})$ can be computed deterministically in $O(\log^2(\Delta/\eps) + \log^*\zeta)$ rounds in the \CONGEST model. Here, we also use that the necessary message size of \cref{lemma:distributedonG} is $O(\log n)$, since $|\Labels|=2$ and $\log\zeta=O(\log n)$).
\end{proof}

\subsection{Approximation for Graphs with Bounded Neighborhood Independence}
\label{sec:boundedindep}

We next show how to efficiently compute a $(1-\eps)/\beta$-approximation to the maximum weight independent set (MWIS) problem of a graph $G=(V,E)$, where $\beta$ denotes the neighborhood independence of $G$ (i.e., the maximum number of pairwise non-adjacent neighbors of a node). We use the following linear program (LP) to define the family of fractional solutions that we use.
\begin{equation}\label{eq:fractionISLP}
	\max \sum_{v\in V} w(v)\cdot x_v\qquad\text{s.t.}\quad
	\forall v\in V\, :\!\!\! \sum_{u\in N^+(v)}\!\!\! x_u \leq 1\quad\text{and}\quad
	x_v \geq 0. 
\end{equation}
In the following, let $S^*(w)$ be the value of an optimal solution of LP \eqref{eq:fractionISLP} on graph $G$ with weight function $w$. Note that in a graph of neighborhood independence $\beta$, for each node $u\in V$, the number of nodes in $N^+(u)$ in any independent set is at most $\beta$. Hence, the weight of a maximum weight independent set is at most $\beta\cdot  S^*(w)$. We therefore need to show that we can efficiently compute an independent set of weight at least $(1-\eps)S^*(w)$. We first show that we can efficiently compute an independent set of weight at least $S^*(w)/4$.

\begin{lemma}\label{lemma:LPapproxIS}
  Let $G=(V,E)$ be a graph with maximum degree $\Delta$ and node weight function $w:V\to \N$. If $G$ is equipped with an $O(\Delta^2)$-vertex coloring, an independent set of total weight at least $S^*(w)/4$ can be computed in $O\big(\log^2(\Delta W)\big)$ deterministic rounds in the \CONGEST model.
\end{lemma}
\begin{proof}
  First, observe that a $(2/3)$-approximation of the LP \eqref{eq:fractionISLP} can be computed in $O\big(\log^2(\Delta W)\big)$ rounds by using the distributed covering and packing LP algorithm of \cite{nearsighted}.

  Next, note that any feasible solution of \eqref{eq:fractionISLP} satisfies $\utility(\vec{x})\geq 2 \cost(\vec{x})$:
  \begin{eqnarray*}
    \cost(\vec{x})
    &=&\sum_{ \{u,v\} \in E}x_u\cdot x_v \cdot \min\{w(u),w(v)\}\\
    & =  & \frac{1}{2}\sum_{u \in V} \sum_{v \in N(u)} x_u\cdot x_v \cdot \min \{w(u),w(v)\}\\
    & \leq &\frac{1}{2}\sum_{u \in V} x_u \cdot w(u)  \cdot \sum_{v \in N(u)} x_v\\
    & \stackrel{\eqref{eq:fractionISLP}}{\leq} &  \frac{1}{2}\sum_{u \in V} x_u \cdot w(u) =  \frac{1}{2}\ \utility(\vec{x}),
  \end{eqnarray*}
  The claim of the lemma now directly follows from \cref{lemma:basicISalgorithm}.
\end{proof}

Note that in \cref{lemma:LPapproxIS}, the constant $4$ could be replaced by any constant larger than $2$. In order to obtain an independent set of weight $(1-\eps)S^*(w)$, we however have to combine the lemma with some additional ideas. To this end, we adapt a technique that has been used in \cite{KawarabayashiKS20} (and which is based on the local-ratio technique described in \cite{localratio}). We first describe and analyze an abstract iterative process to compute an independent set. Consider a sequence of independent sets $I_1, I_2, \dots, I_T$ and node weight functions $w_1, w_2, \dots, w_{T+1}$ and $w_1',w_2'\dots,w_T'$, which are constructed as follows. The weight function $w_1$ is equal to the original weight function $w$ (i.e., $w_1:=w$) and $I_1$ is an arbitrary independent set of $G$. For each $i\in \set{1,\dots,T}$, $w_{i+1}$ and $w_{i}'$ are defined as follows:
\[
  \forall u\in V\,:\, w_{i+1}(u) = \max\set{0, w_{i}(u) - \sum_{v\in N^+(u)\cap I_{i}} w_i(v)}
  \quad\text{and}\quad
  w_{i}'(u) := w_{i}(u) - w_{i+1}(u).
\]
Further, for $i>1$, $I_i$ is an arbitrary independent set of the subgraph of $G$ induced by the nodes $v$ with $w_i(v)>0$. Hence, the process runs in phases. In phase $i$, we start with edge weights $w_i$ and we determine some independent set $I_i$ of the nodes of $G$ with positive weight. We then adjust the weights as follows. For each node in $I_i$, we set the weight $w_{i+1}$ to $0$ and for each node that has neighbors in $I_i$, we deduce the sum of the weights of the neighboring nodes in $I_i$ (or we set the weight to $0$ in case it would become negative otherwise). The weights $w_i'$ are just defined as the amount by which the weights are reduced when going from weights $w_i$ to weights $w_{i+1}$. Because for nodes $v\in I_i$, we set $w_{i+1}(v)=0$ and for such nodes, we therefore have $w_i(v)=w_i'(v)$. Also note that since the weights $w_{i+1}$ of nodes in independent set $I_i$ are set to $0$, the independent sets $I_i$ for different $i$ are disjoint.

We now define a combined independent set $I$ as follows. We start with $I=\emptyset$ and go through the sets $I_1, \dots, I_T$ in reverse order (i.e., starting with $I_T$). When considering set $I_j$, we add all nodes of $I_j$ to $I$, which do not already have a neighbor in $I$. That is, more formally, we consider $j=T,T-1,\dots,1$ and in step $j$, we first define $I_j' := I_j \setminus N^+(I)$ and we then set $I:=I\cup I_j'$. The following lemma shows that the total weight $w(I)$ of the independent set $I$ can be lower bounded by the sum of the weights of the independent sets $I_i$ w.r.t.\ the weight function $w_i$ that is used when computing $I_i$.

\begin{lemma}\label{lemma:totalindsetsize}
  \[
  w(I) \geq \sum_{i=1}^T w_i(I_i) = \sum_{i=1}^T w_i'(I_i).
  \] 
\end{lemma}
\begin{proof}
  Let $\bar{I}:= I_1\cup\dots\cup I_T$. Note that we have $I\subseteq \bar{I}$. For each node $u\in \bar{I}$, we define $i_u$ as the index $i\in \set{1,\dots,T}$ for which $v\in I_i$. We further define a function $\nu:\bar{I}\to I$ as follows. For $u\in I$, we set $\nu(u)=u$. Now, consider a node $u\in\bar{I}\setminus I$. The fact that $u$ was not added to $I$ implies that there is at least one neighbor $v\in N(u)$ such that $v$ was added to $I$ before considering set $I_{i_u}$, i.e., $i_v>i_u$. If there are several such nodes, we pick $v$ arbitrarily among them and we set $\nu(u)=v$. For each node $v\in I$, we further define $\nu^{-1}(v):=\set{u\in \bar{I} : \nu(u)=v}$. We next show that
  \begin{equation}\label{eq:indsetcharging}
      \forall v\in I\,:\,w(v)\geq \sum_{u\in \nu^{-1}(v)} w_{i_u}(u)
  \end{equation}
  To see \eqref{eq:indsetcharging}, recall the above process. Also recall that $i_v>i_u$ for all $u\in \nu^{-1}(v)$. Whenever a neighbor $u$ of $v$ is added to independent set $I_{j}$ for some $j<i_u$, the weight of $v$ is reduced by $w_{i_u}(u)$. Note that since $v\in I_{i_v}$, we have $w_{i_v}(v)>0$ and thus when $u$ is added to $I_{j}$, the full weight $w_{j}(u)$ is reduced from the weight of $v$. We therefore have
  \[
  w_{i_u}(v) = w(v) - \sum_{j=1}^{i_u-1} w_j(I_j\cap N(v))
  \]
  \cref{eq:indsetcharging} now follows because $\nu^{-1}(v)\subseteq \set{v}\cup\bigcup_{j<i_u} (I_j\cap N(v))$. From \cref{eq:indsetcharging}, we directly get the claim of the lemma by summing over all nodes in $I$. For the second part of the claim, note that for every $v\in I_i$, we have $w_i'(v)=w_i(v)$.
\end{proof}

Recall that $w_i'$ is the difference between $w_i$ and $w_{i+1}$, i.e., it is the amount by which the weights $w_i$ are decreased after adding independent set $I_i$ to the sequence. The next technical lemma shows that w.r.t.\ the weights $w_i'$, independent set $I_i$ achieves the bound given by the LP \eqref{eq:fractionISLP}.

\begin{lemma}\label{lemma:piecewiseindset}
	\[
	\forall i \in \set{1,\dots,T}, w_i(I_i) = w_i'(I_i) \geq S^*(w'_i).
	\]
\end{lemma}
\begin{proof}
  We prove this by considering the dual LP of LP \eqref{eq:fractionISLP} on graph $G$ with weight function $w'_i$:
  \begin{equation}\label{eq:dualLP}
	\min \sum_{v\in V}  y_v\qquad\text{s.t.}\quad
	\forall v\in V\, :\!\!\! \sum_{u\in N^+(v)}\!\!\! y_u \geq w'_i(v)\quad\text{and}\quad
	y_v \geq 0.
  \end{equation}
  Note that by the strong duality theorem for linear programs, the optimal value of the dual LP \eqref{eq:dualLP} is equal to the optimal value of LP \eqref{eq:fractionISLP} and thus equal to $S^*(w_i')$. We next show that a feasible solution to \eqref{eq:dualLP} is achieved by setting $y_v= w'_i(v)$ if $v \in I_i$, $y_v=0$ otherwise. We need to show that for all $v\in V$, $\sum_{u\in N^+(v)} y_u \geq w_i'(v)$. This is clearly true for $v\in I_i$. It is also true if $v\not\in I_i$ and $v$ has no neighbor in $I_i$ because in this case $w_i'(v)=0$. Let us therefore assume that $N(v)\cap I_i\neq \emptyset$. Note that in this case, $w_i'(v) = \min\big\{w_i(v), \sum_{u\in N(v)\cap I_i} w_i(u) = \sum_{u\in N(v)\cap I_i} w_i'(u)\big\}$. We therefore have  $\sum_{u\in N^+(v)} y_u \geq w_i'(v)$ and thus $y_v=w_i'(v)$ gives a feasible solution for \eqref{eq:dualLP}. The objective value of this feasible solution is exactly $w_i'(I_i)=w_i(I_i)$, which is therefore an upper bound on the optimal value $S^*(w_i')$ of LP \eqref{eq:dualLP}. 
\end{proof}

The following theorem shows that if each independent set $I_i$ is chosen such that $w_i(I_i)=\Omega(S^*(w_i))$, then after computing a sequence of $O(\log(1/\eps))$ independent sets, the resulting independent set $I$ has weight at least $(1-\eps)\cdot S^*(w)$.

\begin{lemma}\label{lemma:boostingalgo}
	Let $G=(V,E)$ be a graph with node weights $w:V\to\N$ s.t.\ $w(v)\leq W$ for all $v\in V$. Assume that in the above iterative process, for every $i\geq 1$, independent set $I_i$ is chosen such that $w_i(I_i)\geq\rho\cdot S^*(w_i)$ for some given $\rho\in(0,1)$. Then after computing independent sets $I_1,\dots,I_T$ for $T\geq \ln(1/\eps)/\rho$, the resulting independent set $I$ has a total weight of $w(I)\geq (1-\eps)S^*(w)$.
\end{lemma}
\begin{proof}
  Let $I^{(t)}$ be the independent set that we obtain when running the above process for $t\leq T$ steps. That is, $I^{(t)}$ is the independent set that we get when greedily going through the independent sets $I_t, I_{t-1}, \dots, I_1$. Note that by \Cref{lemma:totalindsetsize}, we have $w(I^{(t)})\geq \sum_{i=1}^t w_i(I_i)$. For $t\geq 0$, we define a potential $\Upsilon_t$ as
  \[
  \Upsilon_t := S^*(w) - \sum_{i=1}^t w_i(I_i)
  \]
  to measure the progress towards obtaining an independent set of weight $w(I)$ approaching $S^*(w)$. For all $t\geq 0$, we inductively prove the following properties:
  \[
  \text{(I)}\,:\, \Upsilon_t \leq S^*(w_{t+1})\quad\text{and}\quad
  \text{(II)}\,:\, \Upsilon_t \leq (1-\rho)^{t}\cdot \Upsilon_0.
  \]
  Note that both statements are clearly true for $t=0$ (recall that for all $v$, $w_1(v)=w(v)$). For the induction step, note first that by assumption, we have that $w_t(I_t)\geq \rho S^*(w_t)$ and thus by the induction hypothesis (property (I) for $t-1$), $w_t(I_t)\geq \rho \Upsilon_{t-1}$. Since we have $\Upsilon_t = \max\set{0, \Upsilon_{t-1}-w_t(I_t)}$, this implies that $\Upsilon_t \leq \Upsilon_{t-1} - \rho\Upsilon_{t-1}$, which together with the induction hypothesis (property (II) for $t-1$) proves the induction step for property (II). To do the induction step for property (I), first consider three weight functions $w_a$, $w_b$, and $w_{ab}$ such that for all $v\in V$, $w_{ab}(v)= w_a(v) + w_b(v)$. Clearly, for any feasible solution for LP \eqref{eq:fractionISLP}, the objective value for weights $w_{ab}$ is equal to the sum of the objective values for weights $w_a$ and $w_b$. For the optimal objective values w.r.t.\ the three weight functions, we therefore have $S^*(w_a)+S^*(w_b)\geq S^*(w_{ab})$. Now recall that we have $w_t'=w_t - w_{t+1}$ and we therefore have $S^*(w_{t+1}) + S^*(w_t')\geq S^*(w_t)$. From \Cref{lemma:piecewiseindset}, we know that $S^*(w_t')\leq w_t(I_t)$. We therefore have
  \begin{eqnarray*}
    S^*(w_{t+1}) & \geq & S^*(w_t) - S^*(w_t')\ \geq\ S^*(w_t) - w_t(I_t)\\
    & \geq & S^*(w_t) - (\Upsilon_{t-1} - \Upsilon_t)
    \ \geq\ \Upsilon_{t-1} - (\Upsilon_{t-1}-\Upsilon_t) = \Upsilon_t.
  \end{eqnarray*}
  The last inequality follows from the induction hypothesis (property (I) for $t-1$). This also completes the induction step for property (I) for $t$.
  
  By \Cref{lemma:totalindsetsize} and the definition of $\Upsilon_T$, we have $w(I)\geq \Upsilon_0 - \Upsilon_T$. Note that we have $T\geq \ln(1/\eps)/\rho\geq \ln(1/\eps)/\ln(1/(1-\rho))$. This holds because $\ln(1/(1-\rho)) = \sum_{i=1}^\infty \rho^i/i!\geq \rho$. By property (II) for $t=T$, we therefore have 
  \[
  \Upsilon_T\leq (1-\rho)^T\cdot\Upsilon_0 = (1-\rho)^{\ln(1/\eps)/\ln(1/(1-\rho))}\cdot\Upsilon_0 = \eps\cdot\Upsilon_0.
  \]
  We therefore have $w(I) \geq (1-\eps)\cdot\Upsilon_0 = (1-\eps)\cdot S^*(w)$ as required.
\end{proof}

\paragraph{Proof of \Cref{thm:ISthm}, Part 1:} The following lemma proves result \eqref{eq:ISthm1} of \Cref{thm:ISthm}.

\begin{lemma}\label{lemma:ISbetaapprox}
Let $G=(V,E)$ be a graph with maximum degree $\Delta$, node weight function $w:V\to \Rp$, neighborhood independence $\beta$, and assume that $G$ is equipped with a proper $\xi$-vertex coloring. For any $\eps>0$, there is deterministic \CONGEST algorithm that computes a $(1-\eps)/\beta$-approximation to the maximum weight independent set problem on $G$. The round complexity of the algorithm is $O\big(\log^2(\Delta W) \cdot \log(1/\eps)+\log^* \xi\big)$. 
\end{lemma}
\begin{proof}
  At the start, we compute a proper $O(\Delta^2)$-vertex coloring of $G$. This can be done in $O(\log^*\xi)$ deterministic rounds in the \CONGEST model by using a classic algorithm of \cite{linial92}. We then use the above framework, where in each step $i$, we use the algorithm of \Cref{lemma:LPapproxIS} to compute an independent set $I_i$ of weight $w_i(I_i)\geq S^*(w_i)/4$. By \Cref{lemma:LPapproxIS} and by using the initial $O(\Delta^2)$-coloring, the round complexity for computing each such independent set $I_i$ is $O\big(\log^2(\Delta W)\big)$. By \Cref{lemma:boostingalgo}, we have to do $O(\log(1/\eps))$ steps to obtain an independent set of weight $(1-\eps)S^*(w)$. To see that this is a $(1-\eps)/\beta$-approximation, note that in a graph of neighborhood independence $\beta$, clearly the neighborhood $N^+(v)$ of every node $v$ can contain at most $\beta$ nodes of any independent set. Hence, if we replace the right-hand sides of all constraints in LP \eqref{eq:fractionISLP} by $\beta$, the LP is a relaxation of the maximum weight independent set problem (any independent set satisfies the constraints). Clearly, increasing the right-hand sides of all constraints to $\beta$ increases the optimal objective value by exactly a factor $\beta$. Consequently, the weight of an optimal weighted independent set in a graph of neighborhood independence $\beta$ is at most $\beta\cdot S^*(w)$. This completes the proof.
\end{proof}

\paragraph{Proof of \Cref{thm:ISthm}, Part 2:} Result \eqref{eq:ISthm2} of \Cref{thm:ISthm} can be proven in a similar way. The following lemma proves \eqref{eq:ISthm2} by adapting the inductive argument in \Cref{lemma:boostingalgo}.

\begin{lemma}\label{lemma:ISresult2}
  Let $G=(V,E)$ be a graph with maximum degree $\Delta$ and node weight function $w:V\to \Rp$, and assume that $G$ is equipped with a proper $\xi$-vertex coloring. Then, for any $\eps>0$, there is deterministic \CONGEST algorithm that computes an independent set of $G$ of total weight at least $(1-\eps)\cdot\frac{w(V)}{\Delta+1}$ in $O(\log^2\Delta\cdot\log(1/\eps) + \log^*\xi)$ rounds.
\end{lemma}
\begin{proof}
  As a initial step, we compute a proper $O(\Delta^2)$-vertex coloring of $G$ in time $O(\log^*\xi)$ by using an algorithm of \cite{linial92}. Now, we again use the method of iteratively computing independent sets $I_1,\dots,I_T$ as above. First note that for a given weight function $w:V\to\N$, we can efficiently compute an independent set of weight $w(V)/(4(\Delta+1))$ in $O(\log^2\Delta +\log^*(\Delta^2))=O(\log^2\Delta)$ rounds by using \Cref{lemma:LPapproxIS}. To see this, observe that setting the fractional value of all nodes $v$ to $x_v=1/(\Delta+1)$ gives a feasible solution to LP \eqref{eq:fractionISLP} with objective value $w(V)/(\Delta+1)$. We therefore in particular have $S^*(w)\geq w(V)/(\Delta+1)$ and for $W\leq\poly(\Delta)$, the claim of the lemma directly follows from \Cref{lemma:boostingalgo}. In the following, we show that by using a similar argument as in the proof of \Cref{lemma:boostingalgo}, we can also get a round complexity of $O(\log^2\Delta + \log^* \xi)$.
  
  As discussed, we use the iterative framework from above. In each iteration $t$, we choose the independent set $I_t$ such that $w_t(I_t)\geq w_t(V)/(4(\Delta+1))$. As discussed above, we can find such an independent set $I_t$ in $O(\log^2\Delta)$ rounds. We again define a potential $\Upsilon_t$ that measures how close we are to achieving our goal:
  \[
  \forall t\geq 0\,:\, \Upsilon_t := \frac{w(V)}{\Delta+1} - \sum_{i=1}^t w_i(I_i).
  \]
  For all $t\geq 0$, we inductively prove the following two properties:
  \[
  \text{(i)}\,:\, \Upsilon_t \leq \frac{w_{t+1}(V)}{\Delta+1}\quad\text{and}\quad
  \text{(ii)}\,:\, \Upsilon_t \leq \left(1-\frac{1}{4}\right)^{t}\cdot \Upsilon_0.
  \]
  Both statements are clearly true for $t=0$. Let us therefore consider the induction step from $t-1$ to $t$. We have
  \[
  \Upsilon_t = \Upsilon_{t-1} - w_t(I_t) \leq
  \Upsilon_{t-1} - \frac{w_t(V)}{4(\Delta+1)} \leq
  \Upsilon_{t-1}\cdot\left(1-\frac{1}{4}\right).
  \]
  The first inequality follows from $w_t(I_t)\geq w_t(V)/(4(\Delta+1))$ and the second inequality follows from the induction hypothesis for $t-1$ (part (i)). Together with the induction hypothesis (part (ii)), the above inequality directly implies statement (ii) for $t$. To also prove statement (i) for $t$, recall that $w_t'(V) = w_t(V) - w_{t+1}(V)$. Also note that we have $w_t'(V) \leq (\Delta+1)\cdot w_t(I_t)$ because for every node $v\in I_t$, the weight $w_t(v)$ is deducted (and thus added to $w_t'$) for at most $\Delta+1$ nodes. We therefore have $w_t(I_t) \geq \frac{w_t(V)}{\Delta+1}-\frac{w_{t+1}(V)}{\Delta+1}$. We thus obtain
  \[
  \Upsilon_t = \Upsilon_{t-1} - w_t(I_t) \leq
  \Upsilon_{t-1} + \left(\frac{w_{t+1}(V)}{\Delta+1}-\frac{w_{t}(V)}{\Delta+1}\right) \leq
  \frac{w_{t+1}(V)}{\Delta+1}.
  \]
  The last step follows from the induction hypothesis (part (i)). For $T=\log_{4/3}(1/\eps)$, statement (ii) implies $\Upsilon_T\leq \eps \Upsilon_0 =\eps\cdot\frac{w(V)}{\Delta+1}$ and \Cref{lemma:totalindsetsize} then gives $w(I)\geq (1-\eps)\cdot \frac{w(V)}{\Delta+1}$ as required. The overall time complexity is $O(\log^*\xi)$ for the initial $O(\Delta^2)$-coloring and $O(\log^2\Delta\cdot\log(1/\eps))$ for computing independent sets $I_1,\dots,I_T$ and then independent set $I$.
\end{proof}

We conclude \Cref{sec:boundedindep} by showing that our method yields an alternative to Fischer's $O(\log^2\Delta\cdot\log n)$-round deterministic \CONGEST algorithm~\cite{fischer2020improved} for computing a maximal matching of a graph. This show that all the four classic symmetry breaking problems (MIS, maximal matching, $(\Delta+1)$-vertex coloring, and $(2\Delta-1)$-edge coloring) can be solved by the same general method and in all cases, the method yields the best current distributed algorithms.

\begin{theorem}\label{thm:maximalmatching}
  A maximal matching in an $n$-node graph $G=(V,E)$ of maximum degree $\Delta$ can be computed deterministically in $O(\log^2\Delta\cdot\log n)$ rounds in the \CONGEST model.
\end{theorem}
\begin{proof}[Proof Sketch]
  We start by computing an $O(\Delta^2)$-edge coloring. A simple way to do this in $O(\log^* n)$ rounds in the \CONGEST model is the following (also see \cite{Kuhn2009WeakColoring}). Each node $v$ first uniquely labels its $\deg(v)$ edges with the number $1,\dots,\deg(v)$ in an arbitrary way. In this way, each edge obtains two numbers in $[\Delta]$, and we thus have a coloring of the edges with $O(\Delta^2)$ colors. The coloring is however not proper. For each of the colors, each node can however have at most $2$ incident edges. Therefore, the color classes induce paths and cycles. Those paths and cycles can be properly edge-colored with $3$ colors in $O(\log^* n)$ rounds by using a standard algorithm~\cite{cole86,linial92}.
  
  Given an $O(\Delta^2)$-edge coloring, we next show that a constant factor approximation for the maximum cardinality matching problem can be computed in $O(\log^2\Delta)$ rounds. One can then obtain a maximal matching by repeating $O(\log n)$ times. The matching approximation is computed by using the rounding approach for graphs with bounded neighborhood independence.
  
  As a first step, we need to be able to compute a constant factor approximate feasible solution for LP \eqref{eq:fractionISLP} on the line graph of $G$ (and with all weights being equal to $1$). This can be done by slightly adapting a standard fractional matching approximation algorithm (see, e.g., \cite{fischer2020improved}). We first assign a fractional value $y_e=1/(2\Delta)$ to each edge. We call a node saturated if the sum of the fractional values of its edges is at least $1/4$ and we call an edge frozen if at least one of its nodes is saturated. We now proceed in $O(\log \Delta)$ phases, where in each phase, we double the fractional value of all non-frozen edges. At the end all edges are frozen and it is not hard to see that this gives a $1/4$ approximation for LP \eqref{eq:fractionISLP} (multiplying each fractional edge value with $4$ gives a feasible solution of the dual LP and in the unweighted case, LP \eqref{eq:fractionISLP} and its dual LP both have the same objective function).
  
  We now just have to show that the rounding of this fractional solution to an integral solution can be done in $O(\log^2\Delta)$ rounds. For this, observe that we can use the algorithm for rounding in d2-multigraphs (\Cref{def:d2multigraph}): An edge of the line graph of $G$ is defined by a pair of edges that share a common node and we can thus use this node as the managing node of the edge. We therefore obtain at least the same round complexity as for d2-multigraphs. For the case of rounding independent sets, we have $|\Labels|=2$ and as a consequence, the only expensive step is the computation of the weighted relative average defective coloring. Once the defective coloring is given, doing one rounding step for the edges of a single color can be done in $O(1)$ \CONGEST rounds (cf.\ the proof of \Cref{lemma:d2rounding}). The computation of the defective coloring in \Cref{lemma:d2defective} consists of two parts. First, one computes a weighted $\delta$-relative defective coloring with $O(1/\delta^2)$ colors and one then reduces the number of colors to $O(1/\delta)$ by adapting the algorithm of \cite{BEG18}. The second step requires only $O(1/\delta)$ rounds in d2-multigraphs (cf.\ the proof of \Cref{lemma:d2defective}) and only the first step is expensive and requires the additional $\log\log \Delta$ factor (note that in our case $q=O(\Delta)$). However, on line graphs, a weighted $\delta$-relative defective coloring can be computed fast by using a similar approach as for computing a proper $O(\Delta^2)$-edge coloring at the beginning of the algorithm.
  
  We proceed as follows. Each node $v$ first marks all its edges of weight more than a $\delta/4$-fraction of the total weight of all of $v$'edges. As every node has at most $4/\delta$ such edges, those edges can be properly colored with $O(1/\delta^2)$ colors by using the same algorithm as above. An edge might only be marked from one of its nodes, however, the algorithm also works if an edge only gets a number in $[\Delta]$ from one side. We now proceed with all the edges that are not yet colored and we use a new set of $O(1/\delta^2)$ colors for those. Each node $v$ now labels its edges with numbers between $1$ and $4/\delta$ in such a way that each label is given to edges of total weight at most a $\delta/2$-fraction of the initial total weight of all of $v$'s edges. This is possible by greedily assigning the labels. Now for each combination of two labels, at most at $\delta$-fraction of the total weight of $v$'s edges have this combination of labels (because $v$ must assign one of the two labels). The algorithm can clearly be implemented in $O(\log^*\Delta)$ rounds in the \CONGEST model (with a given initial proper $O(\Delta^2)$-edge coloring.
\end{proof}

\subsection{Approximating a Generalized Caro-Wei-Tur\'{a}n Bound}
\label{sec:turan}

It is well-known that there exist independent sets of size at least $\sum_{v\in V}\frac{1}{\deg(v)+1}$ (sometimes known as the Caro-Wei bound~\cite{Wei1981generalizedturan,griggs1983lower}). This can be lower bounded by $n/\deg_{\mathrm{avg}}$, which is known as the Tur\'{a}n bound~\cite{turan1941external}. 
Moreover, the Caro-Wei independent set  can be computed sequentially in polynomial time by a simple greedy algorithm (see, e.g., \cite{halldorsson1997greed}) and a simple randomized one-round algorithm (one round of the classic Luby algorithm~\cite{luby86}) achieves an independent set of this size in expectation.
For node-weighted graphs, a simple generalization of this greedy algorithm~\cite{kako2005approximation} gives an independent set of weight $\sum_{v\in V}\frac{(w(v))^2}{w(N^+(v))} \geq  \frac{w(V)}{\mathrm{wdeg}_{\mathrm{avg}}+1}$, where $\mathrm{wdeg}_{\mathrm{avg}}= \frac{ \sum_{v\in V} w(N(v))}{w(V)}=\frac{\sum_{v \in V} w(v)\deg(v)}{w(V)}$.\footnote{The inequality $\sum_{v\in V}\frac{(w(v))^2}{w(N^+(v))} \geq  \frac{w(V)}{\mathrm{wdeg}_{\mathrm{avg}}+1}$ can be obtained by a simple application of the Cauchy-Schwarz inequality  $ (\sum_i x^2_i)( \sum_i y^2_i) \geq  (\sum_i x_i y_i )^2$ by assigning $x_v=\sqrt{w(N^+(v))}$ and $y_v=  \frac{w(v)} { \sqrt{w(N^+(v))} }$.} We next show that there is an efficient deterministic \CONGEST algorithm to get within an $\big(\frac{1}{2}-\eps\big)$-factor of those bounds.

\paragraph{Proof of \Cref{thm:ISthm}, Part 3:} The following lemma proves \eqref{eq:ISthm3}, the third claim of \Cref{thm:ISthm}.
	
\begin{lemma}
Let $G=(V,E)$ be a graph with node weights $w:V\to \N$ and let $\vec{x}$ be a vector of fractional node values $x_v= \frac{w(v)}{W_v}$, where $W_v:=w(N^+(v))$. Then, for any $\eps>0$, we can deterministically compute an independent set of total weight $\big(\frac{1}{2}- \eps\big)\cdot \sum_{v\in V}\frac{(w(v))^2}{W_v}$ in time $O\big(\frac{\log^2(\Delta/\eps)}{\eps} + \log^* n\big)$ in the \CONGEST model.
\end{lemma}
	\begin{proof}
		Recall that $\utility(\vec{x}) = \sum_{v\in V} w(v)\cdot x_v$ and $\cost(\vec{x})=\sum_{\set{u,v}\in E} \min\set{w(u), w(v)}\cdot x_u\cdot x_v$. We first show that for the fractional values $x_v$ as given, we have $\utility(\vec{x})\geq 2\cost(\vec{x})$. For this, we first give an upper bound on the cost $\cost(e,\vec{x})=\min\set{w(u), w(v)}\cdot x_u\cdot x_v$ of a single edges $e=\set{u,v}$. We have
		\begin{eqnarray}
			\cost(e,\vec{x}) & = & w(u) \cdot w(v)\cdot\min\set{w(u),w(v)}\cdot
			\frac{1}{W_u\cdot W_v}\nonumber\\
			& \leq &
			w(u)\cdot w(v)\cdot \min\set{w(u),w(v)}\cdot \frac{1}{2}\cdot \left[
			\frac{1}{W_u^2} + \frac{1}{W_v^2}\right]\label{eq:123}\\
			& \leq & 
			\frac{1}{2}\cdot\left[
			\frac{(w(u))^2\cdot w(v)}{W_u^2} + \frac{(w(v))^2\cdot w(u)}{W_v^2}
			\right].\label{eq:edgecostupper}
		\end{eqnarray}
    Inequality \eqref{eq:123} follows because $\big(\frac{1}{W_u}-\frac{1}{W_v}\big)^2=\frac{1}{W_u^2} + \frac{1}{W_v^2} - 
  \frac{2}{W_u\cdot W_v}\geq 0$. We can now use this to obtain an upper bound on the total cost $\cost(\vec{x})$:
  \begin{eqnarray*}
    \cost(\vec{x}) & = & \sum_{e=\set{u,v}\in E} \cost(e,\vec{x})\\
    & \stackrel{\eqref{eq:edgecostupper}}{\leq} &
    \sum_{\set{u,v}\in E} \frac{1}{2}\cdot\left[
     \frac{(w(u))^2\cdot w(v)}{W_u^2} + \frac{(w(v))^2\cdot w(u)}{W_v^2}
    \right]\\
    & = &
    \frac{1}{2}\cdot\sum_{u\in V}\sum_{v\in N(u)} \frac{(w(u))^2\cdot w(v)}{W_u^2}\\
    & = & \frac{1}{2}\cdot\underbrace{\sum_{u\in V}\frac{(w(u))^2}{W_u}}_{=\utility(\vec{x})}\cdot
    \underbrace{\sum_{v\in N(u)}\frac{w(v)}{W_u}}_{\leq 1}\ \leq\ \frac{1}{2}\cdot\utility(\vec{x}).
  \end{eqnarray*}
  The claim of the lemma now follows directly from \Cref{lemma:basicISalgorithm}.
\end{proof}


%% file: setCover.tex
\section{Minimum Set Cover}
\label{sec:setcover}

In this section we use our framework to give a fast deterministic algorithm for the set cover problem. 
In this problem, the input consists of a universe of elements and a family of their subsets. We want to find the smallest subfamily covering all elements. We will use the following bipartite-graph formulation of the problem. 

\begin{definition}
\label{def:set_cover}
In the \emph{set cover} problem, the input consists of a bipartite graph $G$ with $V(G) = U \sqcup V$. 
The goal is to find the smallest possible subfamily $V_{out} \subseteq V$ such that $N(V_{out}) = U$.
\end{definition}

The section is dedicated to the proof of the following theorem. 

\begin{theorem}
\label{thm:set_cover}
Consider an instance of a set cover problem where $\max_{u \in U} \deg(u) \le t$ (i.e., each element $u \in U$ is present in at most $t$ sets) and $\max_{v \in V} \deg(v) \le s$ (i.e., each set $v \in V$ contains at most $s$ elements). There is a distributed $O(\log s)$-approximation algorithm for the problem with round complexity of
\begin{enumerate}
\item $O(\log s \cdot \log ^2 t + \log^* n)$ in the \LOCAL model,
\item $O\left(\log s \ \cdot \left( \log^2 t\ \cdot  \log\log t + \log t \cdot \log^* n\right)\right)$ in the \CONGEST model.
\end{enumerate}
\end{theorem}

The rest of this section is devoted to the proof of \cref{thm:set_cover}. 
The outline of the proof is similar to the MIS algorithm from \cref{sec:MIS}:
We first discuss how this problem can be solved using a randomized algorithm, then we construct a natural pessimistic estimator for the randomized algorithm, and finally we verify that the estimator satisfies the requirements of \cref{lemma:d2rounding}. 

\paragraph{A Randomized Set Cover Algorithm}
We start with an informal discussion of a randomized algorithm that we derandomize to get \cref{thm:set_cover}. 
We rely on the following result of \cite{nearsighted} for the \emph{fractional} variant of the set cover problem. 
There, we ask for a function $x : V \rightarrow [0,1]$ that assigns a fractional value to each set $v \in V$ and we require each element $u \in U$ to satisfy that $\sum_{v \in N(u)} x_v \ge 1$. Note that a solution to the set cover problem directly implies a solution to its fractional variant. 

\begin{theorem}[Kuhn, Moscibroda, and Wattenhofer~\cite{nearsighted}]
There exists a distributed algorithm that computes a $2$-approximate solution to the fractional set cover problem in $O(\log s \cdot \log t)$ rounds of the \congest model. 
\end{theorem}

Consider the following randomized algorithm that has $\tau = O(\log s)$ rounds and in each round it simply constructs a subset $V'_i \subseteq V$ by including each $v \in V'_i$ with probability $x_v$. 
Note that if we define $V' = \bigcup_{i = 1}^{\tau} V'_i$, we have for the expected size of $V'$ that $\E[x_{V'}] = O(\log s) \cdot OPT$. 
On the other hand, we observe that whenever we consider some node $u$ not covered by $V'_1 \cup \dots \cup V'_i$, i.e., $u \not\in N(V'_1 \cup \dots \cup V'_i)$, we have $u \in N(V'_{i+1})$ with constant probability because $\sum_{v \in N(u)} x_v \ge 1$. 
Hence, choosing $\tau = O(\log s)$ large enough, we get that the expected number of uncovered elements is $\E\left[ U \setminus N(V')\right] \le |U|/s$. 

This means that although $V'$ is not expected to be a solution to the set cover problem, we can fix this issue as follows. Observe that for the size of the optimum solution $OPT$ we have 
\begin{align}
\label{eq:opt_lb}
OPT \ge |U| / s    
\end{align}
because every subset of $V$ smaller than $|U|/s$ would cover less than $|U|/s \cdot s =  |U|$ elements. 
This means that we can define $V''$ by letting each uncovered element of $U \setminus N(V')$ choose an arbitrary neighboring node $v \in V$ that we add to $V''$. 
The set $V' \cup V''$ clearly covers all nodes of $U$ and its expected size is at most $|V'| + |U \setminus N(V')| = O(\log s) \cdot OPT + |U|/s = \left( 1 + O(\log s)\right) \cdot OPT$.

\paragraph{The Deterministic Algorithm}

In \cref{alg:set_cover}, we present the derandomized version of the above randomized procedure. 
In the rest of the section, we first explain the algorithm, then verify that we can use \cref{lemma:distributedonG}, prove that the algorithm is $O(\log s)$-approximate, and finally discuss its implementation in the \local and \congest models.

Let us explain the intuition behind \cref{alg:set_cover}. The issue with the derandomization is that we want to optimize two conflicting goals. 
On one hand, we want to cover as many elements of the remaining set $U_i$ as possible. This means that we want to round our fractional solution $x$ to maximize the value of the following pessimistic lower bound for $|U_i| \setminus |U_{i+1}|$ coming from the pairwise analysis:
\[
\sum_{u \in U_i} \sum_{v \in V} \left( x_v - \sum_{v' \not= v \in V} x_{v'} \right) .
\]
On the other hand, we want to minimize the number of selected nodes of $V$, i.e., \[
\sum_{v \in V} x_v.
\]
To optimize these two conflicting expressions, we simply subtract the second one from the first one and invoke \cref{lemma:distributedonG} that can deal with summing terms of different signs. 
An issue with this approach is that to optimize both expressions at once, they should be of the same order of magnitude. The value of $|U_i|$ however gradually decreases from $|U|$ to $|U|/s^{O(1)}$ until we can finish by constructing the set $V''$ as in the randomized algorithm. This is why in \cref{alg:set_cover} we slowly increase the relative weight of the first optimized expression with respect to the second one (see \cref{line:u,line:c}). 
We also add a constant term of $10\sum_{v \in V}x_v$ to the utility function ${\bf u}(\vec{\tx})$ so that our functions satisfy ${\bf u}(\vec{x}) - {\bf c}(\vec{x}) \ge {\bf u}(\vec{x})/2$ which is needed to apply \cref{lemma:distributedonG}. 

\begin{algorithm}[ht]
	\caption{Deterministic Set Cover}
	\label{alg:set_cover}
	Input: A bipartite graph $G$ with $V(G) = U \sqcup V$.\\
	Output: A set $V_{out} \subseteq V$ with $N(V_{out}) = U$
	\begin{algorithmic}[1]
	    \State $\tau \leftarrow O(\log s)$
	    \State $\vec{x}_0: V \rightarrow [0,1]$ is a $2$-approximate fractional set cover of \cite{nearsighted}
	    \State For every $v \in V$, define $x_v \leftarrow x_0(v)/10$ unless $x_0(v) \le 1/(2t)$ in which case $x_v \leftarrow 0$.
	    \State For every $u \in U_i$ define $N^*(u) \subseteq N(u)$ such that $\frac{1}{20} \le \sum_{v \in N^*(u)}x_v \le \frac{1}{5}$ (see \cref{eq:frac2}) \label{line:n*}
	    \State Compute a $2$-hop coloring of $G$ with $\zeta = (st)^2$ colors using Linial's algorithm \cite{linial1987LOCAL}
        \For{$i \leftarrow 1, \dots, \tau$}
	        \State $U_i \leftarrow U \setminus \left( V'_1 \cup \dots \cup V'_{i-1} \right)$
            \State Define ${\bf u}(\vec{\tilde{x}}) \leftarrow  \frac{1}{1.01^{\tau-i}} \cdot \sum_{u \in U_i} \sum_{v \in N^*(u)} \tilde{x}_v + 10\sum_{v \in V} x_v$ \label{line:u}
            \State Define ${\bf c}(\vec{\tilde{x}}) \leftarrow \frac{1}{1.01^{\tau-i}} \cdot \sum_{u \in U_i} \sum_{v \in N^*(u)}\sum_{v'\not=v \in N^*(u)} \tilde{x}_v \tilde{x}_{v'}   
            + \sum_{v \in V} \tilde{x}_v$\label{line:c}
            \State By plugging $\vec{\tilde{x}} = \vec{x}$, use \cref{lemma:distributedonG} with ${\bf u}, {\bf c}, \zeta, \eps = \frac{1}{100}, \mu = \frac12, \lambda_{\min} = \frac{1}{2t}$ to round $\vec{x}$ to $\vec{x_i'}$
            \State Define $V'_i \leftarrow \{v \in V: x_i'(v) = 1\}$
        \EndFor
        \State Define $V''$ by letting each element of $U \setminus  N\left( V'_1 \cup \dots \cup V'_{\tau} \right)$ choose an arbitrary neighbor. 
		\State \Return {$V_{out} = V'_1 \cup \dots \cup V'_\tau \cup V''$}
	\end{algorithmic}
\end{algorithm}

\paragraph{Properties of Fractional Weights}
Let us first discuss the properties of the fractional solution $\vec{x}$ that we repeatedly round inside \cref{alg:set_cover}. By definition, $\vec{x}$ satisfies for every $u \in U$ that
\begin{align}
    \label{eq:frac0}
     x_v \in \{0\} \cup [1/(20t), 1/10],
\end{align}
so in particular $\vec{x}$ is $t$-fractional. Next, we note that for every $u \in U$ we have
\begin{align}
    \label{eq:frac1}
    \sum_{v \in N(u)} x_v \ge 1/20. 
\end{align}
To see this, let $N_{small}(u) \subseteq N(u)$ be the subset of nodes $v \in N(u)$ such that $x_0(v)\le 1/(2t)$. We have $\sum_{v \in N_{small}(u)} x_v \le t \cdot 1/(2t) \le 1/2$ where we used $|N_{small}(u)| \le |N(u)| \le t$. 
Hence $\sum_{v \in N(u)} x_v \ge \sum_{v \in N(u) \setminus N_{small}(u)} x_v \ge (1/2) / 10 = 1/20$. 

In view of \cref{eq:frac1}, in \cref{line:n*} we can define $N^*(u)$ for each $u \in U$ by repeatedly adding neighbors of $u$ to $N^*(u)$ until $\sum_{v \in N^*(u)} x_v \ge 1/20$. 
By \cref{eq:frac0} we then get $\sum_{v \in N^*(u)} x_v \le 1/20 + 1/10 \le 1/5$. Hence, we have  
\begin{align}
    \label{eq:frac2}
    1/20 \le \sum_{v \in N^*(u)} x_v \le 1/5 
\end{align}
as required in \cref{line:n*}.  
Finally, observe that
\begin{align}
    \label{eq:frac3}
    \sum_{v \in V} x_v \le 2\cdot OPT/10
\end{align}
where we use that $x_0$ is $2$-approximation of $OPT$.

\paragraph{Checking Assumptions for \cref{lemma:distributedonG}}

We continue by verifying the assumptions of \cref{lemma:distributedonG}. 
\cref{eq:frac0} implies that $\lambda_{min} = 1/(20t)$ is the smallest nonzero value to round. Next, we need to prove that ${\bf u}(\vec{x}) - {\bf c}(\vec{x}) \ge {\bf u}(\vec{x})/2$. 
To see this, note that for every $u$ we have 
\begin{align}
\label{eq:ass}
\sum_{v \in N^*(u)} \sum_{v' \not= v \in N^*(u)} x_v x_{v'} 
\le \sum_{v \in N^*(u)}\left( x_v \cdot \left( \sum_{v'\in N^*(u)} x_{v'}\right)\right) 
\le \sum_{v \in N^*(u)} x_v/5    
\end{align}
using \cref{eq:frac2}.
That is, the first term of ${\bf c}(\vec{x})$ is dominated by the first term of ${\bf u}(\vec{x})$. Similarly, by definition the second term of ${\bf c}(\vec{x})$ is dominated by the second term in ${\bf u}(\vec{x})$ and we thus have 
\begin{align}
 {\bf c}(\vec{x}) \le {\bf u}(\vec{x})/5.   
\end{align}

\paragraph{Analyzing One Step With Pessimistic Estimators}

The utility and cost functions ${\bf u}(\vec{x})$ and ${\bf c}(\vec{x})$ correspond to a pessimistic estimator for the original randomized procedure.
Namely, consider sampling each $v \in S_i$ with probability $x_v$ that yields a binary vector $\vec{x'}$. For a fixed $u \in U_i$ let $\fE_{u,v}$ be the event that $v$ is the only neighbor of $u$ in $N^*(u)$ that is sampled. By union bound the indicator of the event $R_u(\vec{x})$ that $u \in U_i$ gets covered by a set selected in $\tx'$ can be bounded for any $\vec{x'}$ as 
\begin{align}
\label{eq:domination}
    R_u(\vec{x'}) \ge \sum_{v \in N^*(u)} \left( x'_v  - \sum_{v'\not= v \in N^*(u)} x'_v x'_{v'} \right)
\end{align}

For our rounded weights $x'$, this fact implies that after we plug in the definition of ${\bf u}(\vec{x'}), {\bf c}(\vec{x'})$ we have
\begin{align}
    \label{eq:fingers}
    {\bf u}(\vec{x'}) - {\bf c}(\vec{x'}) - 10\sum_{v \in V} x_v 
    &\le  \frac{1}{1.01^{\tau - i}}\sum_{u \in U_i} R_u(\vec{x'}) - \sum_{v \in V}x'(v)
    =  \frac{1}{1.01^{\tau - i}}\left( |U_i| - |U_{i+1}|\right) - |V'_i|
\end{align}
On the other hand, let us now bound the value of $ {\bf u}(\vec{x'}) - {\bf c}(\vec{x'})$ from the other side by comparing it with the ideal randomized process corresponding to fractional weights $\vec{x}$. For those we have
\begin{align}
\label{eq:above}
    {\bf u}(\vec{x}) - {\bf c}(\vec{x}) - 10\sum_{v \in V}x_v
    &= \frac{1}{1.01^{\tau - i}} \sum_{u \in U_i} \sum_{v \in N^*(u)} \left( x_v  - \sum_{v'\not= v \in N^*(u)} x_v x_{v'} \right)    - 10\sum_{v \in V}x_v\\ 
    &\ge \frac{1}{1.01^{\tau - i}} \cdot |U_{i}|\cdot \frac{4}{5} \cdot \frac{1}{20}   - 2OPT\nonumber
\end{align}
where we first used \cref{eq:ass} to get rid of the quadratic term and then used the lower bound on $\sum_{v \in N(u)} x_u \ge 1/20$ from \cref{eq:frac2}. We also used \cref{eq:frac3}. 
We now use \cref{lemma:distributedonG} to relate $x'$ with $x$. Our choice of $\eps_{\text{L\ref{lemma:distributedonG}}} = 1/100$ implies that
\begin{align}
\label{eq:toes}
    {\bf u}(\vec{x'}) - {\bf c}(\vec{x'}) - 10\sum_{v \in V}x_v
    &\ge \frac{99}{100} \cdot \left( {\bf u}(\vec{x}) - {\bf c}(\vec{x})\right) - 10\sum_{v \in V}x_v\\
    &= \frac{99}{100} \cdot \left( {\bf u}(\vec{x}) - {\bf c}(\vec{x}) - 10\sum_{v \in V}x_v\right) - 0.1\sum_{v \in V}x_v\nonumber\\
    &\ge 0.02 \cdot \frac{1}{1.01^{\tau - i}} |U_i|  - 3OPT\nonumber
\end{align}
where the last inequality follows from  \cref{eq:above} together with the fact that $0.1 \sum_{v \in V} x_v < OPT$ by \cref{eq:frac3}. 
By comparing \cref{eq:fingers} with \cref{eq:toes}, we conclude that 
\begin{align}
\label{eq:pot_dec}
 \frac{1}{1.01^{\tau - i}}\left( |U_i| - |U_{i+1}|\right) - |V'_i|
 &\ge 0.02\cdot \frac{1}{1.01^{\tau - i}} |U_i|  - 3OPT
\end{align}

\paragraph{Finishing the analysis}

To analyze the progress of the algorithm, we define the following potential function $\Phi$ in every step. 

\begin{definition}
For every $0 \le i \le \tau = O(\log s)$, we define 
\[
\Phi_i = \frac{1}{1.01^{\tau-i}} \cdot |U_i| + |V'_1| + |V'_2| + \dots + |V'_{i-1}| + 3(\tau-i)\cdot OPT
\]
\end{definition}

At the very beginning before the first step of the algorithm, we have $|U_0| = |U|$ and $|S_0| = 0$, hence $\Phi_0 = \frac{1}{1.01^{\tau}} \cdot |U| + 2\tau \cdot OPT$. Choosing $\tau = O(\log s)$ large enough, we have $\Phi_0 \le |U|/s + 2\tau \cdot OPT \le 3\tau \cdot OPT$ using \cref{eq:opt_lb}. 

We will now prove that the potential $\Phi$ is monotone, that is $\Phi_{i+1} \le \Phi_{i}$ for every $1 \le i \le \tau$. 
Let us first see why the proof of monotonicity of $\Phi$ also finishes the analysis. 
It implies that $\Phi_\tau \le 3\tau \cdot OPT$. Using the definition of $\Phi$, this means that $|U_\tau| + |V'_1| + \dots + |V'_\tau| \le 3\tau \cdot OPT$. 
Since $|U_\tau| = |V''|$, we conclude that $|V_{out}| \le 3\tau\cdot OPT = O(\log s) \cdot OPT$ as needed.  

To see that $\Phi$ is monotone, we write the difference of two consecutive potentials as
\begin{align}
    \Phi_{i} - \Phi_{i+1}
    &= \frac{1}{1.01^{\tau - i}} \left( |U_{i}| -  1.01|U_{i+1}| \right) - |V'_{i}| + 3OPT\\
    & = \frac{1}{1.01^{\tau - i}} \left( |U_{i}| - |U_{i+1}| \right) - |V'_{i}| + 3OPT - 0.01 |U_{i+1}|/1.01^{\tau - i}\\
\end{align}
Applying \cref{eq:pot_dec}, we conclude that
\begin{align}
    \Phi_{i} - \Phi_{i+1}
    &\ge 0.02 \cdot \frac{1}{1.01^{\tau - i}} |U_i|  - 3OPT + 3OPT - 0.01 |U_{i+1}|/1.01^{\tau - i}
    \ge 0
 \end{align}
and we are done. 

\paragraph{Implementation and Round Complexity}

We begin by discussing the implementation in the \local model. 
The algorithm of \cite{nearsighted} needs $O(\log s \cdot \log t)$ rounds and we construct the 2-hop coloring of $G$ with  $\zeta = (st)^2$  colors in $O(\log^* n)$ rounds using Linial's $\Delta^2$-coloring algorithm \cite{linial1987LOCAL}. 
The complexity of every subsequent step is dominated by the call to \cref{lemma:distributedonG} with round complexity $O(\log^2 (1/\lambda_{\min}) + \log^*(\zeta)) = O(\log^2 t + \log^* (st))$. 
We note that if the first term in this expression does not dominate the second one, we have in particular $t < \log s$. In that case, we can however achieve $O(\log s)$-approximation by simply taking all nodes $v \in V$ with $x_0(v) \ge 1/t$ to $V_{out}$. This solution clearly covers all elements of $U$ since every $u \in U$ has necessarily a neighbor with $x_0(v) \ge 1/|N(u)| \ge 1/t$. Moreover, its weight is at most $t$-times larger than the fractional weight $x$. Hence, we may assume that $O(\log^2 t + \log^*(st)) = O(\log^2 t)$ and we are getting an algorithm with round complexity $O(\log s \cdot \log^2 t + \log^* n)$ as needed. 

We continue by discussing the implementation in the \congest model. There, we simply replace all calls of \cref{lemma:distributedonG} by calls to \cref{lemma:d2rounding}. 
Another change that we do is that we do not construct the $2$-hop coloring with $\zeta = (st)^2$ colors at the beginning, but simply use the unique identifiers instead, i.e., we have $\zeta = n^{O(1)}$. 
We are using the $d2$-multigraph $H$ where each node $u \in U$ simulates a virtual edge between every pair of its neighbors $v, v'$. 
The round complexity of one round of the algorithm becomes $O\left(\log s \ \cdot \left( \log^2 t\ \cdot  \log\log t + \log t \cdot \log^* n\right)\right)$.

\begin{remark}
We believe that one can directly generalize \cref{thm:set_cover} to give an $O(\log (sW))$-approximation for the more general \emph{min cost set cover} problem. There, each node $v \in V$ comes with a cost $w(v)$ such that $1 \le  w(v) \le W$. The goal is to find a subset $V_{out} \subseteq V$ of smallest total cost such that $N(V_{out}) = U$. 

We discuss the changes that need to be done in \cref{alg:set_cover}. We use a min-cost version of the fractional algorithm of \cite{nearsighted} that needs $O(\log (sW) \cdot \log t)$ rounds. Then, on \cref{line:u,line:c} in \cref{alg:set_cover} we change the utility $10\sum_{v \in V} x_v$ to $10\sum_{v \in V} w(v) \cdot x_v$ and the cost $\sum_{v \in V)} \tx_v$ to $\sum_{v \in V)} w(v) \cdot \tx_v$. 
Finally, we run the algorithm for $\tau = O(\log (sW))$ steps. This way, we have $|U \setminus N(V_1 \cup \dots \cup V_\tau)| \le |U|/(sW)$ which implies that $w(V'') \le OPT$. 
The distributed complexity of this algorithm is
\begin{enumerate}
    \item $O(\log (sW) \cdot \log ^2 t + \log^* n)$ in the \LOCAL model,
\item $O\left(\log (sW)\ \cdot \left( \log^2 t\ \cdot  \log\log t + \log t \cdot \log^* n\right)\right)$ in the \CONGEST model, assuming $W \le n^{O(1)}$. 

\end{enumerate}
\end{remark}